\documentclass[dvipsnames]{CSML}
\pdfoutput=1

\usepackage{lastpage}
\newcommand{\dOi}{13(4:2)2017}
\lmcsheading{}{1--\pageref{LastPage}}{}{}%
{Mar.~30,~2017}{Oct.~25,~2017}{}


\ACMCCS{[{\bf Security and privacy~Logic and verification}]: 300; [{\bf Software and its engineering~General programming languages}]: 300; [{\bf Software and its engineering~Compilers}]: 300}

\usepackage[inference]{semantic}
\usepackage{amsmath,amsthm} 
\usepackage{balance}
\usepackage{lastpage}
\usepackage{mathpartir}
\usepackage{natbib}
\usepackage{amssymb}
\usepackage[utf8]{inputenc}
\usepackage{stmaryrd} 
\usepackage{xspace}
\usepackage{latexsym}
\usepackage{hyperref}
\hypersetup{hidelinks}
\usepackage{ifthen}
\usepackage{mathtools}
\usepackage{color}
\usepackage{listings}
\usepackage{verbatim}
\usepackage[colorinlistoftodos]{todonotes}
\usepackage{tikz}
\usetikzlibrary{positioning,shadows,arrows,calc,backgrounds,fit,shapes,shapes.multipart,decorations.pathreplacing,shapes.misc,patterns}
\usepackage{tikzscale}
\usepackage{xspace}
\usepackage[T1]{fontenc}
\usepackage[scaled=.83]{beramono}
\usepackage{epigraph}
\usepackage[capitalize]{cleveref}
\usepackage{booktabs}
\usepackage{float}
\usepackage{etoolbox}
\usepackage{wrapfig}
\usepackage{nameref}


\newcommand{\mi}[1]{\ensuremath{\mathit{#1}}}
\newcommand{\mr}[1]{\ensuremath{\mathrm{#1}}}

\newcommand{\mtt}[1]{\ensuremath{\mathtt{#1}}}
\newcommand{\mf}[1]{\ensuremath{\mathbf{#1}}}
\newcommand{\mk}[1]{\ensuremath{\mathfrak{#1}}}
\newcommand{\mc}[1]{\ensuremath{\mathcal{#1}}}
\newcommand{\ms}[1]{\ensuremath{\mathsf{#1}}}
\newcommand{\mb}[1]{\ensuremath{\mathbb{#1}}}

\newcommand{\isdef}{\ensuremath{\mathrel{\overset{\makebox[0pt]{\mbox{\normalfont\tiny\sffamily def}}}{=}}}}

\newcommand{\relmiddle}[1]{\mathrel{}\middle#1\mathrel{}}

\DeclareMathOperator\mydefsym{\ensuremath{\triangleq}}
\newcommand\bnfdef{\ensuremath{\mathrel{::=}}}

\newcommand{\llb}{\llbracket}
\newcommand{\rrb}{\rrbracket}

\newcommand{\ra}{\rightarrow}
\newcommand{\Ra}{\Rightarrow}
\newcommand{\myset}[2]{\ensuremath{\{#1 ~|~ #2\}}}

\newcommand{\Da}[1]{\ensuremath{\Downarrow^{#1}}}
\newcommand{\Dan}[2]{\ensuremath{\Downarrow^{#1}\indexx{#2}}}

\newcommand{\compskel}[3]{{\bl{\left\llbracket \tl{#1} \right\rrbracket}^{\bl{#2}}_{\bl{#3}}}}
\newcommand{\comp}[1]{\compskel{\tl{#1}}{}{}}
\newcommand{\compgen}[1]{\compskel{#1}{\tl{\mc{S}}}{\ul{\mc{T}}}}
\newcommand{\compsu}[1]{\compskel{#1}{\stlc}{\ulc}}

\newcommand{\funname}[1]{\mtt{#1}}
\newcommand{\fun}[2]{\ensuremath{{\bl{\funname{#1}(}#2{\bl{)}}}}\xspace}
\newcommand{\dom}[1]{\fun{dom}{#1}}

\newcommand{\erase}[1]{\fun{erase}{\tl{#1}}}
\newcommand{\erasen}{\funname{erase}\xspace}
\newcommand{\prot}[1]{\ensuremath{\funname{\ul{protect}_{\tl{#1}}}}}
\newcommand{\conf}[1]{\ensuremath{\funname{\ul{confine}_{\tl{#1}}}}}

\newcommand{\protd}[1]{\ensuremath{\funname{\ul{protect\mbox{-}alt}_{\tl{#1}}}}}
\newcommand{\confd}[1]{\ensuremath{\funname{\ul{confine\mbox{-}alt}_{\tl{#1}}}}}

\newcommand{\ulcname}[0]{\ensuremath{\mbox{untyped } \lambda\mbox{-calculus}}\xspace}

\newcommand{\ulc}[0]{\ensuremath{\bl{\lambda}^{\ul{\ms{u}}}}\xspace}
\newcommand{\stlc}[0]{\ensuremath{\bl{\lambda}^{\tl{\tau}}}\xspace}
\newcommand{\stlcmu}[0]{\ensuremath{\bl{\lambda}^{\tl{\tau};\tl{\mu}}}\xspace}

\newcommand{\lseal}[0]{\ensuremath{\bl{\lambda}^{seal}}\xspace}

\newcommand{\pstlc}[0]{\ensuremath{\bl{\lambda}^{\tltau}}\xspace}
\newcommand{\lc}[0]{\ensuremath{\lambda}-calculus\xspace}

\newcommand{\stlccol}[0]{NavyBlue}
\newcommand{\ulccol}[0]{WildStrawberry}
\newcommand{\neutcol}[0]{black}

\newcommand{\col}[2]{{\color{#1}{#2}}}
\newcommand{\tl}[1]{\col{\stlccol}{\mf{#1}}}				
\newcommand{\ul}[1]{\col{\ulccol }{\ms{#1}}}				

\newcommand{\bl}[1]{\ensuremath{{\col{\neutcol}{#1}}}}

\newcommand{\batype}[0]{\mc{B}\xspace}

\newcommand{\wrong}[0]{\ul{wrong}}

\newcommand{\unitv}[0]{\mtt{unit}\xspace}
\newcommand{\truev}[0]{\mtt{true}\xspace}
\newcommand{\falsev}[0]{\mtt{false}\xspace}
\newcommand{\Bool}[0]{\mtt{Bool}\xspace}
\newcommand{\Unit}[0]{\mtt{Unit}\xspace}

\newcommand{\blto}[0]{\,\bl{\to}\,}

\newcommand{\lam}[2]{\lambda #1\ldotp #2}

\newcommand{\pair}[1]{\langle#1\rangle}
\newcommand{\projone}[1]{\ensuremath{#1.1}}
\newcommand{\projtwo}[1]{\ensuremath{#1.2}}

\newcommand{\fix}[1]{{\mr{fix}}_{#1}}
\newcommand{\case}{\ensuremath{\mr{case}}}
\newcommand{\of}{\ensuremath{\mr{of}}}
\newcommand{\unit}{\ensuremath{\mr{unit}}}

\newcommand{\ifte}[3]{\mr{if}~#1~\mr{then}~#2~\mr{else}~#3}
\newcommand{\casE}[2]{\ensuremath{\mr{case}~#1~\mr{of}~#2}}
\newcommand{\caseof}[3]{\ensuremath{\mr{case}~#1~\mr{of}~\inl{x_1}\mapsto #2\mid\inr{x_2}\mapsto #3}}
\newcommand{\inl}[1]{\ensuremath{\mr{inl}~#1}}
\newcommand{\inr}[1]{\ensuremath{\mr{inr}~#1}}

\newcommand{\tlGamma}[0]{\ensuremath{\tl{\Gamma}}}
\newcommand{\ulGamma}[0]{\ensuremath{\ul{\Gamma}}}			
\newcommand{\tlgamma}[0]{\ensuremath{\tl{\gamma}}}
\newcommand{\ulgamma}[0]{\ensuremath{\ul{\gamma}}}

\newcommand{\tlt}[0]{\ensuremath{\tl{t}}\xspace}
\newcommand{\ult}[0]{\ensuremath{\ul{t}}\xspace}
\newcommand{\tlv}[0]{\ensuremath{\tl{v}}\xspace}
\newcommand{\ulv}[0]{\ensuremath{\ul{v}}\xspace}

\newcommand{\tle}[0]{\ensuremath{\tl{\emptyset}}\xspace}
\newcommand{\ule}[0]{\ensuremath{\ul{\emptyset}}\xspace}
\newcommand{\tltau}[0]{\ensuremath{\tl{\tau}}\xspace}

\newcommand{\tlhat}[1]{\ensuremath{\tl{\hat{#1}}}}

\newcommand{\ulfix}[0]{\ul{\mi{fix}}}
\newcommand{\tlfix}[0]{\tl{\fix{}}}

\newcommand{\red}[0]{\ensuremath{\hookrightarrow}}
\DeclareMathOperator\stlcto{\ensuremath{\tl{\red}}}
\DeclareMathOperator\ulcto{\ensuremath{\ul{\red}}}

\newcommand{\ctx}[0]{\ensuremath{\mb{E}}}
\newcommand{\tlC}[0]{\ensuremath{\tl{\ctx}}}
\newcommand{\tlH}[1]{\ensuremath{\tl{[#1]}}}
\newcommand{\ulC}[0]{\ensuremath{\ul{\ctx}}}
\newcommand{\ulH}[1]{\ensuremath{\ul{[#1]}}}

\newcommand{\progctx}{\mk{C}}
\newcommand{\upc}[0]{\ul{\progctx}}
\newcommand{\tpc}[0]{\tl{\progctx}}

\newcommand{\subst}[2]{\ensuremath{\bl{[}#1\bl{/}#2\bl{]}}}	
\newcommand{\tlsub}[2]{\subst{\tl{#1}}{\tl{#2}}}
\newcommand{\ulsub}[2]{\subst{\ul{#1}}{\ul{#2}}}

\newcommand{\langsp}[1]{\ensuremath{\ms{#1}}\xspace}

\newcommand{\langspfun}[2]{\ensuremath{\langsp{#1}(#2)}}
\newcommand{\tllsfun}[2]{\langspfun{\tl{#1}}{\tl{#2}}}
\newcommand{\ullsfun}[2]{\langspfun{\ul{#1}}{\ul{#2}}}

\newcommand{\stepsfungen}[2]{\ensuremath{\langspfun{lev^{#1}}{#2}}}
\newcommand{\latergen}[1]{\ensuremath{\triangleright^{#1}}}

\newcommand{\obswfungen}[2]{\langspfun{O^{#1}}{#2}}
\newcommand{\futwgen}[1]{\ensuremath{\sqsupseteq^{#1}}}

\newcommand{\strfutwgen}[1]{\ensuremath{\sqsupset^{#1}}}

\newcommand{\stepsfun}[1]{\stepsfungen{}{#1}}
\DeclareMathOperator\later{\latergen{}}

\newcommand{\obswfun}[1]{\obswfungen{}{#1}}
\DeclareMathOperator\futw{\futwgen{}}

\DeclareMathOperator\strfutw{\strfutwgen{}}

\newcommand{\WW}[0]{\ms{\underline{W}}\xspace}

\newcommand{\CWstepsfun}[1]{\stepsfungen{\CW}{#1}}
\DeclareMathOperator\CWlater{\latergen{\CW}}

\newcommand{\CWobswfun}[1]{\obswfungen{\CW}{#1}}
\DeclareMathOperator\CWfutw{\futwgen{\CW}}
\DeclareMathOperator\CWfutwpub{\futwpubgen{\CW}}

\newcommand{\logrelgen}[1]{\ensuremath{\operatorname{\approx}^{#1}}}
\DeclareMathOperator\logrel{\logrelgen{}}

\newcommand{\genlogrel}[0]{\ensuremath{\square}} 

\newcommand{\genrel}[3]{\ensuremath{\mc{#1}\llb\tl{#2}\rrb^{#3}}}
\newcommand{\valrel}[1]{\genrel{V}{#1}{\tlrho}}

\newcommand{\contrel}[1]{\genrel{K}{#1}{\tlrho}}
\newcommand{\termrel}[1]{\genrel{E}{#1}{\tlrho}}
\newcommand{\envrel}[1]{\genrel{G}{#1}{\tlrho}}

\newcommand{\arbsim}{\ensuremath{\mathrel{\square}}}

\newcommand{\UVal}{\ensuremath{\tl{\mr{UVal}}}\xspace}

\newcommand{\myomega}{\ensuremath{\tl{\mr{omega}}}}
\newcommand{\downgrade}{\ensuremath{\tl{\mr{downgrade}}}}
\newcommand{\upgrade}{\ensuremath{\tl{\mr{upgrade}}}}
\newcommand{\emulate}{\ensuremath{\tl{\mr{emulate}}}}

\newcommand{\ef}[1]{\ensuremath{\tl{\langle\!\langle}\ul{#1}\tl{\rangle\!\rangle}}}

\newcommand{\inDV}[1]{\funname{in_{\ms{#1}}}}
\newcommand{\unkUVal}{\funname{unkUVal}}
\newcommand{\caseDV}[1]{\funname{case_{#1}}}

\newcommand{\EmulDV}{\ensuremath{\tl{\mtt{EmulDV}}}\xspace}

\newcommand{\replemul}[1]{\ensuremath{\fun{repEmul}{\tl{#1}}}}
\newcommand{\toemul}[2]{\ensuremath{\fun{toEmul}{{#1}}\indexx{#2}}}

\newcommand{\precise}{\ensuremath{\mtt{precise}}}
\newcommand{\imprecise}{\ensuremath{\mtt{imprecise}}}

\newcommand{\inject}[1]{\funname{inject\indexx{{#1}}}}
\newcommand{\extractf}[1]{\funname{extract\indexx{{#1}}}}

\newcommand{\indexx}[1]{\ensuremath{_{\ms{#1}}}}
\newcommand{\nn}[0]{\indexx{n}}
\newcommand{\np}[0]{\indexx{n;p}}
\newcommand{\tlnn}[0]{\tl{\indexx{n}}}

\newcounter{typerule}
\crefname{typerule}{rule}{rules}

\newcommand{\typeruleInt}[5]{
	\def\thetyperule{#1}%
	\refstepcounter{typerule}%
	\label{tr:#4}%
  \ensuremath{\begin{array}{c}#5 \inference{#2}{#3}\end{array}} 
}

\newcommand{\typerulenolabel}[4]{
  \typeruleInt{#1}{#2}{#3}{#4}{}
}

\makeatletter
\newcommand{\typeruleAlt}[3]{
    \expandafter\gdef\csname trname@#1\endcsname{#2}%
    \hyperdef{rule}{#1}{#3}%
    }
\newcommand{\reftyperuleAlt}[1]{%
  \hyperref{}{rule}{#1}{\csname trname@#1\endcsname}%
  }
\makeatother

\pgfdeclarelayer{background}
\pgfdeclarelayer{veryback}
\pgfdeclarelayer{veryback2}
\pgfdeclarelayer{veryback3}
\pgfdeclarelayer{back2}
\pgfdeclarelayer{foreground}
\pgfsetlayers{veryback3,veryback2,veryback,background,back2,main,foreground}

\def\botrule{\vspace{-1mm}\hrule\vspace{1mm}}

\DeclareMathOperator\nsimeq{\ensuremath{\mathrel{\not\simeq}}}
\DeclareMathOperator\nequiv{\ensuremath{\equiv\!\!\!\!\!/\ }}

\DeclareMathOperator\ceq{\ensuremath{\mathrel{\simeq_{\mi{ctx}}}}}
\DeclareMathOperator\nceq{\mathrel{\nsimeq_{\mi{ctx}}}}

\DeclareMathOperator\ceqstlc{\tl{\ceq}}

\DeclareMathOperator\cequlc{\ul{\ceq}}

\def\teqaux#1{\vcenter{\hbox{\ooalign{\hfil
       \raise6pt \hbox{\scriptsize{T}}\hfil\cr\hfil
       $=$}}}}

\def\ceqwaux#1{\vcenter{\hbox{\ooalign{\hfil
       \raise6pt \hbox{\scriptsize{w}}\hfil\cr\hfil
       $\ceq$}}}}

\crefname{lem}{Lemma}{Lemmas}
\crefname{thm}{Theorem}{Theorems}
\crefname{defi}{Definition}{Definitions}
\crefname{exa}{Example}{Examples}

\makeatletter
\def\renewtheorem#1{%
  \expandafter\let\csname#1\endcsname\relax
  \expandafter\let\csname c@#1\endcsname\relax
  \gdef\renewtheorem@envname{#1}
  \renewtheorem@secpar
}
\def\renewtheorem@secpar{\@ifnextchar[{\renewtheorem@numberedlike}{\renewtheorem@nonumberedlike}}
\def\renewtheorem@numberedlike[#1]#2{\newtheorem{\renewtheorem@envname}[#1]{#2}}
\def\renewtheorem@nonumberedlike#1{  
\def\renewtheorem@caption{#1}
\edef\renewtheorem@nowithin{\noexpand\newtheorem{\renewtheorem@envname}{\renewtheorem@caption}}
\renewtheorem@thirdpar
}
\def\renewtheorem@thirdpar{\@ifnextchar[{\renewtheorem@within}{\renewtheorem@nowithin}}
\def\renewtheorem@within[#1]{\renewtheorem@nowithin[#1]}
\makeatother
\theoremstyle{plain}\renewtheorem{lem}[thm]{Lemma}
\theoremstyle{definition}\renewtheorem{exa}[thm]{Example}
\theoremstyle{definition}\renewtheorem{defi}[thm]{Definition}

\renewcommand{\CWfutwpub}{\futw}
\renewcommand{\CWfutw}{\futw}
\renewcommand{\CWstepsfun}{\stepsfun}
\renewcommand{\CWlater}{\later}
\renewcommand{\CWobswfun}{\obswfun}

\renewcommand{\valrel}[1]{\genrel{V}{#1}{}}
\renewcommand{\contrel}[1]{\genrel{K}{#1}{}}
\renewcommand{\termrel}[1]{\genrel{E}{#1}{}}
\renewcommand{\envrel}[1]{\genrel{G}{#1}{}}

\renewcommand{\unkUVal}{\mr{unk}}

\renewcommand{\mydefsym}{\isdef}

\begin{document}

\title[Modular, Fully-Abstract Compilation by Approximate Back-Translation]{Modular, Fully-Abstract Compilation \\by Approximate Back-Translation\rsuper*}
\titlecomment{{\lsuper*}extended version of the paper in POPL'16.}

\author[D.~Devriese]{Dominique Devriese\rsuper{a}}	
\address{\lsuper{a}imec-DistriNet, KU Leuven, Belgium}	
\email{first.last@cs.kuleuven.be}  

\author[M.~Patrignani]{Marco Patrignani\rsuper{b}$^{,\dagger}$}	
\address{\lsuper{b}MPI-SWS, Saarbr\"ucken, Germany}	
\email{\lsuper{b}first.last@mpi-sws.org}  
\thanks{$\lsuper{\dagger}$Currently at CISPA} 

\address{\lsuper{c}UGent, Belgium}	
\email{first.last@ugent.be}  

\author[F.~Piessens]{Frank Piessens\rsuper{a}}	
\author[S.~Keuchel]{Steven Keuchel\rsuper{c}}	

\begin{abstract}
A compiler is {\em fully-abstract} if the compilation from source language programs to target language programs reflects and preserves behavioural equivalence.
Such compilers have important security benefits, as they limit the power of an attacker interacting with the program in the target language to that of an attacker interacting with the program in the source language.
Proving compiler full-abstraction is, however, rather complicated.
A common proof technique is based on the \emph{back-translation} of target-level program contexts to behaviourally-equivalent source-level contexts.
However, constructing such a back-translation is problematic when the source language is not strong enough to embed an encoding of the target language.
For instance, when compiling from a simply-typed \lc (\pstlc) to an untyped \lc (\ulc), the lack of recursive types in \pstlc prevents such a back-translation.

We propose a general and elegant solution for this problem.
The key insight is that it suffices to construct an \emph{approximate} back-translation. The approximation is only accurate up to a certain number of steps and conservative beyond that, in the sense that the context generated by the back-translation may diverge when the original would not, but not vice versa.
Based on this insight, we describe a general technique for proving compiler full-abstraction and demonstrate it on a compiler from \pstlc to \ulc.
The proof uses asymmetric cross-language logical relations and makes innovative use of step-indexing to express the relation between a context and its approximate back-translation.
The proof extends easily to common compiler patterns such as modular compilation and, to the best of our knowledge, it is the first compiler full abstraction proof to have been fully mechanised in Coq.
We believe this proof technique can scale to challenging settings and enable simpler, more scalable proofs of compiler full-abstraction.
\end{abstract}

\maketitle

\smallskip

We typeset source and target language terms in \tl{blue} resp.\ \ul{pink}; we recommend to view/print this paper in colour for maximum clarity.

\section{Introduction}
\label{sec:introduction}

A compiler is {\em fully-abstract} if the compilation from source language programs to target language programs preserves and reflects behavioural equivalence~\citep{abadiFa,faEHM}.
Such compilers have important security benefits. It is often realistic to assume that attackers can interact with a program in the target language, and depending on the target language this can enable attacks such as  improper stack manipulation, breaking control flow guarantees,
reading from or writing to private memory of other components,
inspecting or modifying the implementation of a function etc.~\citep{abadiFa,Kennedy,scoo-j,protOnLayRand,fstar2js,Agten2012SecComp}.
A fully-abstract compiler is sufficiently defensive to rule out such attacks: the power of an attacker interacting with the program in the target language is limited to attacks that could also be performed by an attacker interacting with the program in the source language.

Formally, we model a compiler as a function $\comp{\cdot}$ that maps source language terms $\tl{t}$ to target language terms $\comp{\tl{t}}$.
Elements of the source language are typeset in a \tl{blue}, \tl{bold} font, while elements of the target language are typeset in a \ul{pink}, \ul{sans\mbox{-}serif} font.
Roughly, the compiler is fully-abstract, if for any
two source language terms $\tl{t_1}$ and $\tl{t_2}$, we have that they
are behaviourally equivalent ($\tl{t_1 \ceqstlc t_2}$) if and only if their compiled counterparts
are behaviourally equivalent ($\comp{\tl{t_1}} \cequlc\comp{\tl{t_2}}$)~\citep{abadiFa}.
The notion of behavioural equivalence used here is the canonical notion of contextual
equivalence: two terms are equivalent if they behave the same when
plugged into any valid context.
Specifically, we take contextual equivalence to be equi-termination: $t\ceq t' \mydefsym \forall\progctx,\progctx[t]\Da{} \iff \progctx[t']\Da{}$.
The universal quantification over contexts \progctx\ ensures that the results produced by $t$ and $t'$ are the same~\citep{lcfConsidered,definFA}.

The full-abstraction property can be split into two parts: the right-to-left implication and the
left-to-right implication, which we call
(contextual) equivalence \emph{reflection} and \emph{preservation} respectively.

\subsubsection*{Equivalence reflection}
($\tl{t_1 \ceqstlc t_2} \Leftarrow \comp{\tl{t_1}}
\cequlc\comp{\tl{t_2}}$)
requires that if the compiler produces equivalent target programs,
then the source programs must have been equivalent. In other words, non-equivalent source programs must be compiled to non-equivalent target programs. Intuitively, this
property captures an aspect of compiler correctness: if programs with
different source language behaviour become equivalent after
compilation, the compiler must have incorrectly compiled at least one
of them (this is also called adequacy of the translation by~\citet{SchmidtSchauss201598}).

\begin{figure}
  \centering
  \includegraphics{proving-compiler-correctness.tikz}
  \caption{Proving one half of full-abstraction: compiler correctness. Only one direction of this half is presented ($\Rightarrow$), the other one follows by symmetry.}
  \label{fig:proving-compiler-correctness}
\end{figure}
We build on cross-language logical relations: a technique that has been proposed for proving compiler correctness~\citep{Hur:2011:KLR:1926385.1926402,bistcc,realizability}.
The general idea of this approach is depicted in
\cref{fig:proving-compiler-correctness} (purposely ignoring
language-specific things such as the types of the terms involved). The
proof starts from the knowledge that
$\comp{\tl{t_1}}\cequlc\comp{\tl{t_2}}$ and sets out to prove that $\tl{t_1 \ceqstlc t_2}$.
That is, for an arbitrary valid context $\tpc$, it shows that
$\tl{\progctx\tlH{t_1} \Da{}}$ if and only if
$\tl{\progctx\tlH{t_2}\Da{}}$. By symmetry, it suffices to show that
$\tl{\progctx\tlH{t_1} \Da{}} \Ra\tl{\progctx\tlH{t_2}\Da{}}$.

The idea of the approach is to define a cross-language logical relation
$ \tl{t} \logrel \ul{t}$ that expresses when a compiled term
$\ul{t}$ behaves as a target-level version of source-level term
$\tl{t}$.
This logical relation is not compiler-specific: it should be understood as a specification of a target-level calling convention rather than precise representation choices for a specific compiler.
If we can then prove that any term is logically related to
its compilation ($ \tl{t} \logrel \comp{\tl{t}}$), and that the
same result holds for contexts ($ \tpc \logrel \comp{\tpc}$), then equivalence reflection follows.%
\footnote{
    As contexts $\tpc$ are also programs, they can be compiled with the same compiler for terms.
}
Starting from $ \tl{t_1} \logrel \comp{\tl{t_1}}$ and
$ \tl{t_2} \logrel \comp{\tl{t_2}}$ and
$ \tpc \logrel \comp{\tpc}$, the proof uses the
inherent compositionality of logical relations to know
$ \tl{\progctx\tlH{t_1}} \logrel \comp{\tpc}\ulH{\bl{\comp{\tl{t_1}}}}$
and the same for $\tl{t_2}$.
If the logical relations are constructed adequately, then related terms necessarily equi-terminate.
Thus, $\tl{\progctx\tlH{t_1} \Da{}}$ iff $\comp{\tpc}\ulH{\bl{\comp{\tl{t_1}}}} \ul{\Da{}}$ and similarly for $\tl{t_2}$.
In particular, this yields the implications (1) and (3) in \cref{fig:proving-compiler-correctness}.
Since implication (2) follows directly from the hypothesis of (contextual) equivalence
for $\comp{\tl{t_1}}$ and $\comp{\tl{t_2}}$, the proof for equivalence reflection is finished.

This direction of compiler full-abstraction is often called \emph{compiler correctness}, as a compiler that is correct trivially has this property.
However, it is important to note that full abstraction alone does not yield correctness as it only talks about equivalence classes and not about respecting a cross-language relation that encodes the meaning of compiled code.
For example, a fully-abstract compiler can swap how \tl{\truev} and \tl{\falsev} are compiled (and how compiled boolean operators use booleans) and still be fully abstract.
However, this intuitively violates correctness, that intuitively tells that \tl{\truev} is compiled to \ul{\truev}.
This is why often, fully-abstract compilers are also proven to be correct (as we also do), to both get one half of full abstraction and be sure to respect the intended meaning of compiled programs.

\subsubsection*{Equivalence preservation}
($\tl{t_1 \ceqstlc t_2}\Ra \comp{\tl{t_1}} \cequlc\comp{\tl{t_2}}$)
requires that equivalent programs remain equivalent after compilation.
This means that no matter what target-level manipulations are done on
compiled programs, the programs must behave equivalently if the source
programs were equivalent. This precludes all sorts
of target-level attacks that break source-level guarantees.

\begin{figure}
  \centering
  \includegraphics{proving-compiler-security-exact.tikz}
  \caption{Proving the other half of full-abstraction: compiler security.}
  \label{fig:proving-compiler-security-exact}
\end{figure}
If the source language is strong enough, it is possible to apply a
strategy analogous to proving equivalence reflection for
proving preservation, as depicted in
\cref{fig:proving-compiler-security-exact}.\footnote{Actually, both
  \cref{fig:proving-compiler-security-exact,fig:proving-compiler-security-approx}
  are simplifications. Perceptive readers may notice that the proof
  depicted here would falsely imply equivalence preservation
  for \emph{any} correct compiler. We correct the simplifications in
  \cref{sec:approx-final}.} Given an arbitrary target-level context
$\upc$, we need to prove that $\upc\ulH{\comp{\tl{t_1}}}\ul{\Da{}}$
implies $\upc\ulH{\comp{\tl{t_2}}}\ul{\Da{}}$. In a
sufficiently-powerful source language, we can construct a
\emph{back-translation} $\ef{ \upc }$ for any target-level context
$\upc$. Using the same logical relation as above, it then suffices to
prove that $\ef{\upc}$ is a valid source-level context and that
$ \ef{\upc} \logrel \upc$ for any valid context $\upc$. Together with
$\tl{t_1}\logrel \comp{\tl{t_1}}$, and similarly for $\tl{t_2}$,
compositionality and adequacy of the logical relation then yield
implications (1) and (3) in the Figure. The remaining implication (2)
follows from the assumed (contextual) equivalence of $\tl{t_1}$ and
$\tl{t_2}$.

Constructing a back-translation of contexts is not easy, but it can be done if the source language is sufficiently expressive.
Consider, for example, a compiler that translates terms from
a simply-typed \lc \emph{with} recursive types (\stlcmu) to an \ulcname (\ulc).
Constructing a back-translation of target-level contexts can be done based on a \stlcmu type that can represent
arbitrary \ulc values. Particularly, we can encode the unitype of \ulc values in a type $\UVal$ as follows:
\begin{align*}
  \UVal &\isdef \tl{\mu\alpha \ldotp \batype \uplus (\alpha\times\alpha) \uplus (\alpha \uplus \alpha) \uplus (\alpha \to \alpha)}
\end{align*}
given that \ulc has base values of type \tl{\batype}, pairs, coproducts and lambdas. In other words, all \ulc values can be
represented as \stlcmu values of type $\UVal$. We can then construct a
back-translation of \ulc contexts to \stlcmu contexts such that the latter work with values
in $\UVal$ wherever the original \ulc contexts work with arbitrary \ulc values.

\subsubsection*{Contributions of this paper}
If the types of the source language are not powerful enough to embed
an encoding of target terms, is it possible to have a fully-abstract
compiler between those languages? In this paper we answer positively
to this question and develop a general technique for proving this. We
instantiate this proof technique and develop a fully-abstract compiler
from a simply-typed \lc \emph{without} recursive types (\pstlc) to an
untyped \lc (\ulc). With such a source language, we cannot construct a
type like $\UVal$ to represent the values that a \ulc context works
with. Fortunately, we can solve this problem by observing that a fully
accurate back-translation is sufficient for the proof but in fact not
necessary. An \emph{approximate} back-translation is enough for the
full-abstraction proof to work, without sacrificing the overall
simplicity and elegance of the proof technique. The basic idea is
depicted in \cref{fig:proving-compiler-security-approx}. The
differences from \cref{fig:proving-compiler-security-exact} are the
use of asymmetric logical relations $\lesssim$ and $\gtrsim$ (also
known as logical approximations) to express (roughly) that a term (or
context) $\tl{t}$ terminates whenever $\ul{t}$ does
($\tl{t}\gtrsim\ul{t}$) and vice versa ($\tl{t} \lesssim\ul{t}$) and
the addition of subscripts $n$ where logical approximations hold only
up to a limited number of steps $n$. Note that $n$ in the figure is
defined as the number of steps in the evaluation
$\upc\ulH{\comp{\tl{t_1}}} \ul{\Downarrow_n}$ and that we write $\_$
for an unknown number of steps.

\begin{figure}
  \centering
  \includegraphics{proving-compiler-security-approx.tikz}
  \caption{Proving equivalence preservation using an $n$-\emph{approximate} back-translation. An \_ subscript indicates \emph{any} number of steps.}
  \label{fig:proving-compiler-security-approx}
\end{figure}
The proof starts, again, from an arbitrary target-level context $\upc$
and the knowledge that $\ul{\upc\ulH{\comp{\tl{t_1}}}\Dan{}{n}}$.
We then construct a \pstlc context $\tl{\ef{\upc}\nn}$ that satisfies two conditions.
First, it approximates $\upc$ \emph{up to $n$ steps}: $ \tl{\ef{\upc}\nn} \gtrsim\nn \upc$.
This means that if $\upc\ulH{t}$ terminates in less than $n$ steps then $\tl{\ef{\upc}\nn}\tlH{t}$ will also terminate for a term $\tl{t}$ related to $\ul{t}$.
This, together with the knowledge that
$ \tl{t} \gtrsim \comp{\tl{t}}$, allows us to deduce
implication (1) in the figure. As before, implication (2) follows
directly from the (contextual) equivalence of $\tl{t_1}$ and $\tl{t_2}$.

Then we use a second condition on the
$n$-approximation $\tl{\ef{\upc}\nn}$, namely that it is \emph{conservative}, to deduce implication (3).
Intuitively, the source-level context produced by the $n$-approximation may diverge in situations where the original did not, but not vice versa.
Intuitively, the divergence will occur when the precision $n$ of the approximate back-translation $\tl{\ef{\upc}_n}$ is not sufficient for the context to accurately simulate the behaviour of $\upc$.
This is expressed by the logical approximation
$ \tl{\ef{ \upc}\nn} \lesssim \upc$ which
implies that if $\tl{\ef{\upc}\nn}\tlH{t}$ terminates (in any number of steps), then so must $\upc\ulH{t}$.
This allows us to deduce implication (3).

The advantage of this approximate back-translation approach is that it
can be easier to construct a conservative approximate back-translation
than a full one. For example, considering \pstlc without
recursive types, we can construct a family of \pstlc types $\tl{\UVal\nn}$,
indexed by non-negative numbers $n$:
\begin{align*}
  \tl{\UVal\indexx{0}} &\isdef \tl{\Unit}\\
  \tl{\UVal\indexx{n\mbox{+}1}} &\isdef \tl{
                       \begin{multlined}\Unit \uplus \batype \uplus (\UVal\nn\times\UVal\nn) \uplus\\
                          (\UVal\nn \uplus \UVal\nn) \uplus (\UVal\nn \to \UVal\nn) \text{.}
                       \end{multlined}}
\end{align*}
Without giving full details here, $\tl{\UVal_n}$ is an $n$-level
unfolding of $\tl{\UVal}$ with additional unit values at every level
to represent failed approximations. This approximate version of
$\UVal$ is enough to construct a conservative $n$-approximate
back-translation of an untyped program context, and as such, it allows
us to circumvent the lack of expressiveness of \pstlc without
recursive types.

In order to make this approximate back-translation approach work, we
need a way to formalise the relation between an untyped context and
its approximate back-translation. However, it turns out that existing
well-known techniques from the field of logical relations are almost
directly applicable. Asymmetric logical relations (like
$ \tl{\ef{ \upc}\nn} \lesssim \upc$ above) are a well-established
technique. More interestingly, the approximateness of the relation can
very naturally be expressed using step-indexed logical relations.
Despite this naturality, it appears that this use of step-indexing is
novel. The technique is normally used as a way to construct
well-founded logical relations and one is not actually interested in
terms being related only up to a limited number of steps.

An earlier version of this paper, published at POPL 2016,
introduced the technique of approximate back-translations, and applied
it to prove full abstraction for a whole-program compiler from \pstlc to  \ulc ~\citep{Devriese2016FullyAbsApprox}.
This journal version extends the conference version in two ways.
First, we extend the full abstraction proof to a {\em modular} compiler from \pstlc to  \ulc; i.e., one that operates on open programs (or, components) and links them together after compilation.
This is how most modern compilers operate for efficiency reasons, so this furthers our belief that this proof technique scales to real-world compilers.
Moreover, the work required to extend the proof to a modular compiler is relatively small, so this provides evidence of the broader applicability of our proof technique.
Finally, proving modular full abstraction yields that the compiler ensures component-based compartmentalisation, which provides more fine-grained security guarantees than plain full abstraction.

Second, the original proof has been completely mechanised in Coq, providing additional assurance about the correctness of our results.
Additionally, this highlights that the reasoning principle behind our proof technique is amenable to mechanisation.
To the best of our knowledge, this is the first fully mechanised proof of compiler full-abstraction.

To summarise, the contributions of this work are:
\begin{itemize}
\item a new and general proof technique for proving compiler modular full-abstraction using asymmetric, cross-language logical relations and targeting untyped languages;
\item a fully mechanised  instantiation of that proof technique showing full abstraction of a modular compiler from  a simply-typed \lc without recursive types to the untyped \lc ;
\item a novel application of step-indexed logical relations for expressing approximateness of a back-translation.
\end{itemize}

\smallskip 

This paper is structured as follows.
First, we formalise the source and target languages \pstlc and \ulc (\cref{sec:source-target}).
Second, we present the cross-language logical relations that we use to express the relation between \pstlc terms and their compilations as well as between \ulc contexts and their back-translation (\cref{sec:logical-relations}).
We define the compiler in \cref{sec:compiler}. It applies type erasure and dynamic type wrappers that enforce the requirements and guarantees of \pstlc types during execution.
We then present the approximate back-translation (\cref{sec:appr-back-transl}) which we use to prove compiler full-abstraction (\cref{sec:comp-fa}).
Then we present how to scale the proof technique to modular compilers (\cref{sec:modular}).
Finally, we discuss the mechanisation of the proofs (\cref{sec:coq-proof}).
We then offer some discussion (\cref{sec:disc}), compare with related work (\cref{sec:related-work}) and conclude (\cref{sec:conclusion}).

\section{Source and Target Languages}
\label{sec:source-target}

\begin{figure*}[t]
  \begin{align*}
    \tl{t} &\bnfdef \tl{\unitv} \mid \tl{\truev} \mid \tl{\falsev}\mid \tl{\lambda x:\tau.\ t}\mid \tl{x} \mid \tl{t~t} \mid \tl{\projone{t}} \mid \tl{\projtwo{t}} \mid \tl{\pair{t,t}} \mid\tl{\inl{t}} \mid \tl{\inr{t}}
    \\ &\quad \mid \tl{\caseof{t}{t}{t}}\mid\tl{t;t} \mid\ \tl{\ifte{t}{t}{t}}\mid \tl{\fix{\tau \ra \tau}~t}
    \\
    \tl{v} &\bnfdef \tl{\unitv} \mid \tl{\truev} \mid \tl{\falsev}\mid \tl{\lambda x:\tau.\ t}\mid \tl{\pair{v,v}} \mid \tl{\inl{v}}\mid\tl{\inr{v}}
    \\
    \tl{\tau} &\bnfdef \tl{\Unit} \mid \tl{\Bool} \mid \tl{\tau\to\tau}\mid\tl{\tau\times\tau}\mid \tl{\tau\uplus\tau}
    \\
    \tlGamma&\bnfdef \tle \mid \tlGamma,\tl{x}:\tltau
    \\
    \tlC &\bnfdef \tlH{\cdot} \mid \tlC~\tl{t} \mid \tl{v}~\tlC \mid \tl{\projone{\tlC}} \mid \tl{\projtwo{\tlC}} \mid \tl{\pair{\tlC,t}} \mid \tl{\pair{v,\tlC}} \mid \tl{\inl{\tlC}}\mid \tl{\inr{\tlC}}\mid \tl{\casE{\tlC}{\inl{x_1}\mapsto t_1 \mid \inr{x_2}\mapsto t_2}} 
    \\
    &\quad \mid \tl{\tlC;t} \mid \tl{\ifte{\tlC}{t}{t}} \mid \tl{\fix{\tau \ra \tau}\ \tlC}
  \end{align*}

\botrule
  \begin{mathpar}


    \typerulenolabel{\stlc-unit}{ }{
        \tlGamma\vdash\tl{\unitv}:\tl{\Unit}
    }{stlc-unit} \and

    \typerulenolabel{\stlc-true}{ }{
        \tlGamma\vdash\tl{\truev}:\tl{\Bool}
    }{stlc-true} \and


    \typerulenolabel{\stlc-Type-var}{
        (\tl{x}:\tltau)\in\tlGamma
    }{
        \tlGamma\vdash\tl{x}:\tltau
    }{stlc-var} \and

    \typerulenolabel{\stlc-Type-fun}{
        \tlGamma, (\tl{x}:\tltau)\vdash\tl{t}:\tl{\tau'}
    }{
        \tlGamma\vdash\tl{\lambda x:\tau.~t}:\tl{\tau\to\tau'}
    }{stlc-lambda} \and

    \typerulenolabel{\stlc-Type-pair}{
        \tlGamma\vdash\tl{t_1}:\tl{\tau_1}                      \\
        \tlGamma\vdash\tl{t_2}:\tl{\tau_2}
    }{
        \tlGamma\vdash\tl{\pair{t_1,t_2}}:\tl{\tau_1\times\tau_2}
    }{stlc-pair} \and

    \typerulenolabel{\stlc-Type-proj1}{
        \tl\Gamma\vdash\tl{t}:\tl{\tau_1\times\tau_2}
    }{
        \tlGamma\vdash\tl{\projone{t}}:\tl{\tau_1}
    }{stlc-projone} \and


    \typerulenolabel{\stlc-Type-app}{
        \tlGamma\vdash\tl{t}:\tl{\tau'\to\tau}                &
        \tlGamma\vdash\tl{t'}:\tl{\tau'}
    }{
        \tlGamma\vdash\tl{t~t'}:\tltau
    }{stlc-app} \and



    \typerulenolabel{\stlc-Type-inl}{
        \tlGamma\vdash\tl{t}:\tltau
    }{
        \tlGamma\vdash\tl{\inl{t}} : \tl{\tau\uplus\tau'}
    }{stlc-inl} \and


    \typerulenolabel{\stlc-Type-Fix}{
      \tlGamma\vdash \tl{t} : \tl{(\tau_1 \ra \tau_2) \ra (\tau_1 \ra \tau_2)}
    }{
      \tlGamma\vdash\tl{\fix{\tau_1 \ra \tau_2}}\ \tl{t} : \tl{\tau_1\to\tau_2}
    }{stlc-fix} \and

    \typerulenolabel{\stlc-Type-case}{
        \tlGamma\vdash\tl{t}:\tl{\tau_1\uplus\tau_2}           \\
        \tlGamma, (\tl{x_1}:\tl{\tau_1}) \vdash \tl{t_1} : \tltau   &
        \tlGamma, (\tl{x_2}:\tl{\tau_2}) \vdash \tl{t_2} : \tltau
    }{
        \tlGamma\vdash\tl{\casE{t}{\inl{x_1}\mapsto t_1 \mid \inr{x_2}\mapsto t_2}} : \tltau
    }{stlc-case} \and

    \typerulenolabel{\stlc-Type-if}{
        \tlGamma\vdash\tl{t}:\tl{\Bool}          \\
        \tlGamma \vdash \tl{t_1} : \tltau  &
        \tlGamma \vdash \tl{t_2} : \tltau
    }{
        \tlGamma\vdash\tl{\ifte{t}{t_1}{t_2}} : \tltau
    }{stlc-if} \and

    \typerulenolabel{\stlc-Type-seq}{
        \tlGamma \vdash \tl{t_1} : \tl{\Unit}  \\
        \tlGamma \vdash \tl{t_2} : \tltau
    }{
        \tlGamma\vdash\tl{t_1;t_2} : \tltau
    }{stlc-seq}
  \end{mathpar}

 \botrule
  \begin{mathpar}
	\typerulenolabel{\stlc-Eval-ctx}{
		\tl{t}\stlcto\tl{t'}
	}{
		\tlC\tlH{t}\stlcto\tlC\tlH{t'}
	}{stlc-evalctx} \and

	\typerulenolabel{\stlc-Eval-beta}{
	}{
		\tl{(\lambda x:\tau.~t)~v} \stlcto \tl{t}\tlsub{v}{x}
	}{stlc-evalbeta} \and

	\typerulenolabel{\stlc-Eval-proj1}{
	}{
		\tl{\projone{\pair{v_1,v_2}}} \stlcto \tl{v_1}
	}{stlc-evalprojone} \and



	\typerulenolabel{\stlc-Eval-case-inl}{
	}{
          \tl{\case{~\inl{v}~\of~}\left|
              \begin{aligned}[c]&\inl{x_1}\mapsto t_1\\
              &  \inr{x_2}\mapsto t_2
            \end{aligned}\right.
          \stlcto \tl{t_1}\tlsub{v}{x_1}}
	}{stlc-evalcase-inl} \and


	\typerulenolabel{\stlc-Eval-if-v}{
      \tl{v} \equiv \tl{\truev} \Ra \tl{t'} \equiv \tl{t_1} \\
      \tl{v} \equiv \tl{\falsev} \Ra \tl{t'} \equiv \tl{t_2}
	}{
		\tl{\ifte{v}{t_1}{t_2}} \stlcto \tl{t'}
	}{stlc-evalif-v} \and


	\typerulenolabel{\stlc-Eval-seq-next}{
	}{
		\tl{\unit;t} \stlcto \tl{t}
	}{stlc-evalseq-next} \and

    \typerulenolabel{\stlc-Eval-fix}{
    }{
      \begin{multlined}
        \tl{\fix{\tau_1 \ra \tau_2}\ (\lambda x:\tau_1 \ra \tau_2\ldotp t)} \stlcto\\
          \tl{t\tlsub{(\lambda\ y: \tau_1\ldotp \fix{\tau_1 \ra \tau_2}\ (\lambda x:\tau_1 \to \tau_2\ldotp t)\ y)}{x}}
      \end{multlined}
    }{stlc-eval-fix}
  \end{mathpar}

  \caption{Syntax, static and dynamic semantics of the source language \pstlc (selection of). 
  We denote syntactic equivalence as $\equiv$.
  } 
\label{fig:stlc-fix}
\end{figure*}

The source language \pstlc is presented in \cref{fig:stlc-fix}. It is a
strict, simply-typed \lc with \tl{\Unit}, \tl{\Bool}, lambdas,
product and sum types and a $\tlfix$ operator providing
general recursion. The figure presents the syntax of terms $\tlt$,
values $\tlv$, types $\tltau$, typing contexts $\tlGamma$ and
evaluation contexts $\tlC$. 
We indicate the reduction relation with $\stlcto$; we define that a term \tl{t} terminates (\tl{t\Downarrow}) if it reduces in a finite number of steps to a value: $\tl{t\Downarrow} \isdef \exists \tl{n}, \tl{v}. \tl{t\stlcto^n v}$.
Apart from the type and evaluation rules
for $\tl{\fix{\tau_1\to \tau_2}}$, the typing rules and evaluation
rules are standard. The evaluation rules use evaluation
contexts to impose a strict evaluation order. The type and evaluation
rule for $\tl{\fix{\tau_1\to\tau_2}}$ are somewhat special compared to a more standard definition (see e.g.~\cite{pierce2002types}): the
operator is restricted to function types and an additional
$\eta$-expansion occurs during evaluation. This is because we have
chosen to make $\tl{\fix{}}$ model the Z fixed-point combinator (also known as the call-by-value Y combinator)~\cite[\S 5]{pierce2002types} rather
than the Y combinator.
The reason revolves around the compiler devised in this paper.
The target language of that compiler is a \emph{strict} untyped lambda calculus, where Y does not work but Z does and using Z in \pstlc as well keeps the compiler simpler.
Working with the more standard Y fixpoint combinator in \pstlc is probably possible but would require the compiler to use an encoding that would be pervasive but irrelevant to the subject of this paper.

\pstlc program contexts $\tpc$ are \pstlc terms that contain exactly
one hole~$\tlH{\cdot}$ in place of a subterm. We also omit the typing judgement for program contexts
$\vdash \tpc : \tl{\Gamma'},\tl{\tau'} \to \tl{\Gamma},\tl{\tau}$,
defined by inductive rules close to those for terms in
\cref{fig:stlc-fix}. The judgement guarantees that substituting a
well-typed term $\tl{\Gamma'} \vdash \tl{t} : \tl{\tau'}$ in a
well-typed context
$\vdash \tpc : \tl{\Gamma'},\tl{\tau'} \to \tl{\Gamma},\tl{\tau}$
produces a well-typed term
$\tl{\Gamma} \vdash \tpc\tlH{t} : \tl{\tau}$.

\begin{figure*}[t]
  \begin{align*}
    \ul{t} &\bnfdef \ul{\unitv}\mid\ul{\truev}\mid\ul{\falsev}\mid\ul{\lambda x\ldotp t} \mid \ul{x} \mid \ul{t~t} \mid \ul{\projone{t}} \mid \ul{\projtwo{t}} \mid \ul{\pair{t,t}} \mid\ul{\inl{t}} \mid \ul{\inr{t}}
    \\
    &\quad \mid\ul{\casE{t}{\inl{x}\mapsto t\mid\inr{x}\mapsto t}}\mid\ul{t;t}\mid\ul{\ifte{t}{t}{t}} \mid \wrong
    \\
    \ul{v} &\bnfdef \ul{\unitv} \mid \ul{\truev} \mid \ul{\falsev}\mid \ul{\lambda x\ldotp t}\mid \ul{\pair{v,v}} \mid \ul{\inl{v}}\mid\ul{\inr{v}}
    \\
    \ulGamma&\bnfdef \ule \mid \ulGamma,\ul{x}
    \\
    \ulC &\bnfdef \ulH{\cdot} \mid \ulC~\ul{t} \mid \ul{v}~\ulC \mid \ul{\projone{\ulC}} \mid \ul{\projtwo{\ulC}} \mid \ul{\pair{\ulC,t}} \mid \ul{\pair{v,\ulC}} \mid \ul{\inl{\ulC}}\mid \ul{\inr{\ulC}}
    \\
    &\quad\mid \ul{\casE{\ulC}{\inl{x_1}\mapsto t_1 \mid \inr{x_2}\mapsto t_2}} \mid \ul{\ulC;t} \mid \ul{\ifte{\ulC}{t}{t}}
  \end{align*}
  \botrule
  \begin{mathpar}
	\typerulenolabel{\ulc-Eval-ctx}{
		\ul{t}\ulcto\ul{t'}
	}{
		\ulC\ulH{t}\ulcto\ulC\ulH{t'}
	}{ulc-evalctx} \and

	\typerulenolabel{\ulc-Eval-ctx-wrong}{
      \ulC \neq \ul{\ulH{\cdot}}
    }{
    	\ulC\ulH{\wrong}\ulcto\wrong
	}{ulc-evalctxwrong} \and

	\typerulenolabel{\ulc-Eval-beta}{}{
      \ul{(\lambda x.~t)~v} \ulcto \ul{t}\ulsub{v}{x}
	}{ulc-evalbeta} \and

	\typerulenolabel{\ulc-Eval-proj1}{}{
      \ul{\projone{\pair{v_1,v_2}}} \ulcto \ul{v_1}
	}{ulc-evalprojone} \and

	\typerulenolabel{\ulc-Eval-seq-next}{
      \ul{v} \equiv \ul{\unitv} \Ra \ul{t'} \equiv \ul{t}\\
      \ul{v} \nequiv \ul{\unitv} \Ra \ul{t'} \equiv \wrong
	}{
		\ul{v;t} \ulcto \ul{t'}
	}{ulc-evalseq-next}


	\typerulenolabel{\ulc-Eval-case-inl}{}{
      \ul{\case{~\inl{v}~\of~}
        \left|\begin{aligned}
            &\inl{x_1}\mapsto t_1\\
            &\inr{x_2}\mapsto t_2
        \end{aligned}\right.}
		\ulcto \ul{t_1}\ulsub{v}{x_1}
	}{ulc-evalcase-inl} \and





	\typerulenolabel{\ulc-Eval-if-v}{
      \ul{v} \equiv \ul{\truev} \Ra \ul{t'} \equiv \ul{t_1} & \ul{v} \equiv \ul{\falsev} \Ra \ul{t'} \equiv \ul{t_2}\\
      (\ul{v} \not\equiv \ul{\truev} \wedge \ul{v}\not\equiv\ul{\falsev}) \Ra \ul{t'} \equiv \ul{\wrong}
	}{
		\ul{\ifte{v}{t_1}{t_2}} \ulcto \ul{t'}
	}{ulc-evalif-v} \and




  \end{mathpar}
  \caption{Syntax and dynamic semantics of the target language \ulc (selection of).} 
  \label{fig:ulc}
\end{figure*}
\Cref{fig:ulc} presents the syntax, well-scopedness and evaluation
rules for the target language \ulc: a standard untyped \lc. The
calculus has unit, booleans, lambdas, product and sum values,
and produces a kind of unrecoverable exception in case of type errors
(e.g. projecting from a non-pair value, case splitting on a non-sum
value etc.). Such an unrecoverable exception is represented in a
standard way (see, e.g., \cite[\S 14.1]{pierce2002types}) as a
non-value term $\wrong$ with a special reduction rule. We omit unsurprising 
well-scopedness rules. 
The reduction relation $\ulcto$ and termination for terms \ul{t\Downarrow} are defined analogously to \stlc.
The evaluation rules again
use evaluation contexts to impose a strict evaluation order. Note that
the termination judgement $\ul{t\Da{}}$ requires termination with a
value, i.e. not $\ul{\wrong}$. Again, we omit the well-scopedness judgement for contexts
$\vdash \upc : \ul{\Gamma'} \to \ul{\Gamma}$, which is inductively defined and guarantees that substituting a well-scoped term
$\ul{\Gamma'} \vdash \ul{t}$ for the hole produces a well-scoped
result term $\ul{\Gamma} \vdash \upc\ulH{t}$.

The interested reader can find all proofs in the companion tech report~\citep{Devriese2017ModularFullyAbsApproxTR}.

\section{Logical Relations} \label{sec:logical-relations}
This section presents the Kripke, step-indexed logical relations that we use to prove compiler full-abstraction.
First, this section describes the specifications of the world used by the logical relation (\cref{fig:logrels-worlds}).
Then, it defines the logical relations (\cref{fig:logrels}) and finally it proves standard properties that the relations enjoy.
Part of the logical relation is postponed until \cref{sec:ulc-values-vs}, where we define the back-translation infrastructure that this part depends on.
The goal of this section is to provide an understanding of what it means for two terms to be related; this will be needed for understanding properties of the compiler in the following sections.

\smallskip

The cross-language logical relations used in this paper are roughly based on one by \cite{Hur:2011:KLR:1926385.1926402}.
Essentially, we instantiate their language-generic logical relations to \pstlc and \ulc and simplify them by removing complexities deriving from the System~F type system, public/private transitions, references and garbage collection.

\begin{figure}
  \centering
  \begin{minipage}{.5\linewidth}
    \begin{align*}
      \WW \bnfdef &\ (k) \text{ with } k \in \mathbb{N}\\
      \CWstepsfun{\WW} \mydefsym &\ \WW.k \\
      \CWlater{(0)} \mydefsym &\ (0)
    \end{align*}
  \end{minipage}%
  \begin{minipage}{.5\linewidth}
    \begin{align*}
      \CWlater{(k+1)} \mydefsym &\ (k) \\
      (k) \CWfutw (k') \mydefsym &\ k\leq k' \\
      (k)\strfutw (k')\isdef&\ k < k'
    \end{align*}
  \end{minipage}
  \begin{align*}
    \CWobswfun{\WW}_\lesssim \mydefsym &
                                         \left\{ (\tl{t},\ul{t}) \relmiddle|
                                         \begin{multlined}
                                           \exists k \leq \CWstepsfun{\WW},\tl{v}\ldotp \tl{t\stlcto^k v} \Ra
                                           \exists k',\ul{v}\ldotp \ul{t \ulcto^{k'} v})
                                         \end{multlined}
    \right\}\\
    \CWobswfun{\WW}_\gtrsim \mydefsym &
                                         \left\{ (\tl{t},\ul{t}) \relmiddle|
                                         \begin{multlined}
                                           \exists k \leq \CWstepsfun{\WW},\ul{v}\ldotp \ul{t\ulcto^k v} \Ra
                                           \exists k',\tl{v}\ldotp \tl{t \stlcto^{k'} v})
                                         \end{multlined}
    \right\}
  \end{align*}

  \caption{Logical relations: Worlds.}
  \label{fig:logrels-worlds}
\end{figure}
Since we do not deal with mutable references, we use a very simple
notion of worlds, consisting just of a step-index $k$ that can be accessed with the \CWstepsfun{\cdot} function (\cref{fig:logrels-worlds}).
We define a $\CWlater{}$ modality and a
future world relation $\CWfutw$, expressing that future worlds allow
less reduction steps to be taken.
We define two different observation
relations $\CWobswfun{\WW}_\lesssim$ and $\CWobswfun{\WW}_\gtrsim$.
The former defines that a \pstlc term $\tl{t}$ approximates a \ulc term
$\ul{t}$ if termination of the first in less than $\CWstepsfun{\WW}$
steps implies termination of the second (in an unknown number of
steps).
The latter requires the reverse.
All of our logical relations will be defined in
terms of either $\CWobswfun{\WW}_\lesssim$ or
$\CWobswfun{\WW}_\gtrsim$. For definitions and lemmas or theorems that
apply for both instantiations, we use the symbol $\genlogrel$ as a
metavariable that can be instantiated to either $\lesssim$ and
$\gtrsim$.

\begin{figure*}
  Pseudo-types \tl{\tlhat{\tau}}, pseudo-contexts $\tlhat{\Gamma}$, $\fun{oftype}{\cdot}$ and $\replemul{\cdot}$.
  \begin{align*}
    \tl{\tlhat{\tau}} &\bnfdef\tl{\Bool} \mid \tl{\Unit} \mid \tl{\tlhat{\tau} \times \tlhat{\tau}} \mid \tl{\tlhat{\tau} \uplus \tlhat{\tau}} \mid \tl{\tlhat{\tau}\to \tlhat{\tau}} \mid \tl{\EmulDV\np}\\
    \tl{\tlhat{\Gamma}} &\bnfdef \tl{\emptyset} \mid \tl{\tlhat{\Gamma}},\tl{x} : \tl{\tlhat{\tau}}\\
    \replemul{\tlhat{\tau}} &\isdef \cdots \text{ (to be defined later, in \cref{fig:helper-emul})}\\
    \tllsfun{oftype}{\tl{\tlhat{\tau}}} &\isdef \{ \tl{v} \mid \tle \vdash \tl{v} : \replemul{\tl{\tlhat{\tau}}} \}\\
    \ullsfun{oftype}{\tl{\tlhat{\tau}}} &\isdef \left\{ \ul{v} \relmiddle|
                               \begin{aligned}
                                &\ul{v} = \ul{\unitv} &&\text{ if } \tl{\tlhat{\tau}} = \tl{\Unit}\\
                                &\ul{v} = \ul{\truev} \text{ or } \ul{v} = \ul{\falsev} &&\text{ if } \tl{\tlhat{\tau}} = \tl{\Bool}\\
                                &\exists \ul{t}\ldotp \ul{v} = \ul{\lam{x}{t}}&&\text{ if } \exists \tl{\tlhat{\tau_1},\tlhat{\tau_2}}\ldotp \tl{\tlhat{\tau}} = \tl{\tlhat{\tau_1}\to \tlhat{\tau_2}}\\
                                &\exists \ul{v_1} \in \ullsfun{oftype}{\tl{\tlhat{\tau_1}}}, \ul{v_2} \in\ullsfun{oftype}{\tl{\tlhat{\tau_2}}} \ldotp \ul{v} = \ul{\pair{v_1,v_2}}&&\text{ if } \exists \tl{\tlhat{\tau_1},\tlhat{\tau_2}}\ldotp \tl{\tlhat{\tau}} = \tl{\tlhat{\tau_1}\times \tlhat{\tau_2}}\\
                                &\exists \ul{v_1} \in \ullsfun{oftype}{\tl{\tlhat{\tau_1}}}\ldotp \ul{v} = \ul{\inl{v_1}} 
                                \\
                                &\ \text{ or } \exists\ul{v_2} \in\ullsfun{oftype}{\tl{\tlhat{\tau_2}}} \ldotp \ul{v} = \ul{\inr{v_2}}&&\text{ if } \exists \tl{\tlhat{\tau_1},\tlhat{\tau_2}}\ldotp \tl{\tlhat{\tau}} = \tl{\tlhat{\tau_1}\uplus \tlhat{\tau_2}}
                               \end{aligned}\right\}\\
    \fun{oftype}{\tl{\tlhat{\tau}}} &\isdef \left\{ (\tl{v},\ul{v}) \relmiddle| \tl{v} \in \tllsfun{oftype}{\tl{\tlhat{\tau}}} \wedge \ul{v} \in \ullsfun{oftype}{\tl{\tlhat{\tau}}} \right\}
  \intertext{Logical relations for values (\valrel{\cdot}), contexts (\contrel{\cdot}), terms (\termrel{\cdot}) and substitutions (\envrel{\cdot}).}
    \later R &\isdef \myset{(\WW,\tl{v},\ul{v})}{ \CWstepsfun{\WW} > 0 \Ra (\later\WW,\tl{v},\ul{v}) \in R}
    \\
    \valrel{\Unit}_\genlogrel &\isdef \myset{ (\WW, \tlv, \ulv) }{\tl{v} = \tl{\unitv} \text{ and } \ul{v} = \ul{\unitv} }
    \\
    \valrel{\Bool}_\genlogrel &\isdef \myset{ (\WW, \tlv, \ulv) }{\exists \bl{v}\in\{\truev,\falsev\}\ldotp \tl{v} = \bl{v} \text{ and } \ul{v} = \bl{v} }
    \\
    \valrel{\tlhat{\tau'}\to\tlhat{\tau}}_\genlogrel &\mydefsym \left\{ (\WW, \tl{v}, \ul{v}) \relmiddle|
                                              \begin{aligned}
                                                &(\tl{v},\ul{v}) \in \fun{oftype}{\tl{\tlhat{\tau'}\to\tlhat{\tau}}} \text{ and } 
                                                \\
                                                &\exists \tl{t},\ul{t}\ldotp \tl{v} = \tl{\lam{x: \replemul{\tlhat{\tau'}}}{t}} \text{ and } \ul{v} = \ul{\lam{x}{t}} \text{ and } 
                                                \\
                                                &\forall \WW'\strfutw \WW, (\WW', \tl{v'},\ul{v'}) \in\valrel{\tlhat{\tau'}}_\genlogrel \ldotp (\WW', \tl{t}\tlsub{v'}{x}, \ul{t}\ulsub{v'}{x}) \in \termrel{\tl{\tlhat{\tau}}}_\genlogrel
                                              \end{aligned}
                                                  \right\} \\
    \valrel{\tlhat{\tau_1}\times\tlhat{\tau_2}}_\genlogrel &\isdef\left\{ (\WW, \tl{v}, \ul{v}) \relmiddle|
                                                       \begin{aligned}
                                                         &(\tl{v},\ul{v}) \in \fun{oftype}{\tl{\tlhat{\tau_1} \times \tlhat{\tau_2}}} \text{ and } 
                                                         \\
                                                         &\exists \tl{v_1,v_2},\ul{v_1,v_2}\ldotp \tl{v} = \tl{\pair{v_1,v_2}} \text{ and } \ul{v} = \ul{\pair{v_1,v_2}} \text{ and }
                                                         \\
                                                         &(\WW, \tl{v_1}, \ul{v_1}) \in \later\valrel{\tlhat{\tau_1}}_\genlogrel \text{ and } (\WW, \tl{v_2}, \ul{v_2}) \in \later\valrel{\tlhat{\tau_2}}_\genlogrel
                                                       \end{aligned}
                                                     \right\}\\
    \valrel{\tlhat{\tau_1}\uplus \tlhat{\tau_2}}_\genlogrel &\isdef\left\{ (\WW, \tl{v}, \ul{v}) \relmiddle|
                                                      \begin{aligned}
                                                        &(\tl{v},\ul{v}) \in \fun{oftype}{\tl{\tlhat{\tau_1}\uplus\tlhat{\tau_2}}} \text{ and either }
                                                        \\
                                                        &\exists \tl{v'},\ul{v'}.\  (\WW,\tl{v'},\ul{v'}) \in \later\valrel{\tlhat{\tau_1}}_\genlogrel \text{ and } \tl{v} = \tl{\inl{v'}} \text{ and } \ul{v} = \ul{\inl{v'}} \text{ or }
                                                        \\
                                                        &\exists \tl{v'},\ul{v'}.\	(\WW,\tl{v'},\ul{v'}) \in \later\valrel{\tlhat{\tau_2}}_\genlogrel \text{ and } \tl{v} = \tl{\inr{v'}} \text{ and } \ul{v} = \ul{\inr{v'}}
                                                      \end{aligned}
                                                          \right\}
    \\
    \valrel{\EmulDV\np}_\genlogrel &\isdef \cdots  \text{ (to be defined later, in \cref{fig:emuldv})}
    \\
    \contrel{\tl{\tlhat{\tau}}}_\genlogrel &\isdef\myset{ (\WW,\tlC, \ulC) }{ \forall\WW'\CWfutwpub\WW, (\WW', \tl{v},\ul{v})\in\valrel{\tlhat{\tau}}_\genlogrel \ldotp (\tlC\tlH{v},\ulC\ulH{v})\in\CWobswfun{\WW'}_\genlogrel }
    \\
    \termrel{\tl{\tlhat{\tau}}}_\genlogrel &\isdef\myset{ (\WW,\tl{t},\ul{t}) }{ \forall(\WW,\tlC,\ulC) \in \contrel{\tl{\tlhat{\tau}}}_\genlogrel\ldotp (\tlC\tlH{t},\ulC\ulH{t})\in\CWobswfun{\WW}_\genlogrel}\\
    \envrel{\tle}_\genlogrel &\isdef\{ (\WW, \tle, \ule) \}
    \\
    \envrel{\tl{\tlhat{\Gamma}}, (\tl{x}:\tl{\tlhat{\tau}})}_\genlogrel &\isdef \myset{ (\WW, \tlgamma[\tl{x\mapsto v}], \ulgamma[\ul{x \mapsto v}]) }{ (\WW,\tlgamma,\ulgamma) \in \envrel{\tl{\tlhat{\Gamma}}}_\genlogrel \text{ and } (\WW, \tl{v}, \ul{v}) \in \valrel{\tlhat{\tau}}_\genlogrel}
  \end{align*}
  \caption{\protect Logical relations (partial, the missing definition can be found in \cref{fig:emuldv,fig:helper-emul}).}
\label{fig:logrels}
\end{figure*}

\begin{figure*}
  Logical relations for open terms and program contexts.
  \begin{align*}
    \tl{\tlhat{\Gamma}}\vdash \tl{t}\arbsim\nn\ul{t} : \tl{\tlhat{\tau}} &\isdef \left\{
                                                                           \begin{multlined}
                                                                           \replemul{\tl{\tlhat{\Gamma}}}\vdash\tl{t}:\replemul{\tl{\tlhat{\tau}}} \text{ and } \dom{\tlhat{\Gamma}} \vdash \ul{t} \text{ and }
    \\
     \forall \WW\ldotp \CWstepsfun{\WW} \leq n \Ra \forall (\WW, \tlgamma, \ulgamma)\in \envrel{\tl{\tlhat{\Gamma}}}_\genlogrel\ldotp (\WW, \tl{t}\tlgamma, \ul{t}\ulgamma) \in \termrel{\tl{\tlhat{\tau}}}_\genlogrel
  \end{multlined}
      \right.
    \\
    \tl{\tlhat{\Gamma}}\vdash \tl{t}\arbsim\ul{t} : \tl{\tlhat{\tau}} &\isdef \tl{\tlhat{\Gamma}}\vdash \tl{t}\arbsim\nn\ul{t} : \tl{\tlhat{\tau}} \text{ for all } n
    \\
    \vdash\tpc \arbsim\nn\upc : \tl{\tlhat{\Gamma}'},\tl{\tlhat{\tau}'\to\tlhat{\Gamma}},\tl{\tlhat{\tau}} &\isdef\
                                                                                                             \left\{
                                                                                                             \begin{multlined}
                                                                                                               \vdash\tpc :\replemul{\tl{\tlhat{\Gamma}'}},\replemul{\tl{\tlhat{\tau}'}} \to \replemul{\tl{\tlhat{\Gamma}}},\tl{\replemul{\tlhat{\tau}}}\\
                                                                                                              \text{ and } \vdash\upc :\dom{\tl{\tlhat{\Gamma}'}} \to \dom{\tlhat{\Gamma}}
                                                                                                               \\
                                                                                                               \text{and for all }\tlt,\ult\ldotp \text{ if } \tl{\tlhat{\Gamma'}} \vdash \tl{t}\arbsim\nn\ul{t} : \tl{\tlhat{\tau'}} \text{, then } \tl{\tlhat{\Gamma}} \vdash \tpc\tlH{t} \arbsim\nn\upc\ulH{t} :\tl{\tlhat{\tau}}
                                                                                                             \end{multlined} \right.
  \end{align*}
  \caption{\protect Definition of logically related terms and contexts.}
\label{fig:logrel-termcontext}
\end{figure*}
\Cref{fig:logrels} contains the definition of the logical relations.
The first thing to note is that our logical relations are not indexed
by \pstlc types~\tltau, but by \emph{pseudo-types}
$\tl{\tlhat{\tau}}$. The syntax for these pseudo-types contains all the constructs of \pstlc types, plus an
additional token type $\tl{\EmulDV\np}$, indexed by a
non-negative number $n$ and a value $p ::= \precise \mid \imprecise$.
This token type is not a \pstlc type; it is needed because of the approximate back-translation.
When necessary, we use a
function $\replemul{}$ for converting a pseudo-type to a \pstlc type.
The function replaces all occurrences of $\tl{\EmulDV\np}$ with a
concrete \pstlc type.
We postpone the definitions and explanations of $\tl{\EmulDV\np}$ and of $\valrel{\EmulDV\np}_\genlogrel$ to \cref{sec:ulc-values-vs}, after we have given some more information about the back-translation.
We will sometimes silently use normal types where pseudo-types are expected, which makes sense since the latter are a superset of the former.

The value relation $\valrel{\tlhat{\tau}}_\genlogrel$ is defined by induction on
the pseudo-type. Most definitions are quite standard. All cases require related terms to be in the $\funname{oftype}$ relation, which requires well-typedness of the \pstlc term and an appropriate shape for the \ulc value. \tl{\Unit} and \tl{\Bool}
values are related in any world iff they are the same base value. Pair
values are related if both are pairs
and the corresponding components are related in strictly future worlds
at the appropriate pseudo-type. Similarly, sum values are related if
they are both of either the form $\inl{\cdots}$ or $\inr{\cdots}$ and
if the contained values are related in strictly future worlds at the
appropriate pseudo-type. Finally, function values are related if they have the right type, if both are lambdas and if substituting
related values in the body yields related terms in any strictly future world.

The relation on values, evaluation contexts and terms are defined mutually
recursively, using a technique known as biorthogonality (see, e.g., \cite{bistcc}).
So, evaluation contexts are related in a world if plugging in
related values in any future world yields related observations.
Similarly, terms are related if plugging the terms in
related evaluation contexts yields related observations.
Relation $\envrel{\tlGamma}_\genlogrel$ relates substitutions instantiating a context
$\tlGamma$, which simply requires that substitutions for all variables
in the context are related at their types.

\Cref{fig:logrel-termcontext} contains the definition of logically-related open and closed terms as well as contexts.
For open terms, we define a logical relation $\tl{\tlhat{\Gamma}}\vdash \tl{t}\arbsim\nn\ul{t} : \tl{\tlhat{\tau}}$.
This relation expresses that an open \pstlc term $\tl{t}$ is related up to $n$ steps to an open \ulc term $\ul{t}$ at pseudo-type $\tl{\tlhat{\tau}}$ in pseudo-context $\tl{\tlhat{\Gamma}}$ if the first is well-typed, the second is well-scoped and if closing $\tl{t}$ and $\ul{t}$ with substitutions related at pseudo-context $\tl{\tlhat{\Gamma}}$ produces terms related at pseudo-type $\tl{\tlhat{\tau}}$, in any world $\WW$ such that $\CWstepsfun{\WW} \leq n$.
If $\tl{\tlhat{\Gamma}}\vdash \tl{t}\arbsim\nn\ul{t} : \tl{\tlhat{\tau}}$ for any $n$, then we write $\tl{\tlhat{\Gamma}}\vdash \tl{t}\arbsim\ul{t} : \tl{\tlhat{\tau}}$.
Finally, we define a logical relation for program contexts
$\vdash\tpc \arbsim\upc :
\tl{\tlhat{\Gamma'}},\tl{\tlhat{\tau'}} \to
\tl{\tlhat{\Gamma}},\tl{\tlhat{\tau}}$
which requires that substituting terms related at the appropriate
pseudo-type produces terms related at the appropriate
pseudo-type.

It is interesting to note that the simple type system of our source
calculus does not actually present a technical need for the use of
step-indexing. Because there are no recursive types or general
references, it is a simple enough system that we can give well-founded
logical relations without any step-indexing. However, as mentioned
before, we use step-indexing for a different reason than other work:
not for constructing a well-founded logical relation, but for
stating that two terms are related \emph{only up to a certain number of
steps}. More details follow in \cref{sec:appr-back-transl}.

These logical relations are constructed so that termination of one
implies termination of the other, according to the direction of the
approximation ($\lesssim$ or $\gtrsim$, \cref{lem:adequacy}).
\begin{lem}[Adequacy for $\lesssim$ and $\gtrsim$] \label{lem:adequacy}
  \hfill
  \begin{itemize}
  \item If $\tle \vdash \tl{t} \lesssim\nn \ul{t} : \tltau$ and $\tl{t \stlcto^m
      v}$ with $n \geq m$, then $\ul{t \Downarrow}$.
  \item If $\tle \vdash \tl{t} \gtrsim\nn \ul{t} : \tltau$ and $\ul{t \ulcto^m v}$ with $n \geq m$, then $\tl{t \Downarrow}$.
  \end{itemize}
\end{lem}

\section{The Compiler} \label{sec:compiler}
This section presents our compiler from \pstlc to \ulc.
The compiler proceeds in two passes: type erasure (\cref{fig:erase}) and dynamic typechecking wrappers (\cref{fig:dynamictypechecks}).

\begin{figure}
  \centering
  \begin{align*}
    \ul{\mi{fix}} \isdef \ul{\lambda f\ldotp (\lambda x\ldotp f\ (\lambda y\ldotp x\ x\ y))\ (\lambda x\ldotp f\ (\lambda y\ldotp x\ x\ y))}
  \end{align*}
  \begin{minipage}{.3\linewidth}%
    \begin{align*}
      \erase{\unitv}  &\isdef \ul{\unitv}\\
      \erase{\falsev} &\isdef \ul{\falsev}\\
      \erase{x} &\isdef \ul{x}
    \end{align*}
  \end{minipage}\,
  \begin{minipage}{.6\linewidth}
    \begin{align*}
      \erase{\pair{t_1,t_2}}  &\isdef \ul{\pair{\erase{t_1},\erase{t_2}}} \\
      \erase{t_1;t_2} &\isdef \ul{\erase{t_1};\erase{t_2}}\\
      \erase{t_1~t_2} &\isdef   \erase{t_1}~\erase{t_2}
    \end{align*}
  \end{minipage}
  \begin{minipage}{.4\linewidth}%
    \begin{align*}
      \erase{\projone{t}} &\isdef \ul{\projone{\erase{t}}}\\
      \erase{\projtwo{t}} &\isdef \ul{\projtwo{\erase{t}}} \\
      \erase{\inl{t}} &\isdef \ul{\inl{\erase{t}}}\\
      \erase{\inr{t}} &\isdef \ul{\inr{\erase{t}}}
    \end{align*}
  \end{minipage}
  \begin{minipage}{.5\linewidth}
    \begin{align*}
      \erase{\truev} &\isdef \ul{\truev}\\
      \erase{\lambda x:\tau.~t} &\isdef \ul{\lambda x.~}\erase{t} \\
      \erase{\tl{\fix{\tau_1\ra\tau_2}\ t}} &\isdef \ul{\mi{fix}\ \erase{t}}
    \end{align*}
  \end{minipage}
  \begin{align*}
    \erase{\ifte{t}{t_1}{t_2}} &\isdef 
      \ul{\mr{if}~\erase{t}~}
      \ul{\mr{then}~\erase{t_1}~}
      \ul{\mr{else}~\erase{t_2}}
  \end{align*}
  \vspace{-.5cm}
  \begin{align*}
    \erase{\caseof{t}{t_1}{t_2}} &\isdef
    \ul{\case~\erase{t}~\of}~\color{\ulccol}\left|\begin{aligned}
        &\ul{\inl{x_1}\mapsto\erase{t_1}}\\
        &\ul{\inr{x_2}\mapsto\erase{t_2}}
      \end{aligned}\right.
  \end{align*}
  \caption{Type erasure: the first pass of the compiler.}
  \label{fig:erase}
\end{figure}
The erasure function is called \erasen; it converts all \pstlc constructs to the corresponding \ulc constructs.
$\tl{\fix{\tau_1\to\tau_2}}$ is erased to a \ulc definition of the Z combinator $\ulfix$.

The $\erasen$ function can be considered as a compiler, but it is only
a correct compiler, not a fully-abstract one, as explained in \cref{ex:erase-is-not-fa}.

\begin{exa}[Erasure is correct but not secure~\citep{scoo-j,fstar2js}]\label{ex:erase-is-not-fa}
Consider the following, contextually equivalent \pstlc functions of type $\tl{\Unit \to \Unit}$:
\begin{equation*}
  \tl{\lam{ x: \Unit}{ x} \qquad \ceq \qquad \lam{x: \Unit}{ \unitv}}
\end{equation*}
The $\erasen$ function will map these to the following \ulc functions:
\begin{equation*}
  \ul{\lam{ x}{ x} \qquad \nceq \qquad \lam{x}{ \unitv}}
\end{equation*}
The results of \erasen are \emph{not} contextually equivalent, essentially
because applying them to a non-unit value like $\ul{\truev}$ will
produce $\ul{\truev}$ for the left lambda and $\ul{\unitv}$ for the
right lambda. In this example, contextual equivalence is not preserved
because the original functions are only defined for $\Unit$ values,
but their compilations can be applied to other values too.
\end{exa}

The following lemma states that every \pstlc term is related to its erased term at its type.
\begin{lem}[Erase is semantics-preserving (for terms)]\label{lem:erase-correct}\ \\
  If $\tl{\Gamma \vdash t : \tau}$, then $\tl{\Gamma} \vdash \tlt \arbsim \erase{\tlt} : \tltau$.
\end{lem}
An analogous result applies to program contexts:
\begin{lem}[Erase is semantics-preserving (for contexts)]\label{lem:erase-sempres-ctx}~\\
  If $\vdash\tpc : \tl{\Gamma'},\tl{\tau'} \to \tlGamma,\tltau$,
  then
  $\vdash \tpc \arbsim\erase{\tpc} : \tl{\Gamma'},\tl{\tau'} \to \tlGamma,\tltau$.
\end{lem}
One should intuitively understand this result as ``$\tlt$ behaves the
same as $\erase{\tlt}$ when both are treated as values of type
$\tltau$''. The result does not specify what happens when we treat
$\ul{t}$ as a value of a different type, like we did in
\cref{ex:erase-is-not-fa} to demonstrate a full abstraction failure.
Intuitively, it only specifies a kind of \emph{equivalence
  reflection} for the $\erasen$ function, not \emph{preservation}.

\begin{figure}
  \centering
  \begin{align*}
  \prot{\Unit} &\isdef \ul{\lambda x\ldotp x}
    \qquad \qquad
  \prot{\Bool} \isdef \ul{\lambda x\ldotp x}
    \\
  \prot{\tl{\tau_1\times\tau_2}} &\isdef \ul{\lambda y\ldotp \ul{\pair{\prot{\tl{\tau_1}}~\projone{y},\prot{\tl{\tau_2}}~\projtwo{y}}}}
    \\
  \prot{\tl{\tau_1\uplus\tau_2}} &\isdef \ul{\lambda y\ldotp \case~y~\of~\left|
                                          \begin{aligned}
                                            &\ul{\inl{x}\mapsto \inl{(\prot{\tl{\tau_1}}\ x)}}\\
                                            &\ul{\inr{x}\mapsto \inr{(\prot{\tl{\tau_2}}\ x)}}
                                          \end{aligned}\right.}
    \\
  \prot{\tl{\tau_1\to\tau_2}} &\isdef \ul{\lambda y\ldotp \lambda x. \prot{\tl{\tau_2}}~(y~(\conf{\tl{\tau_1}}~x))}
    \\
    \\
    \conf{\Unit} &\isdef \ul{\lam{ y}{ (y;\unitv)}}
    \\
    \conf{\Bool} &\isdef \ul{\lam{ y}{ \ifte{y}{\truev}{\falsev} }}
    \\
  \conf{\tl{\tau_1\times\tau_2}} &\isdef \ul{\lambda y\ldotp \ul{\pair{\conf{\tl{\tau_1}}~{\projone{y}},\conf{\tl{\tau_2}}~\projtwo{y}}}}
    \\
  \conf{\tl{\tau_1\uplus\tau_2}} &\isdef \ul{\lambda y\ldotp \case~y~\of~\left|
                                          \begin{aligned}
                                            &\ul{\inl{x}\mapsto \inl{(\conf{\tl{\tau_1}}\ x)}}\\
                                            &\ul{\inr{x}\mapsto \inr{(\conf{\tl{\tau_2}}\ x)}}
                                          \end{aligned}\right.}
    \\
  \conf{\tl{\tau_1\to\tau_2}} &\isdef \ul{\lambda y\ldotp \lambda x\ldotp \conf{\tl{\tau_2}}~(y~(\prot{\tl{\tau_1}}~x))}
  \end{align*}

  \caption{Dynamic type checking wrappers: the second pass of the compiler.}\label{fig:dynamictypechecks}
\end{figure}
Remember that a fully-abstract compiler must protect terms from being used in ways that are not allowed by their
type, as in \cref{ex:erase-is-not-fa}.
This is taken care of by the second pass of the compiler.

We construct a family of dynamic typechecking
wrappers $\prot{\tltau}$ and $\conf{\tltau}$.
$\prot{\tltau}$ is a \ulc term that wraps an argument to
enforce that it can only \emph{be used} in ways that are valid according to type
$\tltau$, as often done in secure compilation work~\citep{scoo-j,nonintfree,fstar2js,Ahmed:2008:TCC:1411203.1411227}.
Dually, $\conf{\tltau}$ wraps its argument so that it
can only \emph{behave} in ways that are valid according to type $\tltau$.
In the
definition, the cases for product and coproduct types simply
recursively descend on their subterms preserving the expected syntax of a product or coproduct argument.
Protecting at a function type means wrapping the function to confine its arguments and protect its results, and dually for confining at a function type.
Finally,
protecting at a base type (i.e., \tl{\Unit} or \tl{\Bool}) does nothing, simply because there is nothing
one can do to a base value that is not allowed by its type.
Confining
a value at a base type is more interesting. Both for \tl{\Unit} and \tl{\Bool}
values, we use the value in such a way that will only work when the
value is actually of the correct type. If it is, we return the
original value, otherwise the term will reduce to \wrong.\footnote{It would also be valid to diverge in this case, if \ulc had some form of dynamic type test which allowed us to do that.}

\paragraph{\emph{Remark on different inhabitants for \tl{\Unit}}}
In some lambda calculi, any value can be given type \tl{\Unit} and the sequencing operator does not require the term before the semicolon to be \tl{\unitv} in order to reduce, but an arbitrary value.
Thus, two syntactically-different source terms \tl{\truev} and \tl{\unitv} can be semantically equivalent at type \tl{\Unit} in this language.
Interestingly, in such a language, the definitions of \prot{\tl{\Unit}} and \conf{\Unit} would need to be swapped:
\begin{align*}
  \protd{\Unit} &\isdef \ul{\lam{ y}{ (y;\unitv)}}
  &
  \confd{\Unit} &\isdef \ul{\lambda x\ldotp x}
  \end{align*}
In fact, in the target language a context can distinguish what is syntactically different but semantically equivalent in the source.
Thus the compiler needs to enforce that the chosen representation for terms of type \tl{\Unit} in the compiled code is unique.
Conversely, any value received from a target context is now valid at type \tl{\Unit}, so no checks are made.



\begin{exa}[Protect and confine make a term secure]\label{ex:protect}
  Consider the protect wrapper $\prot{\Unit \to \Unit}$ for type \tl{\Unit\to\Unit}, which is (roughly) equal to
  $\ul{\lam{y}{\lam{x}{y~(x;\unitv)}}}$. Applying that wrapper to a
  function \ul{f} (i.e.
  $\prot{\Unit \to \Unit}~\ul{f}$) reduces to
  $\ul{\lam{x}{f~(x;\unitv)}}$. Applying this value to a non-\unitv value
  will simply evaluate to $\ul{\wrong}$, therefore addressing the
  issues of \cref{ex:erase-is-not-fa}.
\end{exa}
For the second pass of the compiler, \cref{lem:protect-compatibility-approx} holds.
\begin{lem}[Protect and confine are semantics-preserving]\label{lem:protect-compatibility-approx}~\\
  If $\tl{\Gamma} \vdash \tlt \arbsim\nn \ul{t} : \tltau$, then
       $\tlGamma \vdash \tlt \arbsim\nn \prot{\tltau}~\ul{t} : \tltau$ and
       $\tlGamma \vdash \tlt \arbsim\nn \conf{\tltau}~\ul{t} : \tltau$.
\end{lem}
\Cref{lem:protect-compatibility-approx} states that if $\tlt$ is related to $\ult$ at type
$\tltau$, then adding a $\prot{\tltau}$ or $\conf{\tltau}$ wrapper around
$\ul{t}$ does not change that. In other words, the wrappers do not
change the behaviour of $\ul{t}$ as long as they are treated as values of type $\tltau$. In \cref{sec:injextr}, we will have more to say about the security of the wrappers.

\smallskip

This section concludes with the definition of the compiler used in this paper.
\begin{defi}[The $\comp{\cdot}$ compiler]\label{def:comp}
If $\tlGamma \vdash \tlt : \tltau$, then $\tlt$ is compiled to $\comp{\tlt}$ and:
 $ \comp{\tlt} \isdef \prot{\tltau}~(\erase{\tlt})\text{.}$
\end{defi}
\Cref{lem:erase-correct,lem:protect-compatibility-approx} about the first and second pass of the compiler can be combined into \cref{thm:sempres} to obtain that a \pstlc term of type \tltau behaves like its compilation when both are treated as terms of type~$\tltau$.
\begin{lem}[$\comp{\cdot}$ is semantics-preserving]\label{thm:sempres}
  If  $\tlGamma\vdash\tlt:\tltau$, then $\tlGamma\vdash\tlt\arbsim\comp{\tlt} : \tltau$.
\end{lem}


\section{Approximate Back-Translation}
\label{sec:appr-back-transl}
This section presents the core idea of our proof technique: the approximate back-translation.
As explained in \cref{sec:introduction},
the idea is to translate a target language program context $\upc$ to a
source language program context $\tl{\ef{\upc}\nn}$ which
conservatively $n$-approximates \upc.
Intuitively, this means that
$\tl{\ef{\upc}\nn}$ behaves like $\upc$ for up to $n$ steps but it may
diverge in cases where the original did not if $\upc$ takes more than
$n$ steps.
We will make this more precise in \cref{sec:ulc-values-vs}.

At the core of the approximate back-translation is the \pstlc type
$\tl{\UVal\nn}$.
The type is essentially a \pstlc encoding of the unitype of \ulc.
Where the untyped context $\upc$ manipulates
arbitrary \ulc values, its back-translation $\tl{\ef{\upc}\nn}$ manipulates values of type $\tl{\UVal\nn}$.
\cref{sec:uval} defines \tl{\UVal\nn} and the basic tools (constructors and destructors) for working with it.
To explain how values in \tl{\UVal\nn} model values in \ulc, \cref{sec:ulc-values-vs} fills in the missing piece of the logical relations of \cref{fig:logrels} by defining $\valrel{EmulDV\np}_\genlogrel$.

The type $\tl{\UVal\nn}$ is sufficiently large to contain $n$-approximations of \ulc values.
However, it also contains approximations of \ulc values \emph{up to less than $n$ steps}.
Sometimes, values of type $\tl{\UVal\nn}$ will be \emph{downgraded} to a type $\tl{\UVal\indexx{m}}$ with $m < n$.
Dually, there will be cases where some values need to \emph{upgrade}.
\cref{sec:updown} defines functions to perform value upgrading and downgrading.

With $\tl{\UVal\nn}$ and the related machinery introduced,
\cref{sec:emulate} constructs the function $\tl{\emulate\nn}$, responsible for emulating a context such that it translates a \ulc term $\ul{t}$ into a \pstlc term of type $\tl{\UVal\nn}$.
This function is easily extended to work with program contexts, producing contexts with hole of type $\tl{\UVal\nn}$ as expected.

However, remember from \cref{fig:proving-compiler-security-approx} in \cref{sec:introduction} that the goal of the back-translation is generating a context $\tl{\ef{\upc}\nn}$ whose hole can be filled with \pstlc terms $\tl{t_1}$ and $\tl{t_2}$.
Their type is not $\tl{\UVal\nn}$ but an arbitrary \pstlc type $\tltau$.
Thus, there is a type mismatch between the hole of the emulated context $\tl{\emulate\indexx{n}(\upc)}$ and the terms that we want to plug in there.
Since the emulated contexts work with $\tl{\UVal\nn}$ values, we need a function that wraps terms of an arbitrary type $\tltau$ into a value of type $\tl{\UVal\nn}$.
This is precisely what \cref{sec:injextr} defines, namely a function $\tl{\inject{\tau;n}}$ of type $\tl{\tau \to \UVal\nn}$.

Finally, \cref{sec:approx-final} defines the approximate
back-translation function $\tl{\ef{\cdot}\indexx{\tau;n}}$, mapping a
\ulc context \upc\ to a \pstlc context $\tl{\ef{\upc}\indexx{\tau;n}}$.
The additional index $\tau$ w.r.t. earlier discussions is needed to
introduce an appropriate call to $\tl{\inject{\tau;n}}$ as discussed
above, so that the hole of $\tl{\ef{\upc}\indexx{\tau;n}}$ is of type
$\tltau$. Plugging a term $\tl{t_1}$ in $\tl{\ef{\upc}\indexx{\tau;n}}$
$n$-approximates plugging in the compilation $\comp{\tl{t_1}}$ in
context $\upc$.

\smallskip

Immediately after the definition of each of the concepts discussed above ($\tl{\downgrade}$, $\tl{\upgrade}$, $\tl{\inject{\tau;n}}$ and $\tl{\emulate\nn}$), this section formalises the results about their behaviour.
These results are expressed in terms of the logical relations of \cref{fig:logrels} and of the $\tl{\EmulDV\np}$ pseudo-type; they will be used to prove equivalence preservation in \cref{sec:comp-fa}.

\subsection{\UVal and its Tools}\label{sec:uval}
The family of types \tl{\UVal\tlnn} is defined as follows:
\begin{align*}
  \tl{\UVal\indexx{0}} &\mydefsym \tl{\Unit}\\
  \tl{\UVal\indexx{n+1}} &\mydefsym
                     \begin{multlined}
                       \tl{\Unit \uplus \Unit \uplus \Bool \uplus (\UVal\nn \times \UVal\nn)\uplus }\\
                       \tl{(\UVal\nn \uplus \UVal\nn)\uplus (\UVal\nn \ra \UVal\nn) }
                     \end{multlined}
\end{align*}
\tl{\UVal\nn} is the type that emulated \ulc terms have when
back-translated into
\pstlc. 
For every $n$, $\tl{\UVal\nn}$ is clearly a valid \pstlc type. At
non-zero levels, the type $\tl{\UVal\indexx{n+1}}$ is a disjunct sum
of base values (the second occurrence of $\Unit$ and $\Bool$),
products and coproducts of $\tl{\UVal\nn}$s and functions mapping a
$\tl{\UVal\nn}$ to a $\tl{\UVal\nn}$. All of these cases are used to
emulate a corresponding \ulc value. Additionally, at every level
including $n=0$, the type $\tl{\UVal\nn}$ contains a $\Unit$ case which is
needed to represent an arbitrary \ulc value in cases where the
precision of the approximate emulation is insufficient to provide more
information. Note that the two occurrences of $\Unit$ in the
definition of $\tl{\UVal\indexx{n+1}}$ are not a typo. The first is
used for imprecisely representing arbitrary \ulc terms while the
second accurately represents \ulc $\ul{\unitv}$ values.

\begin{figure}
  \begin{minipage}{0.5\linewidth}
  \begin{align*}
    \tl{\inDV{\mr{unk};n}} &: \tl{\UVal\indexx{n+1}}\\
    \tl{\inDV{\Unit;n}} &: \tl{\Unit \ra \UVal\indexx{n+1}}\\
    \tl{\inDV{\Bool;n}} &: \tl{\Bool \ra \UVal\indexx{n+1}}\\
    \tl{\inDV{\times;n}} &: \tl{(\UVal\nn \times \UVal\nn) \ra \UVal\indexx{n+1}}\\
    \tl{\inDV{\uplus;n}} &: \tl{(\UVal\nn \uplus \UVal\nn) \ra \UVal\indexx{n+1}}\\
    \tl{\inDV{\ra;n}} &: \tl{(\UVal\nn \ra \UVal\nn) \ra \UVal\indexx{n+1}}
  \end{align*}
  \end{minipage}%
  \begin{minipage}{0.3\linewidth}
  \begin{align*}
    \tl{\unkUVal\nn} &: \tl{\UVal\nn}\\
    \tl{\unkUVal\indexx{0}} &\mydefsym \tl{\unitv}\\
    \tl{\unkUVal\indexx{n+1}} &\mydefsym \tl{\inDV{\mr{unk};n}}
  \end{align*}
  \end{minipage}
  \begin{align*}
    \tl{\myomega_\tau} &: \tl{\tau} \\
    \tl{\myomega_\tau} &\mydefsym \tl{\fix{\Unit \to \tau}\ (\lambda x:\Unit \to \tau\ldotp x)\ \unitv}
  \end{align*}
  \begin{align*}
  \tl{\caseDV{\Unit;n}} &: \tl{\UVal\indexx{n+1} \ra \Unit}\\
    \tl{\caseDV{\Bool;n}} &: \tl{\UVal\indexx{n+1} \ra \Bool}\\
    \tl{\caseDV{\times;n}} &: \tl{\UVal\indexx{n+1} \ra (\UVal\nn\times\UVal\nn)}\\
    \tl{\caseDV{\uplus;n}} &: \tl{\UVal\indexx{n+1} \ra (\UVal\nn\uplus\UVal\nn)}\\
    \tl{\caseDV{\ra;n}} &: \tl{\UVal\indexx{n+1} \ra \UVal\nn\ra\UVal\nn}\\
    \tl{\caseDV{\Unit;n}} &\mydefsym \tl{\lambda x:\UVal\indexx{n+1}\ldotp \case\ x\ \of\ \{ \inDV{\Unit;n}\ x \mapsto x; \_ \mapsto \myomega \}}\\
    \tl{\caseDV{\Bool;n}} &\mydefsym \tl{\lambda x:\UVal\indexx{n+1}\ldotp \case\ x\ \of\ \{ \inDV{\Bool;n}\ x \mapsto x; \_ \mapsto \myomega \}}\\
    \tl{\caseDV{\times;n}} &\mydefsym \tl{\lambda x:\UVal\indexx{n+1}\ldotp \case\ x\ \of\ \{ \inDV{\times;n}\ x \mapsto x; \_ \mapsto \myomega \}}\\
    \tl{\caseDV{\uplus;n}} &\mydefsym \tl{\lambda x:\UVal\indexx{n+1}\ldotp \case\ x\ \of\ \{ \inDV{\uplus;n}\ x \mapsto x; \_ \mapsto \myomega \}}\\
    \tl{\caseDV{\ra;n}} &\mydefsym
                            \tl{\lambda x:\UVal\indexx{n+1}\ldotp \lambda y: \UVal\nn\ldotp \case\ x\ \of}
                            \tl{\{ \inDV{\ra;n}\ z \mapsto z~y; \_ \mapsto \myomega \}}
  \end{align*}
    \caption{Basic tools for working with $\tl{\UVal\nn}$. The subscript of \myomega{} is omitted when it is clear from the context.}
    \label{fig:uval-basic-tools}
\end{figure}
To work with $\tl{\UVal\nn}$ values, we need basic tools for dealing with sum types: tag injections and case extractions (\cref{fig:uval-basic-tools}).
Functions $\tl{\inDV{\mr{unk};n}}$,
$\tl{\inDV{\Unit;n}}$, $\tl{\inDV{\Bool;n}}$, $\tl{\inDV{\times;n}}$,
$\tl{\inDV{\uplus;n}}$, $\tl{\inDV{\to;n}}$ are convenient names for
nested applications of coproduct injection functions for the nested
coproduct in the definition of $\tl{\UVal\indexx{n+1}}$. The term
$\tl{\unkUVal\nn}$ produces either the single value of $\tl{\UVal\indexx{0}}$
or uses $\tl{\inDV{\mr{unk};n}}$ to produce a $\tl{\UVal\indexx{n+1}}$ value
representing a 0-precision approximate back-translation of an
arbitrary untyped term. For using $\tl{\UVal\nn}$ values, we define functions $\tl{\caseDV{\Unit;n}}$,
$\tl{\caseDV{\Bool;n}}$, $\tl{\caseDV{\times;n}}$,
$\tl{\caseDV{\uplus;n}}$, $\tl{\caseDV{\to;n}}$ using a
somewhat liberal pattern matching syntax that can be easily desugared to
nested $\tl{\case}$ expressions.
The functions are lambdas that
inspect their $\tl{\UVal\indexx{n+1}}$ argument and return the contained
value if it is in the appropriate branch of the coproduct, or diverge
otherwise.
To achieve divergence, we use a term $\tl{\myomega_\tau}$
constructed using $\tl{\fix}$. We simply write $\tl{\myomega}$
when the type $\tltau$ can be inferred from the context.

\subsection{\ulc Values vs. $\UVal$} \label{sec:ulc-values-vs}

To make the correspondence between a \ulc term and its emulation in $\tl{\UVal\nn}$ more exact, this section fills in the definition of $\valrel{\EmulDV\np}_\genlogrel$, the missing piece of the logical relations of \cref{fig:logrels}.
Intuitively, the previously presented cases of the logical relations define the relation between a \pstlc term and its compilation.
The $\valrel{\EmulDV\np}_\genlogrel$ case defines the relation between a \ulc term and its $\tl{\UVal\nn}$-typed back-translation, as motivated in \cref{ex:explain-emuldv}.
This relation depends on the index $n$ of type $\tl{\UVal\nn}$ and additionally on a parameter $p::=\precise \mid \imprecise$, that is explained below.

\begin{exa}[The need for \tl{\EmulDV}]\label{ex:explain-emuldv}
Consider the term $\tlt\equiv\tl{\inDV{\Bool;1}~\truev}$.
Since $\tl{\UVal\nn}$ is a sum type, according to the definition of $\valrel{\tau\uplus\tau'}$, it can be related only to terms that have the same tag.
However, for the back-translation we do not want this, we want that term to be related to the \ult term that \tlt approximates (in this case, \ul{\truev}). 
Type $\tl{\EmulDV\indexx{n;p}}$ serves the purpose of expressing this $\ult$-emulates-$\tlt$ relation (as opposed to the $\ult$-is-the-compilation-of-$\tlt$ relation expressed by the other types).
In other words, $\tl{\inDV{\Bool;1}~\truev}$ and $\ul{\truev}$ will be related at pseudo-type $\tl{\EmulDV\indexx{2;p}}$.
\end{exa}

Before explaining the definition of the logical relations for $\tl{\EmulDV\np}$, we should elaborate on the approximateness of the correspondence.
\begin{exa}[Approximate values \tl{\unkUVal\nn}]\label{ex:}
Consider the $\tl{\UVal\indexx{6}}$ value
\[\tl{\inDV{\times;5}~\pair{\inDV{\uplus;4}~(\inl{\unkUVal\indexx{4}}),
    \unkUVal\indexx{5}}}\]
This value might be used by the approximate back-translation to represent the \ulc term $\ul{\pair{\inl{\pair{\unit,\truev}},\lam{x}{x}}}$.
Our $\valrel{\EmulDV\np}_\square$ specification will enforce that terms of the form $\tl{\inDV{\times;n}~\pair{\cdot,\cdot}}$ or $\tl{\inDV{\uplus;n}~(\inl{\cdot})}$ represent the corresponding \ulc constructs, but terms $\tl{\unkUVal\indexx{4}}$ and $\tl{\unkUVal\indexx{5}}$ can represent arbitrary terms (in this case: a pair of base values and a lambda).
\end{exa}
The limited size of the type $\tl{\UVal\nn}$ sometimes forces us to resort to $\tl{\unkUVal\nn}$ values in the back-translation, making it approximate.
However, $\valrel{\EmulDV\np}_\square$ does not allow these $\tl{\unkUVal\nn}$ values to occur just anywhere, because they could compromise the required precision of our approximate back-translation.

In fact, $\valrel{\EmulDV\np}_\square$ provides two different specifications for the occurrence of $\tl{\unkUVal\nn}$, depending on the value of $p$.
The case where $p = \imprecise$ is used when we are proving $\tl{\ef{\upc}\nn} \lesssim \upc$, which means roughly that termination of $\tl{\ef{\upc}\nn}$ in \emph{any} number of steps implies termination of $\upc$.
In this case, $\valrel{\EmulDV\np}_\square$ allows $\tl{\unkUVal\nn}$ values to occur everywhere in a back-translation term, and they can correspond to arbitrary \ulc terms.
These mild requirements on the correspondence of \ulc terms place a large burden on the code in a back-translation $\tl{\ef{\upc}\nn}$.
This code must be able to deal with $\tl{\unkUVal\nn}$ values and produce behaviour for them that approximates the behaviour of $\upc$ for the arbitrary values that the $\tl{\unkUVal\nn}$s correspond with.
Luckily, when we are proving $\tl{\ef{\upc}\nn} \lesssim \upc$, we can achieve this by simply making all the functions in our back-translation diverge whenever they try to use a $\tl{\UVal\nn}$ value that happens to be an $\tl{\unkUVal\nn}$.
This is sufficient because the approximation $\tl{\ef{\upc}\nn} \lesssim \upc$ trivially holds when $\tl{\ef{\upc}\nn}$ diverges: it essentially only requires that $\upc$ terminates whenever $\tl{\ef{\upc}\nn}$ does, but nothing needs to be shown when the latter diverges.

\begin{exa}[Relatedness with \imprecise]\label{ex:imprecise}
Consider the term $\tlt\equiv\tl{\inDV{\times;42}~\pair{\unkUVal\indexx{42},\unkUVal\indexx{42}}}$.
This term will be related to $\ul{\pair{t_1,t_2}}$ at pseudo-type $\tl{\EmulDV\indexx{43;\imprecise}}$ for any terms $\ul{t_1}$ and $\ul{t_2}$ and in any world.
\end{exa}

The case when $p = \precise$ specifies where values $\tl{\unkUVal\nn}$ are allowed when we are proving that $\tl{\ef{\upc}\nn} \gtrsim\nn \upc$, meaning roughly that termination of $\upc$ in less than $n$ steps implies termination of $\tl{\ef{\upc}\nn}$.
In this case, the requirements on the back-translation correspondence are significantly stronger: $\tl{\unkUVal\nn}$ is simply ruled out by the definition of $\valrel{\EmulDV\np}_\square$.
That does not mean, however, that $\tl{\unkUVal\nn}$ cannot occur inside related terms, rather that $\tl{\unkUVal\nn}$ can only occur at depths that cannot be reached using the number of steps in the world.
\begin{exa}[Relatedness with \precise]\label{ex:precise}
Consider again the term $\tl{t} \isdef \tl{\inDV{\times;42}~\pair{\unkUVal\indexx{42},\unkUVal\indexx{42}}}$.
This term will still
be related by $\tl{\EmulDV\indexx{43;\precise}}$ to $\ul{t} \isdef \ul{\pair{t_1,t_2}}$ for any terms $\ul{t_1}$ and $\ul{t_2}$, but only in worlds $\WW$ such that $\CWstepsfun{\WW} = 0$. More precisely, our specification will state that $(\WW,\tl{t},\ul{t}) \in \valrel{\EmulDV\indexx{43;\precise}}_\square$ iff
\begin{equation*}
  (\WW,\tl{\pair{\unkUVal\indexx{42},\unkUVal\indexx{42}}},\ul{\pair{t_1,t_2}}) \in \valrel{\EmulDV\indexx{42;\precise}\times \EmulDV\indexx{42;\precise}}_\square\text{.}
\end{equation*}
By the definition in \cref{fig:logrels}, this requires in turn that $(\WW,\tl{\unkUVal\indexx{42}},\ul{t_1})$ and $(\WW,\tl{\unkUVal\indexx{42}},\ul{t_2})$ are in $\later \valrel{\EmulDV\indexx{42,\precise}}_\square$.
However if $\CWstepsfun{\WW} = 0$, then this is vacuously true by definition of the $\later$ operator, independent of the requirements of $\valrel{\EmulDV\indexx{42,\precise}}_\square$.
\end{exa}

Intuitively, it is sufficient to only forbid $\tl{\unkUVal\nn}$ at depths lower than the number of steps left in the world because we are proving $\tl{\ef{\upc}\nn} \gtrsim\nn \upc$ (emphasis on the index $_{\ms{n}}$ of $\gtrsim\nn$).
So, if $\upc$ terminates \emph{in less than $n$ steps}, then the evaluation of $\upc$ cannot have used values that are deeper than level $n$ in any $\tl{\UVal\nn}$.
The corresponding execution of $\tl{\ef{\upc}\nn}$ will also not have had a chance to encounter the $\tl{\unkUVal\nn}$s.
Therefore, the executions must have behaved identically.

\begin{figure}
  \centering
\begin{align*}
  \valrel{\EmulDV\indexx{0;p}}_\genlogrel \mydefsym& \left\{(\WW,\tl{v},\ul{v}) \relmiddle| \tl{v} = \tl{\unitv} \text{ and } \ms{p} = \imprecise\right\}\\
  \valrel{\EmulDV\indexx{n+1;p}}_\genlogrel \mydefsym&\left\{(\WW,\tl{v},\ul{v}) \relmiddle|
  \begin{aligned}
                                &\tl{v} \in \tllsfun{oftype}{\UVal\indexx{n+1}} \text{ and one of the following holds: }\\
                                &\qquad
                                \left\{\begin{aligned}
                                  &\tl{v} = \tl{\inDV{\mr{unk};n}} \text{ and } \ms{p} = \imprecise\\
                                  &\exists \tl{v'}\ldotp \tl{v} = \tl{\inDV{\tl{\Unit};n}~v'} \text{ and } (\WW,\tl{v'},\ul{v}) \in \valrel{\Unit}_\genlogrel\\
                                  &\exists \tl{v'}\ldotp \tl{v} = \tl{\inDV{\tl{\Bool};n}~v'} \text{ and } (\WW,\tl{v'},\ul{v}) \in \valrel{\Bool}_\genlogrel\\
                                  &
                                  \begin{multlined}
                                    \exists \tl{v'}\ldotp \tl{v} = \tl{\inDV{\times;n}~v'} \text{ and } \\
                                    (\WW,\tl{v'},\ul{v}) \in \valrel{\EmulDV\np \times \EmulDV\np}_\genlogrel
                                  \end{multlined}\\
                                  &
                                  \begin{multlined}
                                    \exists \tl{v'}\ldotp \tl{v} = \tl{\inDV{\uplus;n}~v'} \text{ and } \\
                                    (\WW,\tl{v'},\ul{v}) \in \valrel{\EmulDV\np \uplus \EmulDV\np}_\genlogrel
                                  \end{multlined}\\
                                  &
                                  \begin{multlined}
                                    \exists \tl{v'}\ldotp \tl{v} = \tl{\inDV{\to;n}~v'} \text{ and }\\
                                    (\WW,\tl{v'},\ul{v}) \in \valrel{\EmulDV\np \ra \EmulDV\np}_\genlogrel
                                  \end{multlined}\\
                                \end{aligned}\right.
                              \end{aligned}
                              \right\}
\end{align*}
  \caption{Specifying the relation between \ulc values and their emulation in $\valrel{\EmulDV\np}_\genlogrel$.}  \label{fig:emuldv}
\end{figure}

With this approximation aspect explained, \cref{fig:emuldv} presents the definition of $\valrel{\EmulDV\np}_\genlogrel$.
For relating terms $\tl{v}$ and $\ul{v}$ in a world $\WW$, the definition requires that $\tl{v}$ has the right type and that $p = \imprecise$ if $\tl{v}$ is $\tl{\unkUVal\nn}$.
Additionally, the structure of the \pstlc term stripped of its \tl{\UVal\nn} tag and the structure of the \ulc term must coincide.
Formally, this is expressed by the following conditions: $(\WW,\tl{v'},\ul{v})$ are in $\valrel{\batype}_\genlogrel$ (recall that \tl{\batype} are ground types), $\valrel{\EmulDV\np \times \EmulDV\np}_\genlogrel$, $\valrel{\EmulDV\np \uplus \EmulDV\np}_\genlogrel$ or $\valrel{\EmulDV\np \to \EmulDV\np}_\genlogrel$ if $\tl{v} = \tl{\inDV{\batype;n}~\tl{v'}}$, $\tl{v} = \tl{\inDV{\times;n}~\tl{v'}}$, $\tl{v} = \tl{\inDV{\uplus;n}~\tl{v'}}$ or $\tl{v} = \tl{\inDV{\to;n}~\tl{v'}}$ respectively.

In addition to $\tl{\EmulDV\np}$, we still need to define two helper functions (\cref{fig:helper-emul}).
The first, \replemul{\cdot}, was left open in \cref{fig:logrels}.
It re-maps all variables of a $\tlGamma$ that are of type $\tl{\EmulDV\np}$ to type $\tl{\UVal\nn}$.
A second function, $\toemul{\cdot}{n;p}$, turns an untyped $\ulGamma$ into one where all variables are mapped to $\tl{\EmulDV\np}$.
\begin{figure}[t]
  \centering
  \begin{align*}
    \toemul{\ule}{n;p} &= \tle \\
    \toemul{\ulGamma, \ul{x}}{n;p} &= \toemul{\ulGamma}{n;p}, (\tl{x}:\tl{\EmulDV\np})
      \\
    \replemul{\tle} &= \tle \\
    \replemul{\tlGamma, (\tl{x}:\tlhat{\tltau})} &= \replemul{\tlGamma}, (\tl{x}:\replemul{\tlhat{\tltau}})
  \end{align*}
  \begin{align*}
    \replemul{\tl{\hat{\tau}}\times\tl{\hat{\tau'}}} &= \tl{\replemul{\tl{\hat{\tau}}}\times\replemul{\tl{\hat{\tau'}}}}\\
    \replemul{\tl{\hat{\tau}}\uplus\tl{\hat{\tau'}}} &= \tl{\replemul{\tl{\hat{\tau}}}\uplus\replemul{\tl{\hat{\tau'}}}}\\
    \replemul{\tl{\hat{\tau}}\to\tl{\hat{\tau'}}} &= \tl{\replemul{\tl{\hat{\tau}}}\to\replemul{\tl{\hat{\tau'}}}}\\
      \replemul{\tl{\EmulDV\indexx{n;p}}} &= \tl{\UVal\nn}\\
  \replemul{\tl{\Bool}} &= \tl{\Bool} \\
    \replemul{\tl{\Unit}} &= \tl{\Unit}
  \end{align*}%
  \caption{Helper functions for $\tl{\EmulDV\np}$.}  \label{fig:helper-emul}
\end{figure}

The adequacy property of the logical relations (\cref{lem:adequacy}) holds for the complete definition of the logical relations, including the definition for $\valrel{\EmulDV\np}$.

\subsection{Upgrading and Downgrading Values}\label{sec:updown}
\begin{figure}[t]
  \centering
  \begin{align*}
  \tl{\downgrade\indexx{n;d}} &: \tl{\UVal\indexx{n+d} \ra \UVal\nn}\\
  \tl{\downgrade\indexx{0;d}} &\mydefsym \tl{\lam{v : \UVal\indexx{d}}{\unkUVal\indexx{0}}}\\
  \tl{\downgrade\indexx{n+1;d}} &\mydefsym \tl{\lambda x : \UVal\indexx{n+d+1}\ldotp \mr{case}\ x\ \mr{of}\ }\\
    &\hspace{1.3cm}\tl{\left|
    \begin{aligned}
              \inDV{\mr{unk};n+d} &\mapsto \inDV{\mr{unk};n} \\
              \inDV{\Unit;n+d}\ \tl{y} &\mapsto \inDV{\Unit;n}\ \tl{y} \\
              \inDV{\Bool;n+d}\ \tl{y} &\mapsto \inDV{\Bool;n}\ \tl{y} \\
              \inDV{\times;n+d}\ \tl{y} &\mapsto
                \inDV{\times;n}\ \tl{ \pair{\downgrade\indexx{n;d}\ \projone{y}, \downgrade\indexx{n;d}\ \projtwo{y}}} \\
              \inDV{\uplus;n+d}\ \tl{y} &\mapsto
                \inDV{\uplus;n}\ \case~\tl{y}~\of~ \left| \begin{aligned}
                    &\tl{\inl{x}\mapsto \inl{(\downgrade\indexx{n;d}\ x)}}\\
                    &\tl{\inr{x} \mapsto \inr{(\downgrade\indexx{n;d}\ x)}}
                  \end{aligned}\right.
                \\
                \inDV{\ra;n+d}\ \tl{y} &\mapsto
                  \inDV{\ra;n}\ \tl{(\lambda z: \UVal\nn\ldotp \downgrade\indexx{n;d}}\tl{(y\ (\upgrade\indexx{n;d}\ z)))}
    \end{aligned}\right.}
  \end{align*}

  \begin{align*}
  \tl{\upgrade\indexx{n;d}} &: \tl{\UVal\nn \ra \UVal\indexx{n+d}}\\
  \tl{\upgrade\indexx{0;d}} &\mydefsym \tl{\lambda x : \UVal\indexx{0}\ldotp \unkUVal\indexx{d}}\\
  \tl{\upgrade\indexx{n+1;d}} &\mydefsym \tl{\lambda x : \UVal\indexx{n+1}\ldotp \mr{case}\ x\ \mr{of}\ }\\
    &\hspace{1.5cm}\tl{\left|
    \begin{aligned}
              \inDV{\mr{unk};n} &\mapsto \inDV{\mr{unk};n+d} ;\\
              \inDV{\Unit;n}\ \tl{y} &\mapsto \inDV{\Unit;n+d}\ \tl{y} ;\\
              \inDV{\Bool;n}\ \tl{y} &\mapsto \inDV{\Bool;n+d}\ \tl{y} ;\\
              \inDV{\times;n}\ \tl{y} &\mapsto
                \inDV{\times;n+d}\ \tl{\pair{ \upgrade\indexx{n;d}\ \tl{\projone{y}}, \upgrade\indexx{n;d}\ \tl{\projtwo{y}} }}
              \\
              \inDV{\uplus;n}\ \tl{y} &\mapsto
                \inDV{\uplus;n+d}\ \case~y~\of~ \left|\begin{aligned}
                    &\tl{\inl{x}\mapsto \inl{(\upgrade\indexx{n;d}\ x)}}\\
                    &\tl{\inr{x} \mapsto \inr{(\upgrade\indexx{n;d}\ x)}}
                  \end{aligned}\right.
              \\
              \inDV{\ra;n}\ \tl{y} &\mapsto
                \inDV{\ra;n+d}\ \tl{(\lambda z: \UVal\nn\ldotp \upgrade\indexx{n;d}}\tl{(y\ (\downgrade\indexx{n;d}\ z)))}
    \end{aligned}\right.}
  \end{align*}

  \caption{Upgrade and downgrade for $\tl{\UVal\nn}$.}
  \label{fig:upgrade-downgrade}
\end{figure}
\Cref{fig:upgrade-downgrade} defines the functions
$\tl{\downgrade\indexx{n;d}} : \tl{\UVal\indexx{n+d} \to \UVal\nn}$ and
$\tl{\upgrade\indexx{n;d}} : \tl{\UVal\nn \to \UVal\indexx{n+d}}$ (by induction on
$n$) that we talked about before. Most cases simply work structurally over the type, but some are
more interesting. There is a contravariance in the cases for
function values in both $\tl{\downgrade\indexx{n;d}}$ and
$\tl{\upgrade\indexx{n;d}}$: a function $\tl{\UVal\nn\to\UVal\nn}$ is turned
into a function of type $\tl{\UVal\indexx{n+d}\to\UVal\indexx{n+d}}$ by
constructing a wrapper that downgrades the argument and upgrades the
result and vice versa. Unknown values are always mapped to unknown values, but
additionally, the case for $\tl{\downgrade\indexx{n;d}}$ when $n = 0$ will
throw away the information contained in its argument of type
$\tl{\UVal\indexx{d}}$ and simply returns the single unknown value in
$\tl{\UVal\indexx{0}}$. Note that $\tl{\downgrade\indexx{n;d}}$ and
$\tl{\upgrade\indexx{n;d}}$ are not inverse functions, since
$\tl{\downgrade\indexx{n;d}}$ throws away information that was previously
there.
Informally, while $\tlt \sim \tl{\downgrade\indexx{n;d}~(\upgrade\indexx{n;d}~t)}$, the reverse ($\tlt \sim \tl{\upgrade\indexx{n;d}~(\downgrade\indexx{n;d}~t)}$) is not true, since applying downgrade first reduces precision.

\begin{exa}[Downgrading terms]\label{ex:downgrade}
Suppose that we want to emulate a \ulc term $\ul{\lam{x}{\pair{x,x}}}$
in $\tl{\UVal\indexx{\cdot}}$ for a sufficiently-large $n$. We would expect roughly the following \pstlc term:
\begin{equation*}
 \tl{\inDV{\to;n-1}~(\lam{x : \UVal\indexx{n-1}}{\inDV{\times;n-2}~\pair{x,x}})}
\end{equation*}
Indices $n-1$ and $n-2$ of the $\tl{\UVal\nn}$ constructors are
imposed by the well-typedness constraints.
However, even this is not
enough to guarantee well-typedness.
With a closer inspection, the variable $\tl{x}$ of type $\tl{\UVal\indexx{n-1}}$ is used where a term of type $\tl{\UVal\indexx{n-2}}$ is required (it is inside a pair tagged with $\tl{\inDV{\times;n-2}}$).
This is a problem of type safety, not precision of approximation. Since \tl{x} appears inside a pair, inspecting \tl{x} for any number of steps requires at least one additional step to first project it out of the pair.
In other words, for the pair to be a precise approximation up to $\leq n-1$ steps, \tl{x} needs only to be precise up to $n-2$ steps.
It is then safe to throw away one level of precision and downgrade $\tl{x}$ from type $\tl{\UVal\indexx{n-1}}$ to $\tl{\UVal\indexx{n-2}}$.
\end{exa}

We will use the function \downgrade{} for the situation of
\cref{ex:downgrade} and similar ones in the next sections. In dual
situations we will need to upgrade terms from type $\tl{\UVal\nn}$ to
$\tl{\UVal\indexx{n+d}}$. This will neither increase precision of the
approximation, nor decrease it.

The correctness property for downgrade and upgrade is stated in the following lemma.
\begin{lem}[Compatibility lemma for $\tl{\upgrade\indexx{n;d}}$ and
  $\tl{\downgrade\indexx{n;d}}$]
  \label{lem:downgrade-upgrade-compat}
  Suppose that 
  \begin{itemize}
    \item either ($\ms{n} < \ms{m}$ and $\ms{p} = \precise$) 
    \item or ($\genlogrel = \lesssim$ and $\ms{p} = \imprecise$),
  \end{itemize}
  then
  \begin{itemize}
  \item If
    $\tlGamma \vdash \tl{t} \arbsim\nn \ul{t} : \tl{\EmulDV\indexx{m+d;p}}$, then
    $\tlGamma \vdash \tl{\downgrade\indexx{m;d}~t} \arbsim\nn \ul{t} :
    \tl{\EmulDV\indexx{m;p}}$.
  \item If $\tlGamma \vdash \tl{t} \arbsim\nn \ul{t} : \tl{\EmulDV\indexx{m;p}}$,
    then
    $\tlGamma \vdash \tl{\upgrade\indexx{m;d}~t} \arbsim\nn \ul{t} :
    \tl{\EmulDV\indexx{m+d;p}}$.
  \end{itemize}
\end{lem}
This lemma covers both situations that we discussed previously. It
requires that either $n < m$ (so that the results only hold in worlds
$\WW$ with $\CWstepsfun{\WW} \leq n < m$), in which case
$p = \precise$, or $\genlogrel = \lesssim$ and $p = \imprecise$. If that
is the case, the lemma says that if a term $\tl{t}$ is related to
$\ul{t}$ by $\tl{\EmulDV\indexx{m+d;p}}$ (or $\tl{\EmulDV\indexx{m;p}}$) then it stays
related to $\ul{t}$ after upgrading or downgrading.

\subsection{Emulation}\label{sec:emulate}

\begin{figure}
  \centering
  \noindent%
  \begin{align*}
    \tl{\emulate\nn}(\ul{t}) &: \tl{\UVal\nn}\\
    \tl{\emulate\nn}(\ul{\unitv}) &\mydefsym \tl{\downgrade\indexx{n;1}~(\inDV{\tl{\Unit};n}\ \unitv)}\\
    \tl{\emulate\nn}(\ul{\truev}) &\mydefsym \tl{\downgrade\indexx{n;1}~(\inDV{\tl{\Bool};n}\ \truev)}\\
    \tl{\emulate\nn}(\ul{\falsev}) &\mydefsym \tl{\downgrade\indexx{n;1}~(\inDV{\tl{\Bool};n}\ \falsev)}\\
    \tl{\emulate\nn}(\ul{x}) &\mydefsym \tl{x}\\
    \tl{\emulate\nn}(\ul{\lambda x\ldotp t}) &\mydefsym \tl{\downgrade\indexx{n;1}~
                                                (\inDV{\ra;n}\ (\lambda x: \UVal\nn\ldotp \emulate\nn(\ul{t})))
                                                }\\
    \tl{\emulate\nn}(\ul{t_1\ t_2}) &\mydefsym \tl{
                                       \caseDV{\ra;n}\ (\upgrade\indexx{n;1}~(\emulate\nn(\ul{t_1})))\ \emulate\nn(\ul{t_2})
                                       }\\
    \tl{\emulate\nn}(\ul{\langle t_1, t_2\rangle}) &\mydefsym \tl{
                                                      \downgrade\indexx{n;1}~(\inDV{\times;n}\ \langle \emulate\nn(\ul{t_1}), \emulate\nn(\ul{t_2}) \rangle)
                                                      }\\
    \tl{\emulate\nn}(\ul{\inl{t}}) &\mydefsym \tl{\downgrade\indexx{n;1}~(\inDV{\uplus;n}\ (\inl{\emulate\nn(\ul{t})}))}\\
    \tl{\emulate\nn}(\ul{\inr{t}}) &\mydefsym \tl{\downgrade\indexx{n;1}~(\inDV{\uplus;n}\ (\inr{\emulate\nn(\ul{t})}))}\\
    \tl{\emulate\nn}(\ul{t.1}) &\mydefsym \tl{
                                  (\caseDV{\times;n}\ (\upgrade\indexx{n;1}~(\emulate\nn(\ul{t})))).1 }\\
    \tl{\emulate\nn}(\ul{t.2}) &\mydefsym \tl{
                                  (\caseDV{\times;n}\ (\upgrade\indexx{n;1}~(\emulate\nn(\ul{t})))).2 }\\
    \tl{\emulate\nn}(\ul{t;t'}) &\mydefsym \tl{(\caseDV{\Unit;n}\ (\upgrade\indexx{n;1}(\emulate\nn(\ul{t}))));\emulate\nn(\ul{t'})}
    \\
    \tl{\emulate\nn}(\ul{\wrong}) &\mydefsym \tl{\myomega}
  \end{align*}
  \vspace{-.7cm}
  \begin{multline*}
    \tl{\emulate\nn}(\ul{\case~t_1~\of~} \ul{\inl{x}\mapsto t_2 \mid \inr{x} \mapsto t_3})
    \mydefsym\\
    \tl{
        \case~(\caseDV{\uplus;n}\ (\upgrade\indexx{n;1}~(\emulate\nn(\ul{t_1}))))~\of
        \left|
          \begin{aligned}
            &\inl{x}\mapsto \emulate\nn(\ul{t_2})\\
            &\inr{x}\mapsto \emulate\nn(\ul{t_3})
          \end{aligned} \right.}
  \end{multline*}
  \vspace{-.4cm}
  \begin{multline*}
    \tl{\emulate\nn}(\ul{\ifte{t}{t_1}{t_2}})
    \mydefsym \\
      \tl{\ifte{(\caseDV{\Bool;n}~(\upgrade\indexx{n;1}~(\emulate\nn{\ult})))}{\emulate\nn(\ul{t_1})}{}}
      \tl{\emulate\nn(\ul{t_2})}
    \end{multline*}
  \caption{Emulating \ulc terms in $\tl{\UVal\nn}$.}
  \label{fig:emulate}
\end{figure}
Having defined $\tl{\downgrade}$ and $\tl{\upgrade}$, \cref{fig:emulate}
defines the $\tl{\emulate\nn}$ function.
That function maps arbitrary \ulc terms to
their approximate back-translation: \pstlc terms of type
$\tl{\UVal\nn}$.
\tl{\emulate\nn} is defined by induction on~$\ul{t}$.
The different cases follow the same pattern: every term $\ul{t}$ is mapped to a \pstlc term constructed recursively from the emulation of sub-terms, producing and consuming $\tl{\UVal\nn}$ terms wherever $\ul{t}$ works with untyped terms.
Additionally, the definitions use
$\tl{\upgrade\indexx{n;1}}$ and $\tl{\downgrade\indexx{n;1}}$ to make the
resulting term type-check, as explained in \cref{ex:downgrade}.
For
example, the case for pairs applies $\tl{\inDV{\times;n}}$ to a pair
constructed from the emulations of its components.
Since this produces a
$\tl{\UVal\indexx{n+1}}$, $\tl{\downgrade\indexx{n;1}}$ is used to downgrade this to a $\tl{\UVal\nn}$ term.
Finally, the untyped term $\ul{\wrong}$ is back-translated to a divergent term.

The back-translation produced by $\tl{\emulate\nn}$ is necessarily
approximate, as the type $\tl{\UVal\nn}$ is not large enough for
back-translating arbitrary terms. Inaccuracies in the
back-translation are introduced in the calls to
$\tl{\downgrade\indexx{n;1}}$ in several of the cases.
The
approximation is accurate enough for the following lemma to hold.
\begin{lem}[Emulate relates at $\tl{\EmulDV}$]\label{lem:emulate-works}
  If $\ulGamma \vdash \ult$, and if 
  \begin{itemize}
    \item either ($\ms{m} > \ms{n}$ and $\ms{p} = \precise$) 
    \item or ($\genlogrel = \lesssim$ and $\ms{p} = \imprecise$),
  \end{itemize} then we have that
  $\toemul{\ulGamma}{m;p}\vdash \tl{\emulate\indexx{m}}(\ul{t}) \arbsim\nn \ul{t}: \tl{\EmulDV\indexx{m;p}}$.
\end{lem}

Like \cref{lem:downgrade-upgrade-compat}, \cref{lem:emulate-works} requires that
either $n < m$ (so that the results only hold in worlds $\WW$ with
$\CWstepsfun{\WW} \leq n < m$), in which case $p = \precise$, or
$\genlogrel = \lesssim$ and $p = \imprecise$.
This again covers what we
need for the two logical approximations of $\tl{\ef{\upc}\nn}$ in
\cref{fig:proving-compiler-security-approx}.
The lemma states that
the back-translation of any well-scoped term is related to the term by
$\tl{\EmulDV\indexx{m;p}}$, as intended.

An analogous result holds for contexts.
\begin{lem}[Emulate relates contexts at $\tl{\EmulDV}$]
  \label{lem:emulate-works-ctx}
  If $\vdash \upc : \ul{\Gamma'} \ra \ulGamma$, and if 
  \begin{itemize}
    \item either ($m > n$ and $p = \precise$) 
    \item or
  ($\square = \lesssim$ and $p = \imprecise$),
  \end{itemize} then
  $\vdash \tl{\emulate\indexx{m}}(\upc) \arbsim\nn \upc: \toemul{\ul{\Gamma'}}{m;p},\tl{\EmulDV\indexx{m;p}} \to \toemul{\ulGamma}{m;p}, \tl{\EmulDV\indexx{m;p}}$
\end{lem}

\subsection{Injection and Extraction of Terms}\label{sec:injextr}

\begin{figure}
  \centering
    \begin{align*}
      \tl{\extractf{\tau;n}} &: \tl{\UVal\nn \ra \tau}\\
      \tl{\extractf{\tau;0}} &\mydefsym \tl{\lambda x: \UVal_0\ldotp \myomega}\\
      \tl{\extractf{\Unit;n+1}} &\mydefsym \tl{\lambda x: \UVal\indexx{n+1}\ldotp \caseDV{\Unit;n}\ x}\\
      \tl{\extractf{\Bool;n+1}} &\mydefsym \tl{\lambda x: \UVal\indexx{n+1}\ldotp \caseDV{\Bool;n}\ x}\\
    \tl{\extractf{\tau_1\ra \tau_2;n+1}} &\mydefsym
                                            \tl{\lambda x:\UVal\indexx{n+1}\ldotp \lambda y:\tau_1\ldotp \extractf{\tau_2;n}\ }
                                            \tl{(\caseDV{\ra;n}\ x\ (\inject{\tau_1;n}\ y))}\\
    \tl{\extractf{\tau_1\times \tau_2;n+1}} &\mydefsym
                                               \tl{\lambda x: \UVal\indexx{n+1}\ldotp \langle \extractf{\tau_1;n}\ \projone{(\caseDV{\times;n}\ x)},}
                                               \tl{ \extractf{\tau_2;n}\ \projtwo{(\caseDV{\times;n}\ x)} \rangle}
                                               \\
    \tl{\extractf{\tau_1\uplus \tau_2;n+1}} &\mydefsym
                                               \tl{\lambda x: \UVal\indexx{n+1}\ldotp \case~\caseDV{\uplus;n}\ x~\of}
                                               \tl{\left|\begin{aligned}
                                                 &\tl{\inl{y} \to \inl{(\extractf{\tau_1;n}~y)}}\\
                                                 &\tl{\inr{y} \to \inr{(\extractf{\tau_2;n}~y)}}
                                               \end{aligned}\right. }
    \end{align*}
    \begin{align*}
      \tl{\inject{\tau;n}} &: \tl{\tau \ra \UVal\nn}\\
      \tl{\inject{\tau;0}} &\mydefsym \tl{\lambda x: \tau\ldotp \myomega}\\
      \tl{\inject{\Unit;n+1}} &\mydefsym \tl{\lambda x: \Unit\ldotp \inDV{\Unit;n}\ x}\\
      \tl{\inject{\Bool;n+1}} &\mydefsym \tl{\lambda x: \Bool\ldotp \inDV{\Bool;n}\ x}\\
      \tl{\inject{\tau_1\ra \tau_2;n+1}} &\mydefsym
                                             \tl{\lambda x: \tau_1\ra\tau_2\ldotp \inDV{\ra;n}\ (\lambda x:\UVal\nn\ldotp}
                                             \tl{\inject{\tau_2;n}~(x~(\extractf{\tau_1;n}\ x)))}
                                             \\
      \tl{\inject{\tau_1\times \tau_2;n+1}} &\mydefsym
                                                \tl{\lambda x: \tau_1\times\tau_2\ldotp \inDV{\times;n}\langle \inject{\tau_1;n}\ \projone{x},}
                                                  \tl{\inject{\tau_2;n}\ \projtwo{x}\rangle}
                                                \\
      \tl{\inject{\tau_1\uplus \tau_2;n+1}} &\mydefsym
      \tl{\lambda x: \tau_1\uplus\tau_2\ldotp \inDV{\uplus;n}\ (\case~x~\of~}
      \tl{\left.\left|\begin{aligned}
            &\tl{\inl{y} \mapsto \inr{(\inject{\tau_1;n}\ y)}}\\
            &\tl{\inr{y} \mapsto \inr{(\inject{\tau_2;n}\ y)}}
          \end{aligned}\right. \right)}
    \end{align*}

  \caption{Injecting \pstlc values into $\tl{\UVal\nn}$.}
  \label{fig:inject_tau}
\end{figure}

One final thing is missing to construct a
back-translation $\tl{\ef{\progctx}\nn}$ of an untyped program context
$\upc$. While $\tl{\emulate\nn}(\upc)$ produces a \pstlc context that
expects a $\tl{\UVal\nn}$ value (just like $\upc$ expects an arbitrary
\ulc value), the back-translation should accept values of a given type
$\tltau$ (the type of the terms $\tl{t_1}$ and $\tl{t_2}$ that we are
compiling). To bridge this difference, \cref{fig:inject_tau} defines a
\pstlc function $\tl{\inject{\tau;n}}$ of type
$\tl{\tau \to \UVal\nn}$ which injects values of an arbitrary type
$\tltau$ into $\tl{\UVal\tlnn}$. We define it mutually recursively
with a dual function $\tl{\extractf{\tau;n}} : \tl{\UVal\nn \to \tau}$ which
is needed for contravariantly converting $\tl{\UVal\nn}$ arguments to
the appropriate type in the $\tl{\inject{\tau;n}}$ case for function types.

Generally, $\tl{\inject{\tau;n}}$ converts a value $\tl{v}$ of type
$\tltau$ to a value of type $\tl{\UVal\nn}$ that behaves like the
compilation $\comp{v}$. The cases for base values use the appropriate tagging and case functions
(e.g., $\tl{\inDV{\tl{\Unit};n}}$ and $\tl{\caseDV{\tl{\Bool};n}}$) to achieve this.
For pair and sum values, $\tl{\inject{\tau;n}}$ and
$\tl{\extractf{\tau;n}}$ simply recurse over the structure of the
values, respectively applying $\tl{\inDV{\times;n}}$,
$\tl{\inDV{\uplus;n}}$ and $\tl{\caseDV{\times;n}}$,
$\tl{\caseDV{\uplus;n}}$ to construct and destruct $\tl{\UVal\nn}$s of
a certain expected form. Note that when $\tl{\UVal\nn}$ values do
not have the form expected for type $\tltau$, then
$\tl{\extractf{\tau;n}}$ will diverge by definition of the
$\tl{\caseDV{\cdots;n}}$ functions. This divergence corresponds to the
$\ul{\wrong}$ that we get when an untyped context attempts to use
\ulc values as pairs, disjunct sum values or base values when those
values are of a different form.

For function types, $\tl{\inject{\tau;n}}$ and $\tl{\extractf{\tau;n}}$
produce lambdas that contravariantly extract and inject the argument
and covariantly inject and extract the result. Finally, when $n = 0$,
then the size of our type is insufficient for $\tl{\extractf{\tau;n}}$
and $\tl{\inject{\tau;n}}$ to accurately perform their intended
function. Luckily, to obtain the necessary precision of our
approximate back-translation, it is sufficient for them to simply
diverge in this case: they simply return $\tl{\myomega}$ terms of the
expected type.

For a value $\tl{v}$ of type $\tltau$, $\tl{\inject{\tau;n}}$ will
produce a value $\tl{\UVal\nn}$ that behaves as the compilation of \tlv,
$\comp{v}$. More precisely and more generally, the following lemma
states that if a term $\tl{t}$ is related to a term $\ul{t}$ at type
$\tltau$ (intuitively if $\ul{t}$ behaves as $\tl{t}$ when used in a way
that is valid according to type $\tltau$), then
$\tl{\inject{\tau;n}~t}$ behaves as the emulation of
$\ul{\prot{\tltau}~t}$. A dual result about $\tl{\extractf{\tau;n}}$
and $\ul{\conf{\tltau}}$ states (intuitively) that if a term $\tl{t}$
behaves as an emulation of value $\ul{t}$, then $\ul{\conf{\tltau}~t}$
will behave as $\tl{\extractf{\tau;n}~t}$ when used in ways that are
valid according to type $\tltau$.
\begin{lem}[Inject is protect and extract is confine]
  \label{lem:protect-inj-square}
  If  $\tl{\hat{\Gamma}} \vdash \tl{t} \arbsim\nn \ul{t} : \tltau$ and if
  \begin{itemize}
    \item either ($m \geq n$ and $p = \precise$)
    \item  or ($\square = \lesssim$ and $p = \imprecise$)
  \end{itemize} 
  then
    $\tl{\hat{\Gamma}} \vdash \tl{\inject{\tau;m}~t}\arbsim\nn\ul{\prot{\tau}~t} : \tl{\EmulDV\indexx{m;p}}.$
  
  \smallskip

  If  $\tl{\hat{\Gamma}} \vdash \tl{t} \arbsim\nn \ul{t} : \tl{\EmulDV\indexx{m;p}}$,  and if
  \begin{itemize}
    \item either ($m \geq n$ and $p = \precise$)
    \item  or ($\square = \lesssim$ and  $p = \imprecise$)
  \end{itemize} 
  then
    $\tl{\hat{\Gamma}} \vdash \tl{\extractf{\tau;m}~t} \arbsim\nn \ul{\conf{\tau}~t} :
        \tltau.$
\end{lem}

\begin{exa}
  Consider again \cref{ex:erase-is-not-fa}. We have that
  \begin{equation*}
    \tle \vdash \tl{\lam{ x: \Unit}{ x}} \arbsim \ul{\lam{ x}{ x}} : \tl{\Unit \to \Unit}\text{.}
  \end{equation*}
  $\tl{\lam{ x: \Unit}{ x}}$ behaves like
  $\ul{\lam{ x}{ x}}$, when the latter is used in ways that are valid
  for a value of type $\tl{\Unit\to\Unit}$.
  \Cref{lem:protect-inj-square} then yields:
  \begin{align*}
    \tle \vdash \tl{\inject{\tau;n}~(\lam{ x: \Unit}{ x})} &\arbsim\nn \ul{\prot{\tl{\Unit\to\Unit}}~(\lam{ x}{ x})} : \tl{\EmulDV\indexx{m;n}}\text{.}
  \end{align*}
  For $n$ sufficiently large and modulo some simplifications, these
  terms become:
  \begin{align*}
    \tl{\inject{\tau;n}~(\lam{ x: \Unit}{ x})} &= \tl{\inDV{\to;n-1}~(\lambda x:\UVal\indexx{n-1} \ldotp \inDV{\Unit;n-2}~(\caseDV{\Unit;n-2}~x))} \\
    \ul{\prot{\Unit\to\Unit}~(\lam{x}{x})} &= \ul{\lam{x}{x;\unitv}}
  \end{align*}
  We invite the reader to verify that both expressions behave
  appropriately when applied to any values $\tl{v}$ and $\ul{v}$ that
  are related by $\tl{\EmulDV\np}$: for example
  ($\tl{v} = \tl{\inDV{\Unit;n-1}~\unitv}$ and
  $\ul{v} = \ul{\unitv}$),
  ($\tl{v} = \tl{\inDV{\to;n-1}~(\lam{x:\UVal\indexx{n-1}}{x})}$ and
  $\ul{v} = \ul{\lam{x}{x}}$) or ($\tl{v} = \tl{\unkUVal\nn}$, $\ul{v}$
  is any \ulc term and $\square = \lesssim$).
\end{exa}

\subsection{Approximate Back-Translation}\label{sec:approx-final}

We are now ready to define the approximate back-translation
$\tl{\ef{\upc}\indexx{\tltau;n}}$ of an arbitrary untyped context $\upc$ with a hole of type \tltau. However,
before we do, we need to correct a few simplifications that were made
in \cref{fig:proving-compiler-security-approx}.

First, as we have already explained, the back-translation
$\tl{\ef{\upc}\nn}$ does not just depend on $n$ but also on the type
$\tltau$ of the terms $\tl{t_1}$ and $\tl{t_2}$ that we are compiling.
As such, we define the back-translation with $\tltau$ as an additional
parameter.

\begin{defi}[$n$-approximate back-translation \tl{\ef{\cdot}\indexx{\tau;n}}]\label{def:backtrans}
The $n$-approximate back-translation of a context \upc\ with a hole of type \tltau is defined as follows.
 $ \tl{\ef{\upc}\indexx{\tau;n}} \isdef \tl{\emulate\indexx{n+1}(\upc)\tlH{\inject{\tau;n}~\cdot}}$
\end{defi}

\begin{figure}
  \centering
  \begin{tikzpicture}

    \node at (0,2)[anchor = east] (a){ $\tl{\emulate\indexx{n}(\upc)[}$};
    \node[right =of a.east, xshift = 2.8em, anchor = east](b){\tl{\inject{\tau;n}~}};
    \node[right =of b.east, xshift = -1.5em, anchor = east](c){\tl{t}};
    \node[right =of c.east, xshift = -1.1em, anchor = east](c1){\tl{]}};

    \draw [decorate,decoration={brace,amplitude=10pt}]  ([yshift=.8em]a.west) -- ([xshift=-.1cm,yshift=.8em]b.east)   node [black,midway,yshift=2em] (em){$\tl{\ef{\upc}\indexx{\tau;n}}$};

    \node [anchor = east, below = of a.south] (d) { \upc\ul{[} };
    \node[right =of d.east, xshift = 1.7em, anchor = east](e){\ul{\prot{\tau}}};
    \node[right =of e.east, xshift = 2.4em, anchor = east](f){\ul{\erase{t}}};
    \node[right =of f.east, xshift = -1.1em, anchor = east](f1){\ul{]}};

    \draw [decorate,decoration={brace,amplitude=10pt}]  ([xshift=-.1cm,yshift=-.8em]f.east) -- ([yshift=-.8em]e.west)    node [black,midway,yshift=-2em] (co){$\comp{\tlt}$};

    \node[left =of a.west, anchor = east](t){$\tl{\ef{\upc}\indexx{\tau;n}\tlH{t}}$};
    \node[anchor = east, below = of t.north,yshift=.2cm] (gl) { $\genlogrel\nn$ };
    \node[anchor = east, below = of t.south] (tt) { \upc\ulH{\comp{\tlt}} };

    \node[rounded corners,rounded corners, fill=blue!20,below = of tt.south,xshift = -6em] (th1){ \phantom{a}};
    \node[,right = of th1.west,xshift = -1.5em] (th1t){ Terms related by \cref{lem:erase-correct}};
    \node[rounded corners,rounded corners, fill=red!20,right = of th1.east,xshift=12em] (th2){ \phantom{a}};
    \node[,right = of th2.west,xshift = -1.5em] (th2t){ Terms related by \cref{lem:protect-inj-square}};
    \node[rounded corners,rounded corners, fill=green!20,below = of tt.south,yshift=-1.2em,xshift = -6em] (th3){ \phantom{a}};
    \node[,right = of th3.west,xshift = -1.5em] (th3t){ Terms related by \cref{lem:emulate-works-ctx}};

    \draw[dashed] (-3.3,3) -- (-3.3,-0.5);

    \node[above = of em.south, yshift=-1em,xshift=-1em] (exp){expands to this};
    \node[left = of exp.west, xshift=-.5em,] (txt){This statement};

  \begin{pgfonlayer}{background}
    \node[rounded corners,fit=(c), fill=blue!20] (th1c){};
    \node[rounded corners,fit=(f), fill=blue!20] (th1f){};
    \draw[rounded corners=2em,line width=1.5em,blue!20,cap=round] (c.south) -- (f.north) node [black,midway] (r1){$\genlogrel\nn$};
  \end{pgfonlayer}
  \begin{pgfonlayer}{veryback}
    \node[rounded corners,fit=(b)(th1c),rounded corners, fill=red!20] (th1b){};
    \node[rounded corners,fit=(e)(th1f),rounded corners, fill=red!20] (th1e){};
    \draw[rounded corners,line width=1.5em,red!20,cap=round] (b.south) -- (e.north)  node [black,midway,yshift=.1em] (r2){$\genlogrel\nn$};
  \end{pgfonlayer}
  \begin{pgfonlayer}{veryback2}
    \node[rounded corners,fit=(a)(c1)(th1b),rounded corners, fill=green!20] (th1a){};
    \node[rounded corners,fit=(d)(f1)(th1e),rounded corners, fill=green!20] (th1d){};
    \draw[rounded corners=2em,line width=1.5em,green!20,cap=round] (a.south) -- (d.east) node [black,midway,yshift=.5em] (r3){$\genlogrel\nn$};
  \end{pgfonlayer}
  \end{tikzpicture}

  \caption{A more accurate picture of related components of compiled
    term $\tl{t}$, program context $\upc$, compilation $\comp{\tl{t}}$
    and emulation $\tl{\ef{\upc}\indexx{\tau;n}}$ than in the simplified
    \cref{fig:proving-compiler-security-approx}.}
  \label{fig:compile-backtrans-relations}
\end{figure}

A second simplification in \cref{fig:proving-compiler-security-approx}
was the fact that we claimed $\tl{\ef{\upc}\nn} \gtrsim\nn \upc$ and
$\tl{\ef{\upc}\nn} \lesssim \upc$.
\cref{fig:compile-backtrans-relations} shows a more accurate picture of the
relations that we have.
As we will see in the next
section, this more accurate picture still allows us to conclude the
facts that
$\tle \vdash \tl{\ef{\upc}\indexx{\tau;n}}\tlH{\tl{t_1}} \gtrsim\nn
\upc\ulH{\comp{\tl{t_1}}} : \tl{\EmulDV\indexx{n;\precise}}$
and
$\tle \vdash \tl{\ef{\upc}\indexx{\tau;n}}\tlH{t_2} \lesssim\indexx{n'}
\upc\ulH{\comp{t_2}} : \tl{\EmulDV\indexx{n;\imprecise}}$
so that the proof goes through unchanged.

The correctness of $\tl{\ef{\cdot}\indexx{\tau;n}}$ is captured in \cref{lem:correctness-back-translation}.
\begin{lem}[Correctness of \tl{\ef{\cdot}\indexx{\tau;n}}]
\label{lem:correctness-back-translation}
  If $\tlGamma\vdash\tlt \mathrel{\genlogrel\nn}\ult:\tltau$, and if
  \begin{itemize}
    \item either ($m \geq n$ and $p = \precise$) 
    \item or ($\square = \lesssim$ and $p = \imprecise$)
  \end{itemize} then
  $\tlGamma\vdash\tl{\ef{\upc}\indexx{\tau;m}}\tlH{t}
  \mathrel{\genlogrel\nn} \upc\ulH{\prot{\tltau}~\ult}:
  \tl{\EmulDV\indexx{m;p}}$.
\end{lem}
\begin{proof}
  Follows from \cref{lem:protect-inj-square,lem:emulate-works-ctx}.
\end{proof}

\section{Compiler Full-Abstraction}
\label{sec:comp-fa}
This section presents the proof that the compiler $\comp{\cdot}$ is fully-abstract (\cref{thm:comp-fa}) by relying on the logical relations of \cref{fig:logrels}.
As previously mentioned, this results in proving equivalence reflection (\cref{thm:compiler-correctness}) and preservation (\cref{thm:compiler-sec}).
As suggested by \cref{fig:proving-compiler-correctness} in \cref{sec:introduction}, the lemmas presented in \cref{sec:compiler} are enough to prove equivalence reflection for $\comp{\cdot}$.
Dually, as suggested by \cref{fig:proving-compiler-security-approx} in \cref{sec:introduction}, the lemmas presented in \cref{sec:appr-back-transl} are enough to prove equivalence preservation for $\comp{\cdot}$.

Recall from \Cref{def:comp} that $\comp{\tlt}$ is $\prot{\tltau}~(\erase{\tlt})$.

\begin{thm}[$\comp{\cdot}$ is correct]
  \label{thm:compiler-correctness}
  If $\tle \vdash \tl{t_1} : \tl{\tau}$, $\tle \vdash \tl{t_2} : \tl{\tau}$ and $\ule \vdash \ul{\comp{\tl{t_1}}} \cequlc \ul{\comp{\tl{t_2}}}$,
  then $\tl{\emptyset} \vdash \tl{t_1} \ceqstlc \tl{t_2} : \tltau$.

  \begin{proof}
    Take $\tpc$ so that $\vdash \tpc : \tl{\emptyset},\tl{\tau} \to
    \tl{\emptyset},\tl{\tau'}$. By definition of $\ceqstlc$, we need to prove
    that $\tpc\tlH{t_1} \tl{\Downarrow}$ iff $\tpc\tlH{t_2} \tl{\Downarrow}$. By symmetry, it suffices to prove the $\Rightarrow$ direction. 
    So, assume that $\tpc\tlH{t_1} \tl{\Downarrow}$. We need to prove that $\tpc\tlH{t_2} \tl{\Downarrow}$.

    Define $\upc \mydefsym \erase{\tpc}$, \cref{lem:erase-sempres-ctx} yields $\vdash \tpc \arbsim \upc : \tl{\emptyset},\tltau \to \tl{\emptyset},\tl{\tau'}$. 
    By \cref{thm:sempres}, we get $\tle\vdash \tl{t_1} \arbsim \comp{\tl{t_1}} : \tltau$ and $\tle\vdash \tl{t_2} \arbsim \comp{\tl{t_2}} : \tltau$. 
    By definition of $\vdash \tpc \arbsim \upc : \tl{\emptyset},\tltau \to \tl{\emptyset},\tl{\tau'}$, we get (specifically) that $\tle\vdash \tpc\tlH{t_1} \gtrsim \upc\ulH{\comp{\tl{t_1}}} : \tl{\tau'}$ and $\tle\vdash \tpc\tlH{t_2} \lesssim \upc\ulH{\comp{\tl{t_2}}} : \tl{\tau'}$. 

    $\tpc\tlH{t_1} \tl{\Downarrow}$ and $\tle\vdash \tpc\tlH{t_1} \arbsim \upc\ulH{\comp{\tl{t_1}}} : \tl{\tau'}$ imply that $\upc\ulH{\comp{\tl{t_1}}} \ul{\Downarrow}$ by \cref{lem:adequacy}. 
    From $\comp{\tl{t_1}}\cequlc\comp{\tl{t_2}}$ and $\upc\ulH{\comp{\tl{t_1}}} \ul{\Downarrow}$, we get that $\upc\ulH{\comp{\tl{t_2}}} \ul{\Downarrow}$. 
    $\tle\vdash \tpc\tlH{t_2} \arbsim \upc\ulH{\comp{\tl{t_2}}} : \tl{\tau'}$ and $\upc\ulH{\comp{\tl{t_2}}} \ul{\Downarrow}$ yield $\tpc\tlH{t_2} \tl{\Downarrow}$ by \cref{lem:adequacy}.
  \end{proof}
\end{thm}

\begin{thm}[$\comp{\cdot}$ is secure] \label{thm:compiler-sec}
  If $\tle \vdash \tl{t_1} : \tltau$, $\tle \vdash \tl{t_2} : \tltau$ and $\tl{t_1} \ceqstlc \tl{t_2} : \tau$,
   then $ \comp{\tl{t_1}} \cequlc \comp{\tl{t_2}}$.

  \begin{proof}
    Note that $\comp{\tl{t_1}} = \prot{\tltau}~\ul{(\erase{t_1})}$ by
    definition and similarly for $\tl{t_2}$.

    Take a $\vdash \ul{\progctx} : \ul{\emptyset \blto \emptyset}$ and
    suppose that
    $\ul{\progctx}\ulH{\ul{protect}_\tltau(\erase{\tl{t_1}})}
    \ul{\Downarrow}$,
    then we need to show that
    $\ul{\progctx}\ulH{\ul{protect}_\tltau(\erase{\tl{t_2}})}
    \ul{\Downarrow}$.

    Take $n$ larger than the number of steps in the termination of
    $\ul{\progctx}\ulH{\ul{protect}_\tltau(\erase{\tl{t_1}})}
    \ul{\Downarrow}$.

    By \cref{lem:erase-correct}, we have that
    $\emptyset \vdash \tl{t_1} \gtrsim_n \erase{t_1} : \tltau$.

    By \cref{lem:correctness-back-translation} (taking
    $m = n\geq n$, $p = \precise$ and $\genlogrel = {\gtrsim}$), we then have  that
    \begin{equation*}
      \emptyset \vdash \tl{\ef{\upc}\indexx{\tau;n}}\tlH{t_1} \gtrsim_n \upc\ulH{\prot{\tltau}~\ul{(\erase{t_1})}} : \tl{\EmulDV_{n;\precise}}\text{.}
    \end{equation*}

    Now by \cref{lem:adequacy}, by
    $\ul{\progctx}\ulH{\ul{protect}_\tltau(\erase{\tl{t_1}})}
    \ul{\Downarrow}$,
    and by the choice of $n$, we have that
    $\tl{\ef{\upc}\indexx{\tau;n}}\tlH{t_1} \tl{\Downarrow}$.

    It now follows from
    $\tl{\emptyset} \vdash \tl{t_1} \ceqstlc \tl{t_2} : \tltau$ and
    $\tl{\ef{\upc}\indexx{\tau;n}}\tlH{\tl{t_1}} \tl{\Downarrow}$
    that
    $\tl{\ef{\upc}\indexx{\tau;n}}\tlH{\tl{t_2}} \tl{\Downarrow}$.

    Now take $n'$ the number of steps in the termination of
    $\tl{\ef{\upc}\indexx{\tau;n}}\tlH{\tl{t_2}} \tl{\Downarrow}$. We have
    from \cref{lem:erase-correct} that
    $\emptyset \vdash \tl{t_2} \lesssim_{n'} \erase{t_2} : \tltau$.

    By \cref{lem:correctness-back-translation}, we then have (taking
    $m = n$, $n = n'$, $p = \imprecise$ and $\genlogrel = {\lesssim}$)
    that
    \begin{equation*}
      \emptyset \vdash \tl{\ef{\upc}\indexx{\tau;n}}\tlH{\tl{t_2}} \lesssim_{n'} \ul{\progctx}\ulH{\prot{\tltau}~\ul{(\erase{t_2})}} : \tl{\EmulDV_{n;\imprecise}}
    \end{equation*}

    Now by \cref{lem:adequacy}, by the choice of $n'$ and by the fact that
    $\tl{\ef{\upc}\indexx{\tau;n}}\tlH{\tl{t_2}} \tl{\Downarrow}$,
    we get that
    $\ul{\progctx}\ulH{\prot{\tltau}~ \ul{(\erase{t_2})}}\ul{\Downarrow}$ as required.
  \end{proof}
\end{thm}

\begin{thm}[$\comp{\cdot}$ is fully-abstract]\label{thm:comp-fa}
  If $\tle \vdash \tl{t_1} : \tltau$ and $\tle \vdash \tl{t_2} : \tltau$, then
  \begin{equation*}
    \tl{t_1} \ceqstlc \tl{t_2} \iff \comp{\tl{t_1}} \cequlc \comp{\tl{t_2}}
  \end{equation*}
  \begin{proof}
  \Cref{thm:compiler-correctness} provides the $\Leftarrow$ direction while \cref{thm:compiler-sec} provides the $\Rightarrow$ one.
  \end{proof}
\end{thm}


\section{Modular Fully-Abstract Compilation} \label{sec:modular}





For the sake of simplicity, so far we only considered compilers that take a whole program as input, keeping modular compilers (and thus linking of compiled programs) out of the picture.
However, for the proof technique to be applicable and useful in real-world scenarios, it must scale to modular compilers.
Modular compilers compile different parts of a program separately, leading to faster re-compilation times since only the fragments that changed since the last compilation are recompiled.
This section extends the presented proof technique to modular compilers.

More in detail, a modular compiler considers source programs that are \emph{open}; it compiles these open programs independently and \emph{links} the result to form the runnable program.
Open programs are those that have dependencies on other ones, e.g., a secure transactions program could rely on third-party cryptographic-signing function to accomplish its task.
The program does not implement the signing function itself, rather it relies on such a function to be provided at link time.
Linking is the process of taking open programs and fulfilling their dependencies with the other programs they are linked against.
When the aforementioned secure transaction program is linked against the code that provides the signing function, the linker ensures that whenever the program calls that function, the call is dispatched to the actual implementation.

Full-abstraction as stated in \Cref{thm:comp-fa} is not a correct criterion for modular compilers, as it is stated for closed terms.
Instead, a generalisation exists for modular compilers: modular full-abstraction~\citep{mfac}.
Modular full-abstraction forces one to reason about linking of programs when developing a fully-abstract compiler.
Modular full-abstraction can be derived from compiler modularity and full-abstraction stated with an open environment (so not as in \Cref{thm:comp-fa}).

In the remainder of this section, we explain how to turn the compiler developed so far into a modular one and prove it to be modularly fully-abstract.
In order to discuss modularity and linking, this section first defines what open terms are and introduces a notion of linking in both source and target languages (\cref{sec:openlinking}).
Then it extends the compiler to work for open terms (\cref{sec:seccompopen}) so that it can be proven to be modularly fully-abstract (\cref{sec:mfa-comp}).
The proofs are all carried out using the machinery developed in the previous section, which furthers our belief in the strength of the proof technique.

\subsection{Open Terms and Linking}\label{sec:openlinking}
Open terms are already part of the model since they are those that are type checked against a non-empty environment.
For the sake of simplicity, we only consider linking two terms $\tl{t_1}$ and $\tl{t_2}$ (linking an arbitrary number of terms simply adds an inductive step to the formal development but no additional insight).
Both $\tl{t_1}$ and $\tl{t_2}$ have a single free variable that the other term is intended to fulfill, i.e.  $\tl{t_1}$ has a free variable  $\tl{x_2}$ of the same type as
 $\tl{t_2}$, and  $\tl{t_2}$ has a free variable  $\tl{x_1}$ of the same type as  $\tl{t_1}$.
We allow $\tl{t_1}$ and $\tl{t_2}$ to be mutually dependent, but the case for non-mutually dependent terms follows as a special case.

Intuitively, the linker must return the pair of $\tl{t_1}$ and $\tl{t_2}$ where the free variable of each term is replaced with the other term.
Because these two terms have mutual dependencies, linking is encoded by using a fixpoint to produce a pair containing versions of $\tl{t_1}$ and $\tl{t_2}$ with the occurrences of the free variable of $\tl{t_1}$ filled in with $\tl{t_2}$ (and vice-versa).
Recall that fixpoint is a syntactic form that exists in both languages.

Since we are in a call-by-value setting, fixpoints are a bit delicate.
Specifically, if we just feed any two arbitrary terms $\tl{t_1}$ and $\tl{t_2}$ to the fixpoints, it is not possible to produce the fixpoint without risking divergence.
To address this (known) problem, we restrict the compiler to lambdas $\tl{\lam{x_1':\tau_1'}{t_1}}$ and $\tl{\lam{x_2':\tau_2'}{t_2}}$, as one would expect from a call-by-value program.

However, we cannot simply use a fixpoint to produce the pair that we want because we had to encode $\ul{\fix{}}$ in \ulc as the Z combinator (which can only produce fixpoints that are functions).
While intuitively the linker should just use $\fix{}$ to produce the pair of the two terms, in this case it needs to be wrapped into a lambda that discards its argument.
We choose to supply $\tl{\Unit}$-type values to such a lambda.

\Cref{def:linking} presents linking in \stlc and \ulc.
\begin{defi}[Linking]\label{def:linking}
If
\begin{align*}
  \tl{x_2}:\tl{\tau_2'\to\tau_2}\vdash&\ \tl{t_1}:\tl{\tau_1'\to\tau_1}
  \\
  \tl{x_1}:\tl{\tau_1'\to\tau_1}\vdash&\ \tl{t_2}:\tl{\tau_2'\to\tau_2}
\end{align*}
then
\begin{align*}
  \tl{t_1+t_2} \isdef&\ 
    \tl{
        \left(\begin{aligned}
          &\fix{\Unit\to((\tau_1'\to\tau_1)\times(\tau_2'\to\tau_2))}
          \\
          &\quad (\lam{p:\Unit\to((\tau_1'\to\tau_1)\times(\tau_2'\to\tau_2))}{\lam{\_:\Unit}{
          \\
          &\qquad\qquad\left\langle
            \begin{aligned}
               &\lam{x_1':\tau_1'}{((\lam{x_2:\tau_2'\to\tau_2}{t_1})~((p~\unitv).2)})~x_1', \\
               &\lam{x_2':\tau_2'}{((\lam{x_1:\tau_1'\to\tau_1}{t_2})~((p~\unitv).1)})~x_2'
            \end{aligned}\right\rangle) } }
        \end{aligned}\right)
        ~\unitv
    }
\end{align*}
We can show that the this produces a well-typed term: $$\tl{t_1+t_2} : \tl{((\tau_1'\to\tau_1)\times(\tau_2'\to\tau_2))}$$

If
\begin{align*}
  \ul{x_2}\vdash&\ \ul{t_1} 
  \\
  \ul{x_1}\vdash&\ \ul{t_2}
\end{align*}
then
\begin{align*}
  \ul{ t_1 + t_2} \isdef&
  \ul{ 
    \begin{aligned}
    \left(\mi{fix}\left(\lam{p}{\lam{\_}{
      \left\langle
      \begin{aligned}
        &\lam{x_1'}{((\lam{x_2}{t_1})~(p~\unit).2) ~ x_1'}
        , 
        \\
        &\lam{x_2'}{((\lam{x_1}{t_2})~(p~\unit).1)~ x_2'}
      \end{aligned}
      \right\rangle
    }}\right) \right) ~ \unitv
    \end{aligned}}
\end{align*}

\end{defi}

Both linkers are defined analogously.
They use the recursive fix arguments (\tl{p} and \ul{p} respectively) inside the lambda term, binding the projections of that argument to the
corresponding free variable (for instance, binding the second projection of \tl{p} to \tl{x_2}).
As stated before, the recursive application is done after an eta-expansion to prevent the term from diverging.

In the context of modular full abstraction, it is important that for any term $\tlt$, linking with $\tlt$ produces a valid program context $\cdot + \tl{t}$.
This way, a program context (representing an adversary) can link the program with an arbitrary term of its choosing, i.e.\ a compiled program cannot trust what it is being linked against.
As a result of this, compiled programs must perform checks against that code too; these checks are the modifications to the compiler to which we turn next.
The advantage of the fact that linking produces valid contexts is that if we take $\upc\ulH{t_1+t_2}$, i.e. we let a linked program $\ul{t_1 + t_2}$ interact with an adversary context $\upc$, then if $\tl{t_2}$ contains a security bug, we can still  change our point of view and consider $\upc\ulH{\cdot+t_2}$ as an adversary context that trusted program $\ul{t_1}$ is being linked against.
In other words, modular full abstraction implies a form of compartmentalisation: security bugs in a component do not expose other components' internals.

\subsection{A Secure Compiler for Open Terms}\label{sec:seccompopen}
The compiler definition changes as in \Cref{def:comp-su-mod} to account for open terms and compiled terms being in a lambda-form.
\begin{defi}[A Modular Compiler $\compsu{\cdot}$]\label{def:comp-su-mod}
Assuming
\begin{itemize}
  \item  $\tl{x_2}:\tl{\tau_2'\to\tau_2}\vdash \tl{\lam{x_1':\tau_1'}{t_1}}:\tl{\tau_1'\to\tau_1}$,
 \end{itemize} then:
$$ \compsu{\lam{x_1':\tau_1'}{t_1}} = \prot{\tau_1'\to\tau_1} \ul{(\lam{x_1'}{((\lam{x_2}{\erase{t_1}})(\conf{\tau_2'\to\tau_2}~x_2))})} $$
\end{defi}

The compiler knows that $\erase{\lam{x_1'}{t}}$ will generate an open term with an open variable \ul{x_2}.
So it closes that variable with a $\ul{\lam{x_2}{\cdot}}$ just to open it again with the argument of that lambda (the term $\ul{(\conf{\tau_2'\to\tau_2} x_2)}$).
The point of this is to force a \conf{\cdot} around the free variables.

If one considers two source terms being compiled with this compiler and then linked, then the extra \conf{\cdot} is redundant.
However, linking at the target level can be done with arbitrary terms, so they need to be restricted on how they interoperate with these terms by calling \conf{\cdot} on them.
The term supplied for the open variable is like the argument of a function (the linker really treats it that way), thus the choice of \conf{\cdot}.
Adding this additional \conf{\cdot} does not disrupt the functionality of compiled code, as proved by \Cref{lem:extra-conf}.
\begin{lem}[An extra confine is just fine]\label{lem:extra-conf}
If $\tlGamma,\tl{x:\tau'}\vdash \tl{t}:\tl{\tau}$,
 then 
 \begin{equation*}
\tlGamma,\tl{x:\tau'} \vdash \tl{t} \arbsim_n \ul{(\lam{x}{\erase{t}})(\conf{\tau'}~x)} : \tl{\tau}
\end{equation*}
\end{lem}

\subsection{Modular Full-Abstraction for $\compsu{\cdot}$}\label{sec:mfa-comp}
The property that $\compsu{\cdot}$ must have is modular full-abstraction~\citep{mfac}, which is the combination of compiler full-abstraction with an open environment (\Cref{thm:facomp-o} in \Cref{sec:fa-open}) and compiler modularity (\Cref{thm:comp-mod} in \Cref{sec:compmod}).

\subsubsection{Full-Abstraction with an Open Environment}\label{sec:fa-open}
The first step to re-prove compiler full-abstraction is compiler correctness for open lambda-terms (\Cref{thm:conftransp}).
\begin{thm}[$\compsu{\cdot}$ is correct]\label{thm:conftransp}
If
\begin{equation*}
  \tl{x_2}:\tl{\tau_2'\to\tau_2}\vdash \tl{\lam{x_1':\tau_1'}{t_1}}:\tl{\tau_1'\to\tau_1} 
\end{equation*}
then 
\begin{equation*}
\tl{x_2:\tau_2'\to\tau_2} \vdash \tl{\lam{x_1':\tau_1'}{t_1}} \arbsim_n \prot{\tau_1'\to\tau_1} \ul{(\lam{x_1'}{((\lam{x_2}{\erase{t_1}})(\conf{\tau_2'\to\tau_2}~x_2))})} : \tl{\tau_1'\to\tau_1}.
\end{equation*}
\end{thm}

We then re-state \Cref{thm:compiler-correctness} and \Cref{thm:compiler-sec} to work for open lambda terms only and to work for the new definition of $\compsu{\cdot}$.
\begin{thm}[$\compsu{\cdot}$ reflects equivalence] \label{thm:compiler-correctness-l}
If 
\begin{align*}
   \tl{x}:\tl{\tau'\to\tau}&\vdash \tl{\lam{x_1':\tau_1'}{t_1}}:\tl{\tau_1'\to\tau_1}, \\
   \tl{x}:\tl{\tau'\to\tau}&\vdash \tl{\lam{x_2':\tau_1'}{t_2}}:\tl{\tau_1'\to\tau_1},\\
   \ul{x} &\vdash \compsu{\lam{x_1':\tau_1'}{t_1}} \cequlc \compsu{\lam{x_2':\tau_1'}{t_2}}
\end{align*} then 
\begin{equation*}
\tl{x:\tau'\to\tau} \vdash \tl{\lam{x_1':\tau_1'}{t_1}} \ceqstlc \tl{\lam{x_2':\tau_1'}{t_2}} : \tl{\tau_1'\to\tau_1}.
\end{equation*}
\end{thm}

\begin{thm}[$\compsu{\cdot}$ preserves equivalence]\label{thm:contextual-equivalence-preservation-l}
If 
\begin{align*}
  \tl{x}:\tl{\tau'\to\tau}&\vdash \tl{\lam{x_1':\tau_1'}{t_1}}:\tl{\tau_1'\to\tau_1},\\
  \tl{x}:\tl{\tau'\to\tau}&\vdash \tl{\lam{x_2':\tau_1'}{t_2}}:\tl{\tau_1'\to\tau_1},\\
  \tl{x:\tau'} &\vdash \tl{\lam{x_1':\tau_1'}{t_1}} \ceqstlc \tl{\lam{x_1':\tau_2'}{t_2}} : \tltau,
\end{align*} then
\begin{equation*}
\ul{x} \vdash \compsu{\lam{x_1':\tau_1'}{t_1}} \cequlc \compsu{\lam{x_2':\tau_2'}{t_2}}.
\end{equation*}
\end{thm}

The proofs of \Cref{thm:compiler-correctness-l} and of \Cref{thm:contextual-equivalence-preservation-l} are analogous to their closed-environments analogues except that they rely on \Cref{thm:conftransp}.
They are reported in the companion tech report.

\begin{thm}[Compiler Full Abstraction]\label{thm:facomp-o} If \hfill
\begin{align*}
   \tl{x}:\tl{\tau'\to\tau}&\vdash \tl{\lam{x_1':\tau_1'}{t_1}}:\tl{\tau_1'\to\tau_1},\\
   \tl{x}:\tl{\tau'\to\tau}&\vdash \tl{\lam{x_1':\tau_1'}{t_2}}:\tl{\tau_1'\to\tau_1},
\end{align*} then
\begin{equation*}
\tl{x:\tau'\to\tau} \vdash \tl{\lam{x_1'}{t_1}}\ceqstlc\tl{\lam{x_2'}{t_2}} : \tl{\tau_1'\to\tau_1} \iff \ul{x}\vdash \compgen{\lam{x_1'}{t_1}}\cequlc\compgen{\lam{x_2'}{t_2}}.
\end{equation*}
\end{thm}
\begin{proof}
  By \Cref{thm:contextual-equivalence-preservation-l} and \Cref{thm:compiler-correctness-l}.
\end{proof}

\subsubsection{Compiler Modularity}\label{sec:compmod}
Compiler modularity is a property that is stated just between \ulc terms.
Intuitively, it states that linking two source terms $\tl{t_1}$ and $\tl{t_2}$ and compiling the result is contextually-equivalent to compiling $\tl{t_1}$ and $\tl{t_2}$ individually and then linking the result in the target.
Formally, this is captured by \Cref{thm:comp-mod}, where compiler modularity is written with a closed environment as linking is generally a global step.
\begin{thm}[Compiler modularity]\label{thm:comp-mod}
If 
\begin{align*}
   \tl{x_2:\tau_2'\to\tau_2}&\vdash \tl{\lam{x_1':\tau_1'}{t_1}}:\tl{\tau_1'\to\tau_1}\\
   \tl{x_1:\tau_1'\to\tau_1}&\vdash \tl{\lam{x_2':\tau_2'}{t_2}}:\tl{\tau_2'\to\tau_2}
\end{align*} then
\begin{equation*}
\ule \vdash \compsu{\lam{x_1':\tau_1'}{t_1}+\lam{x_2':\tau_2'}{t_2}} \cequlc \compsu{\lam{x_1':\tau_1'}{t_1}}\ul{+}\compsu{\lam{x_2':\tau_2'}{t_2}}.
\end{equation*}
\end{thm}

The formal setup developed so far (i.e., the logical relation) however, is only built for cross-language reasoning.
Since we do not really have a \ulc logical relation, nor do we want to build one, we resort to an analogous of the proof of compiler security -- the part that relies on the back translation, except that we will have the same term on both sides of the source contextual equivalence (\Cref{fig:proving-compiler-modularity}).

\begin{figure}
  \centering
  \begin{tikzpicture}[scale=0.84,every node/.style={scale=.9}]
  \node at (5.2,4.8) { $\tl{t_1+t_2\mathrel{\ceq} t_1+t_2}$ };

  \node at (3,4.1) { $\tl{\ef{\upc}\nn\big[}\tl{t_1+t_2}\tl{\big] \Dan{}{\_}}$ };
  \node at (5,4.1) { $\mathrel{\Ra}$ };
  \node at (7,4.1) { $\tl{\ef{\upc}\nn\big[}\tl{t_1+t_2}\tl{\big] \Dan{}{\_}}$ };

  \node at (4.35,2.8) { (1) };
  \node at (5,3.6) { (2) };
  \node at (5.65,2.8) { (3) };

  \draw[out=100,in=260,double,-implies,double equal sign distance] (4,2.6) to (4,3.4);

  \draw[out=280,in=80,double,-implies,double equal sign distance] (6,3.4) to (6,2.6);

  \node[align=left] at (9,3) { $ \tl{\ef{\upc}\nn} \lesssim\indexx{\_} \upc$ \\
    $ \tl{t_1+t_2} \lesssim\indexx{\_} \compsu{t_1}\ul{+}\compsu{t_2}$};
  \node[align=left] at (1.5,3) { $ \tl{\ef{\upc}\nn} \gtrsim\nn \upc$ \\
    $ \tl{t_1+t_2} \gtrsim\indexx{\_} \comp{t_1+t_2}$};

  \node at (3,1.9) { $\ul{\progctx\Big[}\comp{t_1+t_2}\ul{\Big] \Dan{}{n}}$ };
  \node at (5,2) { $\overset{?}{\Ra}$ };
  \node at (7.4,1.9) { $\ul{\progctx\Big[}\compsu{t_1}\ul{+}\compsu{t_2}\ul{\Big] \Dan{}{\_}}$ };

  \node at (5.5,1.2) { $\compsu{t_1+t_2}\mathrel{\overset{?}{\cequlc}}\compsu{t_1}\ul{+}\compsu{t_2}$ };
  \end{tikzpicture}
  
  \caption{Proving compiler modularity. Only one direction of this half is presented ($\Rightarrow$), the other one follows by symmetry.}
  \label{fig:proving-compiler-modularity}
\end{figure}

Step 1 is given by the correctness of $\compsu{\cdot}$ (\Cref{thm:conftransp}) and Step 2 is trivial since a term is equivalent to itself.
All that remains to be proven is Step 3, as captured by \Cref{lem:source-target-linking}.
\begin{lem}[Source linking is related to target linking]\label{lem:source-target-linking}
If
\begin{align*}
  \tl{x_2:\tau_2'\to\tau_2}&\vdash \tl{\lam{x_1':\tau_1'}{t_1}}:\tl{\tau_1'\to\tau_1} \\
  \tl{x_1:\tau_1'\to\tau_1}&\vdash \tl{\lam{x_2':\tau_2'}{t_2}}:\tl{\tau_2'\to\tau_2} 
\end{align*} then
\begin{multline*}
  \tle \vdash \tl{(\lam{x_1':\tau_1'}{t_1})+(\lam{x_2':\tau_2'}{t_2})} \arbsim_n \\
  \ul{\compsu{\lam{x_1':\tau_1'}{t_1}}+\compsu{\lam{x_2':\tau_2'}{t_2}}} : \tl{(\tau_1'\to\tau_1)\times(\tau_2'\to\tau_2)}.
\end{multline*}
\end{lem}
Given that source and target linking are defined to be syntactically duals, this proof is a mere application of several compatibility lemmas and \Cref{lem:extra-conf}.


\section{Mechanically verified proof} \label{sec:coq-proof}

Proofs of full abstraction for non-trivial compiler passes are very often only
given on paper. The reason is that they are quite involved and require a
significant effort to mechanically verify. This is unfortunate because the
proofs are often lengthy and non-trivial, so they would benefit from the extra
assurance offered by mechanical verification. In this section, we report on our
succesful mechanical verification using Coq of the full abstraction proof
presented in \cref{sec:seccompopen}. 

This proof has been a significant undertaking ($\pm$ 2-3 man-months, 11k lines
of code excluding comments) resulting in a medium-sized coq development,
available online\footnote{\url{http://people.cs.kuleuven.be/dominique.devriese/permanent/facomp-stlc-coq.tar.xz}}. The proof only assumes a single axiom: functional
extensionality, i.e. the property that two functions are provably equal if they
produce the same result for all inputs. 
This is used in the substitution machinery and for proving propositional equality for the folding/unfolding of the well-founded fixpoints that define the logical relations.
In this section, we provide an overview
of the construction of the proof and some discussion of our experience
constructing it.

The Coq version of the proof largely follows the structure of the ``on-paper''
version. As is often the case in proofs about properties of lambda calculi, a
lot of the overhead arises from the definition of the syntax and more precisely,
its use of variables and variable binding. Since the POPLmark
challenge~\citep{Aydemir2005}, the literature has seen a lot of proposals for
encoding such definitions in proof assistants. However, our requirements on the
encoding go beyond those of the challenge in some places, so that our options
for the encoding were a bit more limited. In particular, we require a notion
of \emph{simultaneous closing substitutions} (i.e. substitutions that
simultaneously substitute closed values for \emph{all} the free variables of a
term), we need to deal with two separate lambda calculi (\pstlc and \ulc)
without duplicating all the binding infrastructure and finally, we want to keep
our options open for extending the proof to System~F (see \cref{sec:disc}).

Because of these requirements, we decided to rely on traditional but proven
technology and we chose a standard encoding using de~Bruijn indices (with a
separate well-scopedness judgement). 
We used the UniDB library\footnote{Available at: \href{https://github.com/skeuchel/unidb-coq}{https://github.com/skeuchel/unidb-coq}} for de Bruijn encodngs in Coq (developed by the last author using his experiences working on the Knot framework~\cite{Keuchel2016}). 
UniDB is instantiated with language specific traversal functions of our source and target language and properties of these traversal, from which the variable binding boilerplate like simultaneous substitution and its properties are derived. 
The library defines a set of Coq type classes as an interface that allows us to use the same notation for the source and target language. 
On top of what UniDB provides, we have constructed a number of Ltac tactics for automating the construction of language-specific well-typedness, well-scopedness and evaluation proofs.

Encoding the logical relation did not cause major concerns. The step-indexing in
the LR could be expressed without many problems using a library for well-founded
induction over natural numbers from Coq's standard library. We did not encounter
major problems in the original on-paper proof, although we did run into some minor
problems, like the need to be more explicit about the required closedness of
untyped terms in the value relation.

In summary, this mechanical verification obviously strengthens our trust in the
proof. However, we point out that, to the best of our knowledge, it is the first
full proof of full abstraction for a non-trivial compiler pass (see
\cref{sec:related-work} for a discussion of related work). As such, it
demonstrates that such proofs are within reach of current tools like Coq,
although there is room for improvement: 2-3 man-months is more effort than we
would hope to spend for a proof that we have already done in much detail on
paper. Conversely, the fact that our proof did not uncover any major problem and
most of the effort went into dealing with variable binding, well-typedness
proofs etc., also shows that this sort of logical relation-based proofs of
full abstraction lend themselves well to mechanisation and the level of detail
of our original proof is a good level to aim for.

\section{Discussion and Future Work} \label{sec:disc}

Our interest in fully-abstract compilation comes from a security
perspective. We think that a fully-abstract compiler from realistic
source languages to a form of assembly that is efficiently executable
by processors has important security applications (combining trusted
and untrusted code at the assembly level and compartmentalising
applications). So far, it remains unclear precisely which security
properties are preserved by fully abstract compilers, although it
seems that at least important security properties like
noninterference~\citep{nonintfree} are. Unless targeting typed assembly
language~\citep{tal}, a crucial step of a secure compiler is a form of
secure type erasure. The contribution of this paper is mostly the
proof technique that proves the type erasure step secure. We intend to
reuse this proof technique in other settings.

There are a number of important problems that need to be solved in order to develop a realistic fully-abstract compiler.
Several widely-implemented high-level language features present significant challenges: parametric polymorphism, (higher-order) references, exceptions etc.
Generally, we believe that low-level assembly languages should be defined that are not only efficiently executable but also provide sufficient abstraction features to enable fully abstract compilation of such standard programming language features.
For now, it remains an open question whether this is feasible.
Let us zoom in on some of these features in more detail.
A long-standing open problem is fully-abstract compilation of parametric polymorphism to a form of operational sealing primitives~\citep{seal,Matthews2008ParaPoly,Neis:2009:NP:1596550.1596572}.
More concretely, several researchers have developed interesting results about fully-abstract compilation from System~F to \lseal (an untyped lambda calculus with sealing primitives), but a fully-abstract compiler in this setting has so far only been conjectured.
We believe that the problem is quite related to the one tackled in this paper.
Without providing details (for space reasons), an exact back-translation from \lseal to System~F seems possible, but only if we assume a form of generally recursive type constructors of kind $\ast\to\ast$, which we cannot add to System~F without causing other problems for the compilation.
We have been working on a proof using an approximate back-translation, but we ran into an unexpected problem: the conjectured full abstraction of Sumii-Pierce's compiler is false.
We will report on this further in future work.

In other settings, it is also not clear whether it is possible to construct a fully-abstract compiler.
For example, if we add typed, higher-order references to \pstlc and untyped references to \ulc, it is not clear if a fully-abstract compiler can be devised.
The problem is essentially to choose a representation for typed references and a way of manipulating them that reconciles a number of requirements: (1) trusted code reading from a reference always produces a type-correct value, (2) trusted code writing a type-correct value to a reference always works, (3) untrusted code should be able to read/write type-correct values from references, (4) dynamic type checks or wrappers may only be added where the context could also choose to fail for other reasons (i.e. not at the time of reading/writing a reference by trusted code), (5) efficiency: we do not want to check the contents of all references every time control is passed from trusted code to the context.
Several obvious solutions do not work: representing references as objects with read and write methods violates requirement (4), just checking the contents of a reference when it is received from the context is not enough to guarantee (1) and (2).
We intend to explore a solution based on trusted but abstract read/write/alloc methods (using sealing primitives as used for parametric polymorphism) but this remains speculation for the moment.

Another interesting problem when compiling to an assembly language is the enforcement of well-bracketed control flow.
The question is essentially how to represent return pointers at the assembly level.
Even if we prevent functions from accessing parts of the stack and only give them access to an opaque invokable return pointer, they still have ways to misuse them~\citep{mfac}.
Imagine a trusted assembly function $f$ invoking an untrusted $g$.
Additionally, assume that $g$ in turn re-invokes $f$ and $f$ simply re-invokes $g$ again.
Now $g$ might attempt to invoke the wrong return pointer, returning on its first invocation without first returning on the second.
Such an attack breaks the well-bracketedness of control flow that trusted code may rely on in languages without call/cc primitives~\citep{Dreyer:2010:IHS:1863543.1863566}.
\citep{ahmedCPS} have demonstrated a solution for this problem which exploits parametric polymorphism to enforce the invocation of the correct continuation, and it is interesting to see if their work can be reused as an intermediate step on the way to assembly language.

On a technical level, we expect few problems for applying our
technique of approximate back-translation to all of these settings.
The Hur-Dreyer-inspired  cross-language logical relations can be applied in diverse
settings including ML and assembly and support references (through
Kripke worlds), parametric polymorphism (through quantification over
abstract type interpretations as relations) and well-bracketed control
flow guarantees (through public/private transitions in the transition
systems stored in the worlds). We have also shown in this paper that
they can be easily modified to an asymmetric setting.

\section{Related Work} \label{sec:related-work}

Secure compilation through full-abstraction was pioneered by
\citet{abadiFa} and successfully applied to many different
settings~\citep{scoo-j,fstar2js,nonintfree,ahmedCPS,Ahmed:2008:TCC:1411203.1411227,depIntoPar,fixTse,protOnLayRand,Jagadeesan:2011:LMV:2056311.2056556,Riecke1993:MSC:4117452,Ritter1995FullyAbs,Mitchell1993141,McCusker1996FullAbs,Smith1998coverage,adriaanaplas}.

Recently, \citet{faEHM} and \citet{gcFA} have argued against the use of the mere existence of
fully-abstract translations as a measure of language expressiveness, because
very often fully abstract translations exist but are in some sense degenerate,
uninteresting and/or unrealistic. These arguments are not directly relevant to our work,
because we are not interested in the mere existence of a fully abstract compiler
as a measure of language expressiveness, but we prove the fully abstractness of
a specific, realistic compiler.

Some secure compilation works prove compiler full-abstraction using logical relations.
\citet{ahmedCPS} and \citet{Ahmed:2008:TCC:1411203.1411227} proved that typed closure conversion and CPS transformation are fully-abstract when compiling from System~F and the simply-typed \lc (respectively) to System~F.
\citet{depIntoPar} started a line of work to compile the dependency core calculus of \citet{dcc} (DCC) into System~F, effectively proving that non interference can be encoded with parametricity.
They achieve a property analogous to fully-abstract compilation where contextual equivalence is replaced with non-interference.
Due to an imprecision in their proof, the result of Tse and Zdancewic does not hold; \citet{fixTse} refined their result for a weaker form of DCC.
A fully-abstract translation from DCC to System~F was provided by \citet{nonintfree}, and that is the closest work to what is presented here.
The formal machinery adopted by Bowman and Ahmed does appear a bit heavier than the one presented here.
Specifically, we do not need a new logical relation to prove well-foundedness of the back-translation.
The secure compilation of DCC to System~F is quite different from our setting, since our target language is untyped and our source and target languages are both non-terminating.

In a paper that is closely related to our work, \citet{New:2016:FAC:2951913.2951941} prove full-abstraction of closure conversion of a simply-typed lambda calculus with recursive types into a simply typed language with exceptions and an effect system to track exceptions.
To achieve this, they apply a back-translation using a universal type, similar to our UVal.
They present a very interesting comparison to a previous version of this work.
They explain how one can see our approximation of target-language terms as an underapproximation because we only back-translate a part of the behaviour of the term.
While they do not need this under-approximation, because their source language includes recursive types, they apply an over-approximation because their universal type can embed more than just their target language: like our UVal, it can embed the full untyped lambda calculus, rather than just the subset of terms that are well-typed in the target language.

Independently from our work, the idea of using an \emph{approximate} back-translation was also mentioned recently by \citet{SchmidtSchauss201598}.
In this work, Schmidt-Schau{\ss} et al. present a framework for defining and reasoning about fully abstract compilation and related notions in a wide variety of languages.
Using the name \emph{families of translations}, they define what we call an approximate back-translation (in relation to full abstraction of a language embedding).
They apply the idea to show that a simply-typed lambda calculus without fix but with stuck terms can be embedded into a simply-typed lambda calculus with fix.
The idea is to use an approximate back-translation that unrolls applications of fix $n$ times in the $n$th approximation.
The proof is not very detailed, but seems a lot simpler than ours.
Partly this is because the proof addresses a simpler problem, but the idea of approximate back-translation also seems simpler to use for a language embedding.
This suggests that the proof in this paper may be simplified by factoring our compiler into two separate compilation passes: (1) embedding \pstlc into \pstlc with recursive types (using an approximate back-translation to prove full abstraction of the embedding) and (2) compiling \pstlc with recursive types into \ulc as we do here (using a full, non-approximate back-translation to prove full abstraction).

Many other secure compilation works prove full-abstraction by replacing target-level contextual equivalence with another equivalent equivalence (most times it is trace equivalence or bisimilarity)~\citep{fstar2js,protOnLayRand,Jagadeesan:2011:LMV:2056311.2056556,scoo-j}.
These works rely on additional results of the equivalence used for full-abstraction to hold, and this can complicate and lengthen proofs relying on this other technique.
Earlier, \citet{McCusker1996FullAbs} has shown that proving
full abstraction of a compiler can be simplified by limiting the
back-translation to contexts that are in a certain sense
\emph{compact}. This is related to our approximate back-translations,
though not quite the same. A downside of McCusker's approach is that
it does not always seem clear how to characterize the compact elements
in a language.

To the best of our knowledge, the only paper on secure compilation that comes with a mechanised proof of full abstraction (or a variant of it), is by \citet{catalin}.
They propose and study \emph{Secure Compartmentalizing Compilation}: a variant of full abstraction that supports unsafe source languages (where full abstraction cannot be expected to hold for components that exhibit undefined behaviour), and includes a notion of modularity.
They report on a Coq mechanization of some of their results, but it is not a ``full mechanization'' like ours, in the sense that they keep many lemmas and results as unproven assumptions.

As mentioned in \cref{sec:introduction}, the presented proof technique borrows from recent results in compiler correctness~\citep{Hur:2011:KLR:1926385.1926402,realizability,bistcc}.
These results build cross-language logical relation based on a common language specification in order to prove compiler correctness.
\citet{bistcc} provided a correct compiler from a call-by-value $\lambda$-calculus as well as for System F with recursion to a SECD machine~\citep{realizability}.
\citet{Hur:2011:KLR:1926385.1926402} devised a correct compiler from an idealised ML to assembly.
The techniques devised in these works were further developed into Relational Transition Systems (RTS) and Parametric Inter-Language Simulations (PILS) in order to prove both vertically- and horizontally-composable compiler correctness~\citep{marriage,pils}.
A different approach to cross-language relations could have been adopting a Matthews and Findler-style multi-language semantics, where source and target language are combined~\citep{Matthews:MLS}.
For example, \citet{perconti} devised a two-step correct compiler for System F with existential and recursive types to typed assembly language using multi-language logical relations.
As the focus of this work is compiler full-abstraction for a compiler that is not multi-pass, there was no necessity to use RTS nor multi-language systems.
However, were our compiler to be multi-pass, we would have had to resort to a different proof technique like those described above.
This is because even if compiler full abstraction scales to multi-pass compilers (i.e., it is vertically composable), compiler correctness proven with logical relations does not.
As compiler full abstraction ought to always be accompanied by a compiler correctness result (\Cref{thm:compiler-correctness} in our case), and since both can be proven with a single proof technique, we believe that to prove compiler full abstraction for a multipass compiler one would have to port the findings of this paper to one of the techniques above.

Some elements of our proof technique are reminiscent of techniques from the
field of denotational semantics. First, our family of types $\tl{\UVal\nn}$ can
be seen as a kind of syntactical version of an iteratively constructed Scott
model for the untyped lambda calculus~\citep{Scott1976Domains}. In fact, the
analogue of our $\tl{\UVal\nn}$ used by \citet{New:2016:FAC:2951913.2951941} extends
this correspondence to a language with effects (using an exception monad to
model a target language with exceptions). We note also that using a family of
finite approximations (like our $\tl{\UVal\nn}$ types) to interpret a recursive type
(like the type $\UVal$ discussed in the introduction) is quite standard in
denotational semantics~\citep{MacQueen:1984:IMR:800017.800528}.


\section{Conclusion} \label{sec:conclusion}
This paper presented a novel proof technique for proving compiler full-abstraction based on asymmetric, cross-language logical relations.
The proof technique revolves around an approximate back-translation from target terms (and contexts) to source terms (and contexts).
The back-translation is approximate in the sense that the context generated by the back-translation may diverge when the target-level counterpart would not, but not vice versa.
The proof technique is demonstrated for a compiler from a simply-typed \lc without recursive types to the untyped \lc.
That compiler is proven to be fully-abstract in Coq, and this is the first such result for fully abstract compilation proofs.
Although logical relations have been used for full-abstraction proofs, this is the first usage of cross-language logical relations for compiler full-abstraction targeting an untyped language.
We believe the techniques developed in this paper scale to languages with more advanced functionalities and they can be used to prove compiler full-abstraction in richer settings.

\section*{Acknowledgements}
Dominique Devriese holds a Postdoctoral Fellowship from the Research Foundation - Flanders (FWO). 
This research is partially funded by project funds from the Research Fund KU Leuven and the Research Foundation - Flanders (FWO).

\bibliographystyle{abbrvnat}   
\bibliography{refs.bib}

\begin{thebibliography}{46}
\providecommand{\natexlab}[1]{#1}
\providecommand{\url}[1]{\texttt{#1}}
\expandafter\ifx\csname urlstyle\endcsname\relax
  \providecommand{\doi}[1]{doi: #1}\else
  \providecommand{\doi}{doi: \begingroup \urlstyle{rm}\Url}\fi

\bibitem[Abadi(1999)]{abadiFa}
M.~Abadi.
\newblock Protection in programming-language translations.
\newblock In \emph{Secure Internet programming}, pages 19--34. Springer-Verlag,
  1999.
\newblock ISBN 3-540-66130-1.

\bibitem[Abadi and Plotkin(2012)]{protOnLayRand}
M.~Abadi and G.~D. Plotkin.
\newblock On protection by layout randomization.
\newblock \emph{ACM Transactions on Information and System Security},
  15:\penalty0 8:1--8:29, July 2012.
\newblock ISSN 1094-9224.
\newblock \doi{10.1145/2240276.2240279}.

\bibitem[Abadi et~al.(1999)Abadi, Banerjee, Heintze, and Riecke]{dcc}
M.~Abadi, A.~Banerjee, N.~Heintze, and J.~G. Riecke.
\newblock A core calculus of dependency.
\newblock In \emph{Principles of Programming Languages}, pages 147--160. ACM,
  1999.
\newblock \doi{10.1145/292540.292555}.

\bibitem[Agten et~al.(2012)Agten, Strackx, Jacobs, and
  Piessens]{Agten2012SecComp}
P.~Agten, R.~Strackx, B.~Jacobs, and F.~Piessens.
\newblock Secure compilation to modern processors.
\newblock In \emph{Computer Security Foundations}, pages 171--185, 2012.

\bibitem[Ahmed and Blume(2008)]{Ahmed:2008:TCC:1411203.1411227}
A.~Ahmed and M.~Blume.
\newblock Typed closure conversion preserves observational equivalence.
\newblock In \emph{International Conference on Functional Programming}, pages
  157--168. ACM, 2008.
\newblock \doi{10.1145/1411204.1411227}.

\bibitem[Ahmed and Blume(2011)]{ahmedCPS}
A.~Ahmed and M.~Blume.
\newblock An equivalence-preserving {CPS} translation via multi-language
  semantics.
\newblock In \emph{International Conference on Functional Programming}, pages
  431--444. ACM, 2011.
\newblock \doi{10.1145/2034773.2034830}.

\bibitem[Aydemir et~al.(2005)Aydemir, Bohannon, Fairbairn, Foster, Pierce,
  Sewell, Vytiniotis, Washburn, Weirich, and Zdancewic]{Aydemir2005}
B.~E. Aydemir, A.~Bohannon, M.~Fairbairn, J.~N. Foster, B.~C. Pierce,
  P.~Sewell, D.~Vytiniotis, G.~Washburn, S.~Weirich, and S.~Zdancewic.
\newblock Mechanized metatheory for the masses: The poplmark challenge.
\newblock In \emph{Theorem Proving in Higher Order Logics}, pages 50--65.
  Springer Berlin Heidelberg, 2005.
\newblock \doi{10.1007/11541868_4}.

\bibitem[Benton and Hur(2009)]{bistcc}
N.~Benton and C.-K. Hur.
\newblock Biorthogonality, step-indexing and compiler correctness.
\newblock In \emph{International Conference on Functional Programming},
  volume~44, pages 97--108. ACM, 2009.
\newblock \doi{10.1145/1596550.1596567}.

\bibitem[Benton and Hur(2010)]{realizability}
N.~Benton and C.-K. Hur.
\newblock Realizability and compositional compiler correctness for a
  polymorphic language.
\newblock Technical report, {MSR}, 2010.

\bibitem[Bowman and Ahmed(2015)]{nonintfree}
W.~J. Bowman and A.~Ahmed.
\newblock Noninterference for free.
\newblock In \emph{International Conference on Functional Programming}. ACM,
  2015.

\bibitem[Curien(2007)]{definFA}
P.-L. Curien.
\newblock Definability and full abstraction.
\newblock \emph{Electron. Notes Theor. Comput. Sci.}, 172:\penalty0 301--310,
  2007.
\newblock ISSN 1571-0661.
\newblock \doi{10.1016/j.entcs.2007.02.011}.

\bibitem[Devriese et~al.(2016)Devriese, Patrignani, and
  Piessens]{Devriese2016FullyAbsApprox}
D.~Devriese, M.~Patrignani, and F.~Piessens.
\newblock Fully abstract compilation by approximate back-translation.
\newblock In \emph{Principles of Programming Languages}. ACM, 2016.

\bibitem[Devriese et~al.(2017)Devriese, Patrignani, and
  Piessens]{Devriese2017ModularFullyAbsApproxTR}
D.~Devriese, M.~Patrignani, and F.~Piessens.
\newblock Modular fully abstract compilation by approximate back-translation:
  Technical appendix.
\newblock Technical Report {CW} 702, Dept. of Computer Science, KU Leuven,
  2017.

\bibitem[Dreyer et~al.(2010)Dreyer, Neis, and
  Birkedal]{Dreyer:2010:IHS:1863543.1863566}
D.~Dreyer, G.~Neis, and L.~Birkedal.
\newblock The impact of higher-order state and control effects on local
  relational reasoning.
\newblock In \emph{International Conference on Functional Programming}, pages
  143--156, 2010.
\newblock \doi{10.1145/1863543.1863566}.

\bibitem[Fournet et~al.(2013)Fournet, Swamy, Chen, Dagand, Strub, and
  Livshits]{fstar2js}
C.~Fournet, N.~Swamy, J.~Chen, P.-E. Dagand, P.-Y. Strub, and B.~Livshits.
\newblock Fully abstract compilation to {JavaScript}.
\newblock In \emph{Principles of Programming Languages}, pages 371--384. ACM,
  2013.
\newblock \doi{10.1145/2429069.2429114}.

\bibitem[Gorla and Nestman(2014)]{faEHM}
D.~Gorla and U.~Nestman.
\newblock Full abstraction for expressiveness: History, myths and facts.
\newblock \emph{Math. Struct. Comp. Science}, 2014.

\bibitem[Hur and Dreyer(2011)]{Hur:2011:KLR:1926385.1926402}
C.-K. Hur and D.~Dreyer.
\newblock A {Kripke} logical relation between {ML} and assembly.
\newblock In \emph{Principles of Programming Languages}, pages 133--146. ACM,
  2011.
\newblock \doi{10.1145/1926385.1926402}.

\bibitem[Hur et~al.(2012)Hur, Dreyer, Neis, and Vafeiadis]{marriage}
C.-K. Hur, D.~Dreyer, G.~Neis, and V.~Vafeiadis.
\newblock The marriage of bisimulations and {Kripke} logical relations.
\newblock In \emph{Principles of Programming Languages}, pages 59--72. ACM,
  2012.
\newblock \doi{10.1145/2103656.2103666}.

\bibitem[Jagadeesan et~al.(2011)Jagadeesan, Pitcher, Rathke, and
  Riely]{Jagadeesan:2011:LMV:2056311.2056556}
R.~Jagadeesan, C.~Pitcher, J.~Rathke, and J.~Riely.
\newblock Local memory via layout randomization.
\newblock In \emph{Computer Security Foundations Symposium}, pages 161--174.
  IEEE Computer Society, 2011.
\newblock \doi{10.1109/CSF.2011.18}.

\bibitem[Juglaret et~al.(2016)Juglaret, Hri\c{t}cu, {Azevedo de Amorim}, and
  Pierce]{catalin}
Y.~Juglaret, C.~Hri\c{t}cu, A.~{Azevedo de Amorim}, and B.~C. Pierce.
\newblock Beyond good and evil: Formalizing the security guarantees of
  compartmentalizing compilation.
\newblock In \emph{CSF}. IEEE Computer Society Press, July 2016.
\newblock URL \url{https://arxiv.org/abs/1602.04503}.

\bibitem[Kennedy(2006)]{Kennedy}
A.~Kennedy.
\newblock Securing the .{NET} programming model.
\newblock \emph{Theor. Comput. Sci.}, 364\penalty0 (3):\penalty0 311--317, Nov.
  2006.
\newblock ISSN 0304-3975.
\newblock \doi{10.1016/j.tcs.2006.08.014}.

\bibitem[Keuchel et~al.(2016)Keuchel, Weirich, and Schrijvers]{Keuchel2016}
S.~Keuchel, S.~Weirich, and T.~Schrijvers.
\newblock Needle {\&} knot: Binder boilerplate tied up.
\newblock In \emph{European Symposium on Programming}, pages 419--445. Springer
  Berlin Heidelberg, 2016.
\newblock \doi{10.1007/978-3-662-49498-1_17}.

\bibitem[Larmuseau et~al.(2015)Larmuseau, Patrignani, and Clarke]{adriaanaplas}
A.~Larmuseau, M.~Patrignani, and D.~Clarke.
\newblock A secure compiler for {ML} modules.
\newblock In \emph{Programming Languages and Systems - 13th Asian Symposium,
  {APLAS} 2015, Pohang, South Korea, November 30 - December 2, 2015,
  Proceedings}, pages 29--48, 2015.
\newblock \doi{10.1007/978-3-319-26529-2_3}.
\newblock URL \url{http://dx.doi.org/10.1007/978-3-319-26529-2_3}.

\bibitem[MacQueen et~al.(1984)MacQueen, Plotkin, and
  Sethi]{MacQueen:1984:IMR:800017.800528}
D.~MacQueen, G.~Plotkin, and R.~Sethi.
\newblock An ideal model for recursive polymorphic types.
\newblock In \emph{Principles of Programming Languages}, pages 165--174. ACM,
  1984.
\newblock \doi{10.1145/800017.800528}.

\bibitem[Matthews and Ahmed(2008)]{Matthews2008ParaPoly}
J.~Matthews and A.~Ahmed.
\newblock Parametric polymorphism through run-time sealing or, theorems for
  low, low prices!
\newblock In \emph{Programming Languages and Systems}, volume 4960 of
  \emph{LNCS}, pages 16--31. Springer Berlin Heidelberg, 2008.
\newblock \doi{10.1007/978-3-540-78739-6_2}.

\bibitem[Matthews and Findler(2009)]{Matthews:MLS}
J.~Matthews and R.~B. Findler.
\newblock Operational semantics for multi-language programs.
\newblock \emph{ACM Transactions on Programming Languages and Systems},
  31:\penalty0 12:1--12:44, Apr. 2009.
\newblock ISSN 0164-0925.
\newblock \doi{10.1145/1498926.1498930}.

\bibitem[McCusker(1996)]{McCusker1996FullAbs}
G.~McCusker.
\newblock Full abstraction by translation.
\newblock \emph{Advances in Theory and Formal Methods of Computing}, 1996.

\bibitem[Mitchell(1993)]{Mitchell1993141}
J.~C. Mitchell.
\newblock On abstraction and the expressive power of programming languages.
\newblock \emph{Science of Computer Programming}, 21\penalty0 (2):\penalty0 141
  -- 163, 1993.
\newblock ISSN 0167-6423.
\newblock \doi{10.1016/0167-6423(93)90004-9}.

\bibitem[Morrisett et~al.(1999)Morrisett, Crary, Glew, Grossman, Samuels,
  Smith, Walker, Weirich, and Zdancewic]{tal}
G.~Morrisett, K.~Crary, N.~Glew, D.~Grossman, R.~Samuels, F.~Smith, D.~Walker,
  S.~Weirich, and S.~Zdancewic.
\newblock {TALx86}: A realistic typed assembly language.
\newblock In \emph{Second Workshop on Compiler Support for System Software},
  pages 25--35, 1999.

\bibitem[Neis et~al.(2009)Neis, Dreyer, and
  Rossberg]{Neis:2009:NP:1596550.1596572}
G.~Neis, D.~Dreyer, and A.~Rossberg.
\newblock Non-parametric parametricity.
\newblock In \emph{International Conference on Functional Programming}, pages
  135--148. ACM, 2009.
\newblock \doi{10.1145/1596550.1596572}.

\bibitem[Neis et~al.(2015)Neis, Hur, Kaiser, McLaughlin, Dreyer, and
  Vafeiadis]{pils}
G.~Neis, C.-K. Hur, J.-O. Kaiser, C.~McLaughlin, D.~Dreyer, and V.~Vafeiadis.
\newblock {Pilsner}: A compositionally verified compiler for a higher-order
  imperative language.
\newblock In \emph{International Conference on Functional Programming}. ACM,
  2015.

\bibitem[New et~al.(2016)New, Bowman, and Ahmed]{New:2016:FAC:2951913.2951941}
M.~S. New, W.~J. Bowman, and A.~Ahmed.
\newblock Fully abstract compilation via universal embedding.
\newblock In \emph{International Conference on Functional Programming}, pages
  103--116. ACM, 2016.
\newblock \doi{10.1145/2951913.2951941}.

\bibitem[Parrow(2014)]{gcFA}
J.~Parrow.
\newblock General conditions for full abstraction.
\newblock \emph{Math. Struct. Comp. Science}, 2014.

\bibitem[Patrignani et~al.(2015)Patrignani, Agten, Strackx, Jacobs, Clarke, and
  Piessens]{scoo-j}
M.~Patrignani, P.~Agten, R.~Strackx, B.~Jacobs, D.~Clarke, and F.~Piessens.
\newblock Secure compilation to protected module architectures.
\newblock \emph{ACM Trans. Program. Lang. Syst.}, 37\penalty0 (2):\penalty0
  6:1--6:50, Apr. 2015.
\newblock ISSN 0164-0925.
\newblock \doi{10.1145/2699503}.

\bibitem[Patrignani et~al.(2016)Patrignani, Devriese, and Piessens]{mfac}
M.~Patrignani, D.~Devriese, and F.~Piessens.
\newblock {On Modular and Fully Abstract Compilation}.
\newblock In \emph{CSF 2016}, 2016.

\bibitem[Perconti and Ahmed(2014)]{perconti}
J.~T. Perconti and A.~Ahmed.
\newblock Verifying an open compiler using multi-language semantics.
\newblock In \emph{ESOP}, volume 8410 of \emph{Lecture Notes in Computer
  Science}, pages 128--148, 2014.

\bibitem[Pierce(2002)]{pierce2002types}
B.~C. Pierce.
\newblock \emph{Types and programming languages}.
\newblock MIT press, 2002.

\bibitem[Plotkin(1977)]{lcfConsidered}
G.~D. Plotkin.
\newblock {LCF} considered as a programming language.
\newblock \emph{Theoretical Computer Science}, 5:\penalty0 223--255, 1977.
\newblock \doi{10.1016/0304-3975(77)90044-5}.

\bibitem[Riecke(1993)]{Riecke1993:MSC:4117452}
J.~G. Riecke.
\newblock Fully abstract translations between functional languages.
\newblock \emph{Mathematical Structures in Computer Science}, 3:\penalty0
  387--415, 12 1993.
\newblock ISSN 1469-8072.
\newblock \doi{10.1017/S0960129500000293}.

\bibitem[Ritter and Pitts(1995)]{Ritter1995FullyAbs}
E.~Ritter and A.~M. Pitts.
\newblock {A fully abstract translation between a $\lambda$-calculus with
  reference types and Standard ML}.
\newblock In M.~Dezani-Ciancaglini and G.~Plotkin, editors, \emph{Typed Lambda
  Calculi and Applications}, volume 902 of \emph{{LNCS}}, pages 397--413.
  Springer Berlin Heidelberg, 1995.
\newblock ISBN 978-3-540-59048-4.
\newblock \doi{10.1007/BFb0014067}.

\bibitem[Schmidt-Schau{\ss} et~al.(2015)Schmidt-Schau{\ss}, Sabel, Niehren, and
  Schwinghammer]{SchmidtSchauss201598}
M.~Schmidt-Schau{\ss}, D.~Sabel, J.~Niehren, and J.~Schwinghammer.
\newblock Observational program calculi and the correctness of translations.
\newblock \emph{Theoretical Computer Science}, 577:\penalty0 98 -- 124, 2015.
\newblock ISSN 0304-3975.
\newblock \doi{http://dx.doi.org/10.1016/j.tcs.2015.02.027}.

\bibitem[Scott(1976)]{Scott1976Domains}
D.~Scott.
\newblock Data types as lattices.
\newblock \emph{SIAM Journal on Computing}, 5\penalty0 (3):\penalty0 522--587,
  1976.
\newblock \doi{10.1137/0205037}.

\bibitem[Shikuma and Igarashi(2007)]{fixTse}
N.~Shikuma and A.~Igarashi.
\newblock Proving noninterference by a fully complete translation to the simply
  typed $\lambda$-calculus.
\newblock In M.~Okada and I.~Satoh, editors, \emph{Advances in Computer Science
  - ASIAN 2006. Secure Software and Related Issues}, volume 4435 of
  \emph{LNCS}, pages 301--315. Springer Berlin Heidelberg, 2007.
\newblock \doi{10.1007/978-3-540-77505-8_24}.

\bibitem[Smith(1998)]{Smith1998coverage}
S.~F. Smith.
\newblock The coverage of operational semantics.
\newblock In \emph{Higher Order Operational Techniques in Semantics},
  Publications of the Newton Institute, pages 307--346. Cambridge University
  Press, 1998.

\bibitem[Sumii and Pierce(2007)]{seal}
E.~Sumii and B.~C. Pierce.
\newblock A bisimulation for dynamic sealing.
\newblock \emph{Theor. Comput. Sci.}, 375\penalty0 (1-3):\penalty0 169--192,
  Apr. 2007.
\newblock ISSN 0304-3975.
\newblock \doi{10.1016/j.tcs.2006.12.032}.

\bibitem[Tse and Zdancewic(2004)]{depIntoPar}
S.~Tse and S.~Zdancewic.
\newblock Translating dependency into parametricity.
\newblock In \emph{International Conference on Functional Programming}, pages
  115--125. ACM, 2004.
\newblock \doi{10.1145/1016850.1016868}.

\end{thebibliography}

\end{document}